\documentclass[a4paper,11pt]{article}
\usepackage[english]{babel}
\usepackage[utf8]{inputenc}
\usepackage[T1]{fontenc}
\usepackage{lmodern}
\usepackage{amssymb,amsmath,amsfonts}
\usepackage{xspace}
\usepackage{graphicx}
\usepackage{hyperref}
\usepackage{url}
\usepackage{stmaryrd}
\usepackage[babel=true]{csquotes}
\usepackage{amsthm}
\usepackage{subfigure}
\usepackage{mathrsfs}
\usepackage{float}
\usepackage[inner=2.5cm,outer=3.5cm]{geometry}

\theoremstyle{plain}
\newtheorem{thm}{Theorem}[section]
\newtheorem{lm}[thm]{Lemma}
\newtheorem{prop}[thm]{Proposition}
\newtheorem{cor}[thm]{Corollary}
\newtheorem{conj}{Conjecture}
\newtheorem*{thmintro}{Theorem}

\theoremstyle{definition}
\newtheorem{dfn}[thm]{Definition}
\newtheorem{rmk}[thm]{Remark}
\newtheorem{ex}[thm]{Example}

\def\dpar#1#2{\frac{\partial #1}{\partial #2}}

\def\Hil{\mathcal{H}\xspace}

\def\R{\mathbb{R}\xspace}
\def\C{\mathbb{C}\xspace}

\def\P{\mathbb{P}\xspace}
\def\S{\mathbb{S}\xspace}

\def\scal#1#2{\left\langle #1,#2\right\rangle}

\def\classe#1#2{\mathcal{C}^{#1}#2}

\DeclareMathOperator{\Tr}{Tr} 

\def\bigO#1{O\mathopen{}\left(#1\right)\mathclose{}}

\def\derLie{\mathcal{L}\xspace}

\newcommand\blfootnote[1]{
  \begingroup
  \renewcommand\thefootnote{}\footnote{#1}
  \addtocounter{footnote}{-1}
  \endgroup
}

\usepackage{xcolor}

\hypersetup{pdfborder={0000}, colorlinks=true, linkcolor=blue, citecolor=red}

\begin{document}

\title{Bounds for fidelity of semiclassical Lagrangian states in K\"ahler quantization}

\author{Yohann Le Floch}

\maketitle

\begin{abstract}
We define mixed states associated with submanifolds with probability densities in quantizable closed K\"ahler manifolds. Then, we address the problem of comparing two such states via their fidelity. Firstly, we estimate the sub-fidelity and super-fidelity of two such states, giving lower and upper bounds for their fidelity, when the underlying submanifolds are two Lagrangian submanifolds intersecting transversally at a finite number of points, in the semiclassical limit. Secondly, we investigate a family of examples on the sphere, for which we manage to obtain a better upper bound for the fidelity. We conclude by stating a conjecture regarding the fidelity in the general case.
\end{abstract}

\vspace{2cm}

\noindent
{\bf Yohann Le Floch} \\
Institut de Recherche Math\'ematique avanc\'ee,\\
UMR 7501, Universit\'e de Strasbourg et CNRS,\\
7 rue Ren\'e Descartes,\\
67000 Strasbourg, France.\\
{\em E-mail:} \texttt{ylefloch@unistra.fr}\\



\section{Introduction}

\subsection{States in geometric quantization}

Let $(M,\omega,j)$ be a closed, connected K\"ahler manifold, equipped with a prequantum line bundle  $(L,\nabla)$. According to the geometric quantization procedure, due to Kostant and Souriau \cite{Kos,Sou}, we define, for any integer $k \geq 1$, the quantum state space as the Hilbert space $\Hil_k = H^{0}(M,L^{\otimes k})$ of holomorphic sections of $L^{\otimes k} \to M$\footnote{In the rest of the paper, we will write $L^k$ instead of $L^{\otimes k}$ to simplify notation.}; the semiclassical limit is $k \to +\infty$. The quantum observables are Berezin-Toeplitz operators, introduced by Berezin \cite{Ber}, whose microlocal analysis has been initiated by Boutet de Monvel and Guillemin \cite{BouGui}, and which have been studied by many authors during the last years (see for instance \cite{ChaBTO,MaMa,Schli} and references therein).

In this paper, we investigate the problem of quantizing a given submanifold $\Sigma$ of $M$, that is constructing a state concentrating on $\Sigma$ in the semiclassical limit (in a sense that we will precise later). This kind of construction has been achieved for a so-called \emph{Bohr-Sommerfeld} Lagrangian submanifold $\Sigma$, that is Lagrangian manifold with trivial holonomy with respect to the connection induced by $\nabla$ on $L^k$ (\cite{BorPauUri}, see also \cite{ChaBS}). The state obtained in this case is a pure state whose microsupport is contained in $\Sigma$. Such states are useful, for instance, to construct quasimodes for Berezin-Toeplitz operators.

Here we adopt a different point of view. We assume that $\Sigma$ is any submanifold, equipped with a smooth density $\sigma$ such that $\int_{\Sigma} \sigma = 1$. Then we construct a mixed state--or rather its density operator--$\rho_k(\Sigma,\sigma)$ associated with this data, by integrating the coherent states projectors along $\Sigma$ with respect to $\sigma$, see Definition \ref{dfn:state}. We prove that this state cannot be pure, and that it concentrates on $\Sigma$ in the semiclassical limit. Similar states, the so-called P-representable or classical quantum states, have been considered in the physics literature \cite{GirBrBr} and have been used recently to explore the links between symplectic displaceability and quantum dislocation \cite{ChaPol}; they are obtained by integrating the coherent projectors along $M$ against a Borel probability measure.

\subsection{Main results}

Given two submanifolds with probability densities $(\Sigma_1, \sigma_1)$ and $(\Sigma_2, \sigma_2)$, we would like to compare the two associated states $\rho_{k,1} = \rho_k(\Sigma_1,\sigma_1)$ and $\rho_{k,2} = \rho_k(\Sigma_2,\sigma_2)$. For the purpose of comparing two mixed states, one often uses the fidelity function \cite{Uhl,Jos}, defined as
\[ F\left(\rho_{k,1},\rho_{k,2}\right) = \Tr\left( \sqrt{\sqrt{\rho_{k,1}} \ \rho_{k,2} \ \sqrt{\rho_{k,1}} } \right)^2 \in [0,1].\]
Because it involves the square roots of the density operators, it is quite complicated to estimate in general. Nevertheless, Miszczak \emph{et al.} \cite{Mis} recently obtained lower and upper bounds for the fidelity function; they introduced two quantities $E(\rho_{k,1},\rho_{k,2})$ and $G(\rho_{k,1},\rho_{k,2})$, respectively called \emph{sub-fidelity} and \emph{super-fidelity}, easier to study, such that $E(\rho_{k,1},\rho_{k,2}) \leq F(\rho_{k,1},\rho_{k,2}) \leq G(\rho_{k,1},\rho_{k,2})$.

We will estimate these quantities, in the semiclassical limit, in the particular case where $\Sigma_1 = \Gamma_1$ and $\Sigma_2 = \Gamma_2$ are two Lagrangian submanifolds intersecting transversally at a finite number of points $m_1, \ldots, m_s$. Our main results can be summarized as follows.

\begin{thmintro} 
There exists some constants $C_i((\Gamma_1,\sigma_1),(\Gamma_2,\sigma_2)) > 0$, $i=1,2$,  depending on the geometry near the intersection points, such that the sub-fidelity satisfies
\[ E(\rho_{k,1},\rho_{k,2}) = \left(\frac{2\pi}{k}\right)^{n}  C_1((\Gamma_1,\sigma_1),(\Gamma_2,\sigma_2)) + \bigO{k^{-(n+1)}} \]
and the super-fidelity satisfies
\[ G(\rho_{k,1},\rho_{k,2}) = 1 - \left(\frac{2\pi}{k}\right)^{\frac{n}{2}}  C_2((\Gamma_1,\sigma_1),(\Gamma_2,\sigma_2)) +  \bigO{k^{-\min\left(n,\frac{n}{2} + 1\right)}}. \]
\end{thmintro}

For instance, the constant in the sub-fidelity involves the principal angles between the two tangent spaces at the intersection points. We refer the reader to Theorems \ref{thm:subfid} and \ref{thm:superfid} for precise statements and explicit expressions for the constants involved in these estimates. Unfortunately, this result does not allow us to obtain an equivalent for the fidelity function when $k$ goes to infinity, as \emph{a priori} this fidelity could display any behaviour between these two ranges $\bigO{k^{-n}}$ and $\bigO{1}$. However, we will study a family of examples on the two-sphere, for which we prove that the fidelity is a $\bigO{k^{-1 + \varepsilon}}$ for every sufficiently small $\varepsilon > 0$ (Theorem \ref{thm:fid_sphere}); this result is non trivial and requires care and a fine analysis of the interactions near intersection points. We also perform some numerical computations regarding these examples.

\begin{rmk} We believe that our results extend without effort to the case where the quantum state space is the space of holomorphic sections of $L^k \otimes K \to M$ where $K$ is an auxiliary Hermitian holomorphic line bundle, for instance in the case where $K = \delta$ is a half-form bundle (which corresponds to the so-called metaplectic correction). These results should also extend to the case of the quantization of a closed symplectic but non necessarily K\"ahler manifold, using for instance the recipe introduced in \cite{Cha_symp}; the main ingredient, namely the decription of the asymptotics of the Bergman kernel, is still available, only more complicated to describe. We do not treat any of these two cases here for the sake of clarity.
\end{rmk}

\subsection{Structure of the article}

The first half of this manuscript is devoted to the definition of the state associated with a submanifold with density and the computation of the sub-fidelity and super-fidelity of such states in the Lagrangian case, in all generality. In Section \ref{sect:prelim}, we discuss the setting and introduce the notions and notation that will be needed to achieve this goal. In Section \ref{sect:def_states}, we explain how to obtain a state from a submanifold with density, and we study the first properties of such states. In particular, we compute their purity to show that they are always mixed for $k$ large enough. We prove our estimates for the sub-fidelity and the super-fidelity of two states associated with Lagrangian submanifolds intersecting transversally at a finite number of points in Section \ref{sect:sub_super}. 

The second half of the paper, corresponding to Sections \ref{sect:examples} and \ref{sect:numerics}, focuses on a family of examples on $\S^2$. A remarkable fact is that one can obtain much better estimates for the fidelity function itself, employing non trivial methods, that can however not be used as they are to study the general case, although some parts of the analysis may be useful to attack the latter. 

\section{Preliminaries and notation}
\label{sect:prelim}

\subsection{The setting: K{\"a}hler quantization}

Throughout the paper, $(M,\omega,j)$ will be a closed, connected K\"ahler manifold, of real dimension $\dim M = 2n$, such that the cohomology class of $(2\pi)^{-1} \omega$ is integral, and $(L,\nabla)$ will be a prequantum line bundle over $M$, that is a Hermitian holomorphic line bundle $L \to M$ whose Chern connection $\nabla$ has curvature $-i\omega$. Let $\mu_M = |\omega^n|/n!$ be the Liouville measure on $M$. For $k \geq 1$ integer, let $h_k$ be the Hermitian form induced on $L^k$, and consider the Hilbert space of holomorphic sections of $L^k \to M$:
\[ \Hil_k = H^0(M,L^k), \qquad \scal{\psi}{\phi}_k = \int_M h_k(\psi,\phi) \mu_M.  \]
Since $M$ is compact, $\Hil_k$ is finite-dimensional; more precisely, it is standard that
\begin{equation} \dim \Hil_k = \left( \frac{k}{2\pi} \right)^n \mathrm{vol}(M) + \bigO{k^{n-1}}. \label{eq:asymp_dim}  \end{equation}
Let $L^2(M,L^k)$ be the completion of $\classe{\infty}{(M,L^k)}$ with respect to $\scal{\cdot}{\cdot}_k$, and let $\Pi_k: L^2(M,L^k) \to \Hil_k$ be the orthogonal projector from $L^2(M,L^k)$ to the space of holomorphic sections of $L^k \to M$. The Berezin-Toeplitz operator associated with $f \in \classe{\infty}(M)$ is
\begin{equation} T_k(f) = \Pi_k f: \Hil_k \to \Hil_k, \label{eq:def_BTO} \end{equation}
where $f$ stands for the operator of multiplication by $f$. More generally, a Berezin-Toeplitz operator is any sequence of operators $(T_k:\Hil_k \to \Hil_k)_{k \geq 1}$ of the form $T_k = \Pi_k f(\cdot,k) + R_k$ where $f(\cdot,k)$ is a sequence of smooth functions with an asymptotic expansion of the form $ f(\cdot,k) = \sum_{\ell \geq 0} k^{-\ell} f_{\ell}$ for the $\classe{\infty}{}$ topology, and $\| R_k \| = \bigO{k^{-N}}$ for every $N \geq 1$. 

Let $p_1,p_2: M \times M \to M$ be the natural projections on the left and right factor. If $U \to M$, $V \to M$ are two line bundles over $M$, we define the line bundle (sometimes called external tensor product) $U \boxtimes V = p_1^*U \otimes p_2^*V \to M \times M$. The Schwartz kernel of an operator $S_k: \Hil_k \to \Hil_k$ is the unique section $S_k(\cdot,\cdot)$ of $L^k \boxtimes \bar{L}^k \to M \times M$ such that for every $\varphi \in \Hil_k$ and every $x \in M$,
\[ (S_k \varphi)(x) = \int_M  S_k(x,y) \cdot \varphi_k(y) \ d\mu_M(y), \]
where the dot corresponds to contraction with respect to $h_k$: for $\bar{u} \in \bar{L}^k_y$ and $v \in L^k_y$, $\bar{u} \cdot v = (h_k)_y(v,u)$. In particular, the Schwartz kernel of $\Pi_k$ is called the \emph{Bergman kernel}. 

In this context, Charles \cite{ChaBTO} has obtained, relying on \cite{BouGui}, a very precise description of the Bergman kernel in the semiclassical limit. For our purpose, we will only need part of it, namely that
\begin{equation} \Pi_k(x,y) = \left( \frac{k}{2 \pi} \right)^n S^k(x,y) \left( a_0(x,y) + \bigO{k^{-1}} \right) \label{eq:asympt_projector} \end{equation}
where $S \in \classe{\infty}{(M^2,L \boxtimes \overline{L})}$ satisfies $S(x,x) = 1$ and $|S(x,y)| <1$ whenever $x \neq y$ (among other properties, see \cite[Proposition 1]{ChaBTO}), $a_0 \in \classe{\infty}{(M^2, \R)}$ is such that $a_0(x,x) = 1$ and the remainder $\bigO{k^{-1}}$ is uniform in $(x,y) \in M^2$. Here $| \cdot |$ denotes the norm induced by $h$ on $L \boxtimes \overline{L}$, and for $x \in M$, we use $h_k$ to identify $L_x \otimes \bar{L}_x$ with $\C$.

\subsection{Generalities about fidelity}

As already explained, one useful tool to compare two states is the fidelity function, see for instance \cite{Uhl, Jos} or \cite[Chapter 9]{ChN}. Recall that the trace norm of a trace class operator $A$ acting on a Hilbert space $\Hil$ is $ \|A\|_{\mathrm{\Tr}} = \Tr(\sqrt{A^*A})$. Given two states $\rho, \eta$ on $\Hil$, that is positive semidefinite Hermitian operators on $\Hil$ of trace one, their fidelity is defined as\footnote{Note that some authors call fidelity the square root of this function, however we prefer to keep the square in order to simplify some of the computations.} 
\[ F(\rho,\eta) = \left\| \sqrt{\rho} \sqrt{\eta} \right\|^2_{\Tr} = \Tr\left( \sqrt{\sqrt{\rho} \ \eta \ \sqrt{\rho} } \right)^2. \]
Even though it is not obvious from this formula, fidelity is symmetric in its arguments. It measures how close the two states are in the following sense; $F(\rho,\eta)$ is a number comprised between $0$ and $1$, and $F(\rho,\eta) = 1$ if and only if $\rho = \eta$, while $F(\rho,\eta) = 0$ if and only if $\rho(\Hil)$ and $\eta(\Hil)$ are orthogonal. In the particular case where both states are pure, i.e. $\rho$ (respectively $\eta$) is the orthogonal projection on the line spanned by $\phi \in \Hil$ (respectively $\psi \in \Hil$), where $\phi$ and $\psi$ are unit vectors, one readily checks that $F(\phi,\psi) = |\scal{\phi}{\psi}|^2$. The fidelity function is interesting for further reasons, such as its invariance under conjugation of both arguments by a common unitary operator, its multiplicativity with respect to tensor products, or its joint concavity. It is, however, very hard to compute in general because it involves square roots of operators.

Consequently, some efforts have been made to give bounds for the fidelity function that would be more easily computable. The following remarkable bounds on the fidelity of states $\rho,\eta$ acting on a finite-dimensional Hilbert space have been obtained in \cite{Mis}: $ E(\rho,\eta) \leq F(\rho,\eta) \leq G(\rho,\eta)$ where the function $E$, called sub-fidelity, is defined as 
\begin{equation} E(\rho,\eta) = \Tr(\rho \eta) + \sqrt{2} \sqrt{\Tr(\rho \eta)^2 - \Tr((\rho \eta)^2)} \label{eq:sub_fid}\end{equation}
and the function $G$, called super-fidelity, is defined as 
\begin{equation} G(\rho,\eta) = \Tr(\rho \eta) + \sqrt{\left( 1 - \Tr(\rho^2) \right)\left( 1 - \Tr(\eta^2) \right)} \label{eq:super_fid}\end{equation}
It turns out that these two quantities keep some of the interesting properties of fidelity, and can be measured using physical experiments; furthermore they both coincide with fidelity when both states are pure. From a mathematical point of view, these quantities seem much more tractable than the fidelity function because they involve only traces of products and powers of operators.

\subsection{Principal angles}

The notion of principal angles (see for example \cite[Section $12.4.3$]{Gol}) will play a crucial part in our estimates. Let $V$ be a real vector space, endowed with an inner product $(\cdot | \cdot)$, and let $E, F$ be two subspaces of $V$ such that $\alpha = \dim E \geq \beta = \dim F \geq 1$. 
 
\begin{dfn} 
The \emph{principal angles} $0 \leq \theta_1 \leq \ldots \leq \theta_{\beta} \leq \frac{\pi}{2}$ between $E$ and $F$ are defined recursively by the formula $ \cos(\theta_{\ell}) = (u_{\ell}|v_{\ell}) := \max_{W_{\ell}} (u|v)$, where 
\[ W_{\ell} = \left\{ (u,v) \in E \times F \ | \ \|u\| = 1 = \|v\|, \quad \forall m \in \llbracket 1,\ell \rrbracket, \ (u|u_m) = 0 = (v|v_m) \right\}. \]
\end{dfn}

Note that $\theta_1 = 0$ if and only if $E \cap F \neq \{0\}$. We will need the two following properties of principal angles; the first one appears in the computation of $\Tr(\rho_{k,1} \rho_{k,2})$ (Theorem \ref{thm:trace}).

\begin{lm}
\label{lm:angle}
Let $V$ be a real vector space of dimension $2n$, $n \geq 1$, endowed with an inner product $(\cdot | \cdot)$, and let $E,F$ be two subspaces of $V$ of dimension $n$. Let $(e_{p})_{1 \leq p \leq n}$ (respectively $(f_{q})_{1 \leq q \leq n}$)  be any orthonormal basis of $E$ (respectively $F$). We introduce the $n \times n$ matrix $G$ with entries $G_{p,q} = (e_{p}|f_q)$; then the quantity $\det(I_n - G^{\top}G)$ does not depend on the choice of $(e_{p})_{1 \leq p \leq n}$ and $(f_{q})_{1 \leq q \leq n}$. Moreover, it satisfies
\[ \det\left(I_n - G^{\top}G\right) = \prod_{\ell = 1}^n \sin^2(\theta_{\ell}) \]
where $0 \leq \theta_1 \leq \ldots \leq \theta_n \leq \frac{\pi}{2}$ are the principal angles between $E$ and $F$.
\end{lm}

\begin{proof}
Let $(\tilde{e}_{p})_{1 \leq p \leq n}$ be another orthonormal basis of $E$, and let $O = (O_{p,q})_{1 \leq p,q \leq n}$ be the matrix such that
\[ \forall p \in \llbracket 1, n \rrbracket, \qquad \tilde{e}_{p} = \sum_{r=1}^n O_{p,r} e_r. \] 
Let $\tilde{G}$ be the matrix with entries $\tilde{G}_{p,q} = (\tilde{e}_{p}|f_q)$; then $\tilde{G} = O G$ and
\[ \det(I_n - \tilde{G}^{\top} \tilde{G}) = \det(I_n - \tilde{G}^{\top} \tilde{O}^{\top} \tilde{O} \tilde{G}) = \det(I_n - G^{\top}G) \]
since $O$ is orthogonal. Now, observe that
\[ (G^{\top}G)_{p,q} = \sum_{r=1}^n (e_r|f_p)(e_r|f_q) = \left( \sum_{r=1}^n (e_r|f_p)e_r \Big|f_q \right) = (P f_p|f_q) \]
where $P$ is the orthogonal projector from $V$ to $E$. Consequently, $(I_n - G^{\top} G)_{p,q} = (Q f_p|f_q)$ with $Q$ the orthogonal projector from $V$ to $E^{\perp}$. Thus, if $(e_{n+1}, \ldots, e_{2n})$ is any orthonormal basis of $E^{\perp}$, then
\[ (I_n - G^{\top} G)_{p,q} = \sum_{r=1}^n (f_p|e_{n+r})(e_{n+r}|f_q); \]
this means that $I_n - G^{\top} G = A^{\top} A$ where $A$ is the matrix with entries given by $A_{p,q} = (e_{n+p}|f_q)$. But it is known that the eigenvalues of $A^{\top} A$ are $\cos^2(\zeta_1), \ldots, \cos^2(\zeta_n)$, where $\zeta_1 \leq \ldots \leq \zeta_{n}$ are the principal angles between $E^{\perp}$ and $F$, see for instance \cite{Sho}. Consequently, $\det( A^{\top} A) = \prod_{\ell = 1}^n \cos^2(\zeta_{\ell})$ and the result follows from the fact that for every $\ell \in \llbracket 1,n \rrbracket$, $\zeta_{\ell} = \tfrac{\pi}{2} - \theta_{n - \ell}$ \cite[Property $2.1$]{Zhu}.
\end{proof}

The second property will be used in the proof of Theorem \ref{thm:trace_square}.

\begin{lm}
\label{lm:det_symp}
Let $(V,\omega)$ be a real symplectic vector space of dimension $2n$, $n \geq 1$, endowed with a complex structure $J:V \to V$ which is compatible with $\omega$, and let $(\cdot|\cdot) = \omega(\cdot,J\cdot)$ be the associated inner product. Let $E,F$ be two complementary Lagrangian subspaces of $V$, and let $(e_{p})_{1 \leq p \leq n}$ (respectively $(f_{p})_{1 \leq p \leq n}$)  be any orthonormal basis of $E$ (respectively $F$). Let $\Xi$ be the $n \times n$ matrix with entries $\Xi_{p,q} = \omega(e_p,f_q)$; then the quantity $ \det\left(I_n +  \Xi^{\top} \Xi\right)$ does not depend on the choice of $(e_{p})_{1 \leq p \leq n}$ and $(f_{p})_{1 \leq p \leq n}$. Moreover, it satisfies
\[ \det\left(I_n + \Xi^{\top} \Xi\right) = \prod_{\ell = 1}^n \left(1 + \sin^2(\theta_{\ell})\right), \]
where $0 \leq \theta_1 \leq \ldots \leq \theta_n \leq \frac{\pi}{2}$ are the principal angles between $E$ and $F$.
\end{lm}

\begin{proof}
The first statement is similar to the first statement of Lemma \ref{lm:angle}. Now, let $G$ be the $n \times n$ matrix defined in the latter, that is the matrix with entries $G_{p,q} = (e_p|f_q)$. A straightforward computation shows that $(\Xi^{\top} \Xi)_{p,q} = (Qf_p|f_q)$, where $Q$ is the orthogonal projection from $V$ to $J(E)$. Since $E$ is Lagrangian, $J(E) = E^{\perp}$, so the previous result means that $\Xi^{\top} \Xi = I_n - G^{\top} G$, which implies (see the proof of Lemma \ref{lm:angle}) that the eigenvalues of the matrix $\Xi^{\top} \Xi$ are $\sin^2(\theta_1), \ldots, \sin^2(\theta_n)$, which yields the result.

\end{proof}

\section{The state associated with a submanifold with density}
\label{sect:def_states}

\subsection{Definition}

We will define the state associated with a submanifold with density by means of coherent states; let us recall how those are constructed in the setting of geometric quantization (here we adopt the convention used in \cite[Section 5]{ChaBTO}). Let $P \subset L$ be the set of elements $u \in L$ such that $h(u,u) = 1$, and let $\pi: P \to M$ denote the natural projection. Given $u \in P$, for every $k \geq 1$, there exists a unique vector $\xi_k^u$ in $\Hil_k$ such that 
\[ \forall \phi \in \Hil_k, \qquad \phi(\pi(u)) = \scal{\phi}{\xi_k^u}_{k} u^k. \]
The vector $\xi_k^u \in \Hil_k$ is called the \emph{coherent vector} at $u$. 

By the properties of coherent states stated in \cite[Section 5]{ChaBTO} and the description of $\Pi_k$ given in Equation (\ref{eq:asympt_projector}), we have that for every $u \in P$,
\begin{equation} \|\xi_k^u\|^2_{k} = \left(\frac{k}{2\pi}\right)^n + \bigO{k^{n-1}} \label{eq:asymp_rawnsley} \end{equation}
when $k$ goes to infinity, and the remainder is uniform in $u \in P$. In particular, there exists $k_0 \geq 1$ such that $\xi_k^u \neq 0$ whenever $k \geq k_0$. For $k \geq k_0$, we set $\xi_{k}^{u,\mathrm{norm}} = \xi_k^u / \|\xi_k^u \|_k$ (and later on we will always implicitly assume that $k \geq k_0$ to simplify notation). This also means that the class of $\xi_k^u$ in the projective space $\P(\Hil_k)$ is well-defined; this class only depends on $\pi(u)$ and is called the \emph{coherent state} at $x = \pi(u)$. Furthermore, the projection
\[ P_k^x: \Hil_k \to \Hil_k, \qquad \phi \mapsto \scal{\phi}{\xi_k^{u,\mathrm{norm}}}_k \xi_k^{u,\mathrm{norm}}  \]
is also only dependent on $x$, and is called the \emph{coherent projector} at $x$. 

Now, let $\Sigma \subset M$ be a closed, connected submanifold of dimension $d \geq 1$, equipped with a positive density $\sigma$ (as defined in \cite[Chapter 3.3]{BerGos}) such that $\int_{\Sigma} \sigma = 1$. Then we can obtain a mixed state by superposition of the coherent projectors over the points of $\Sigma$. 
\begin{dfn}
\label{dfn:state}
We define the state associated with $(\Sigma,\sigma)$ as
\begin{equation} \rho_k(\Sigma,\sigma) = \int_{\Sigma} P_k^x \ \sigma(x) \label{eq:Lag_state}\end{equation}
where $P_k^x$ is the coherent state projector at $x \in \Sigma$. 
\end{dfn}
Clearly, $\rho_k(\Sigma,\sigma)$ is a positive semidefinite Hermitian operator acting on $\Hil_k$, and 
\[ \Tr(\rho_k(\Sigma,\sigma)) = \int_{\Sigma} \Tr(P_k^x) \sigma(x) = \int_{\Sigma} \sigma = 1. \]
Therefore $\rho_k(\Sigma,\sigma)$ is indeed (the density operator of) a state.  

\begin{ex}

We compute an example in a simple (but non compact) case: $M = \R^2$ with its standard symplectic form and complex structure. It is well-known that the relevant quantum spaces are the Bargmann spaces \cite{Bar}
\[ \Hil_k := \left\{ f \psi^k | \ f: \C \to \C \text{ holomorphic}, \quad \int_{\C} |f(z)|^2 \exp(-k|z|^2) \ |dz \wedge d\bar{z}| < +\infty \right\} \]
where $\psi(z) = \exp\left(-\frac{1}{2}|z|^2\right)$, with orthonormal basis $\phi_{k,\ell}: z \to \sqrt{\frac{k^{\ell+1}}{2\pi \ell!}} \ z^{\ell} \psi^k(z)$ for $\ell \geq 0$. By a straightforward computation,
\[ \scal{P_k^z \phi_{k,\ell}}{\phi_{k,m}}_k = \sqrt{\frac{k^{\ell+m}}{\ell! m!}} z^{\ell} \bar{z}^m \exp(-k|z|^2). \] 
We consider $\Sigma = \S^1 = \{ \exp(it)| \ 0 \leq t \leq 2\pi \} \subset \C$ with density $\sigma = \frac{dt}{2 \pi}$, and compute the state $\rho_k(\S^1,\sigma)$ associated with this data. For $\ell,m \geq 0$,
\[ \scal{\rho_k(\S^1,\sigma)\phi_{k,\ell}}{\phi_{k,m}}_k = \int_0^{2\pi} \scal{P_k^{\exp(it)}\phi_{k,\ell}}{\phi_{k,m}}_k \frac{dt}{2\pi} = \sqrt{\frac{k^{\ell+m}}{\ell! m!}} \exp(-k) \int_0^{2\pi} \exp(i(\ell-m)t) \frac{dt}{2\pi}. \]
Hence for every $\ell \geq 0$,
\[ \rho_k(\S^1,\sigma) \phi_{k,\ell} = \frac{k^{\ell}\exp(-k)}{\ell!}   \phi_{k,\ell}. \]
In other words, this state is prepared according to a Poisson probability distribution of parameter $k$ with respect to the basis $(\phi_{k,\ell})_{\ell \geq 0}$.
\end{ex}

\subsection{Computation of the purity}

In order to see how far $\rho_k(\Sigma,\sigma)$ is from being pure, one can compute its purity $\Tr(\rho_k(\Sigma,\sigma)^2)$, which is equal to one for pure states and strictly smaller than one for mixed states.

\begin{prop}
\label{prop:purity}
Let $\mu_{g,\Sigma}$ be the Riemannian volume on $\Sigma$ corresponding to the Riemannian metric induced by the K\"ahler metric $g$ on $\Sigma$. The purity of $\rho_k(\Sigma,\sigma)$ satisfies
\[ \Tr\left(\rho_k(\Sigma,\sigma)^2\right) =   \left( \frac{2 \pi}{k} \right)^{\frac{d}{2}} \left( \int_{\Sigma} f \sigma + \bigO{k^{-1}} \right),\]
where the function $f$ is such that $\sigma = f \mu_{g,\Sigma}$. In particular, for $k$ large enough, this state cannot be pure.
\end{prop}

\begin{proof}
We need to compute
\[ \Tr(\rho_{k}\left(\Sigma,\sigma)^2\right) = \int_{\Sigma} \int_{\Sigma} \Tr\left(P_k^x P_k^y\right) \sigma(x) \sigma(y). \]
In order to do so, let $(\varphi_j)_{1 \leq j \leq d_k}$, where $d_k = \dim(\Hil_k)$, be any orthonormal basis of $\Hil_k$. Let $x,y \in M$ and let $u,v \in L$ be unit vectors such that $u \in L_x, v \in L_y$. Then
\[ P_k^x P_k^y \varphi_j = \frac{\scal{\varphi_j}{\xi_k^v}_k \scal{\xi_k^v}{\xi_k^u}_k}{\| \xi_k ^v \|^2_k \|\xi_k^u\|^2_k} \xi_k^u \]
for every $j \in \llbracket 1,d_k \rrbracket$. Therefore, 
\[ \Tr\left(P_k^x P_k^y\right) =  \frac{\scal{\xi_k^v}{\xi_k^u}_k}{\| \xi_k ^u \|^2_k \|\xi_k^v\|^2_k} \sum_{j=1}^{d_k} \scal{\varphi_j}{\xi_k^v}_k \scal{\xi_k^u}{\varphi_j}_k = \frac{\scal{\xi_k^v}{\xi_k^u}_k \scal{\xi_k^u}{\xi_k^v}_k}{\| \xi_k ^u \|^2_k \|\xi_k^v\|^2_k} = \frac{|\scal{\xi_k^u}{\xi_k^v}|^2_k}{\| \xi_k^u \|^2_k \|\xi_k^v\|^2_k}. \]
We can rewrite this expression, using the properties stated in \cite[Section 5]{ChaBTO}, as
\[ \Tr\left(P_k^x P_k^y\right) = \frac{|\Pi_k(x,y)|^2}{|\Pi_k(x,x)| \ |\Pi_k(y,y)|}. \]
Hence, we finally obtain that
\[ \Tr\left(\rho_{k}(\Sigma,\sigma)^2\right) = \int_{\Sigma} \int_{\Sigma}  \frac{|\Pi_k(x,y)|^2}{|\Pi_k(x,x)| \ |\Pi_k(y,y)|} \sigma(x) \sigma(y). \]
Since the section $S$ introduced in Equation (\ref{eq:asympt_projector}) satisfies $|S(x,y)| < 1$ whenever $x \neq y$,
\[ \Tr\left(\rho_{k}(\Sigma,\sigma)^2\right) = \int_{(x,y) \in V}  \frac{|\Pi_k(x,y)|^2}{|\Pi_k(x,x)| \ |\Pi_k(y,y)|} \sigma(x) \sigma(y) + \bigO{k^{-\infty}} \]
where $V$ is a neighbourhood of the diagonal of $\Sigma^2$ in $\Sigma^2$. By taking a smaller $V$ if necessary, we may assume that $S$ does not vanish on $V$, and define $\varphi = -2\log |S|$ on the latter. We then deduce from Equation (\ref{eq:asympt_projector}) that $\Tr(\rho_{k}(\Sigma,\sigma)^2) =  (1 + \bigO{k^{-1}}) I_k$ where
\[ I_k = \int_{V} \exp(-k \varphi(x,y)) a_0(x,y)^2 (\sigma \otimes \sigma)(x,y) . \]
In order to estimate this integral, we will apply the stationary phase lemma \cite[Theorem $7.7.5$]{Hor}, with the subtlety that the phase function $\varphi$ has a submanifold of critical points. Indeed, by \cite[Proposition 1]{ChaBTO}, its critical locus is given by
\[ \mathcal{C}_{\varphi} = \{ (x,y) \in  V | \ d\varphi(x,y) = 0 \} = \mathrm{diag}(\Sigma^2).  \]
In this situation, we need to check that the Hessian of $\varphi$ is non degenerate in the transverse direction at every critical point $(x,x)$, $x \in \Sigma$. But we know from \cite[Proposition 1]{ChaBTO} that it is the case, since at such a point, the kernel of this Hessian is equal to $T_{(x,x)}\mathrm{diag}(\Sigma^2)$ and its restriction to the orthogonal complement of $T_{(x,x)}\mathrm{diag}(\Sigma^2)$ is equal to $2 \tilde{g}_{(x,x)}$, where $\tilde{g}$ is the K\"ahler metric on $M \times M$ induced by the symplectic form $\omega \oplus -\omega$ and complex structure $j \oplus -j$. We choose a finite cover of $V$ by open sets of the form $U \times U$, with $U$ a coordinate chart for $\Sigma$ with local coordinates $x_1, \ldots x_d$, and use a partition of unity argument to work with
\[ J_k = \int_{U \times U} \exp(-k \varphi(x,y)) a_0(x,y)^2 h(x) h(y)  \ dx_1 \ldots dx_d dy_1 \ldots dy_d \]
where $h$ is the function such that $\sigma = h \ dx_1 \ldots dx_d$ on $U$. Observe that if $x$ belongs to $U$, the determinant of the transverse Hessian of $\varphi$ at $(x,x)$ is equal to the determinant $\det g_{x,\Sigma} \neq 0$, where $g_{x,\Sigma}$ is the matrix of ${g_x}_{|T_x \Sigma \times T_x \Sigma}$ in the basis corresponding to our local coordinates. Therefore the stationary phase lemma yields 
\[ J_k = \left( \frac{2\pi}{k} \right)^{\frac{d}{2}} \int_U \exp(-k\varphi(x,x)) |\det g_{x,\Sigma}|^{-1/2} a_0(x,x)^2 h(x)^2 \ dx_1 \ldots dx_d + \bigO{k^{-(\frac{d}{2}+1)}}.\]
But by definition, $\mu_{g,\Sigma}(x) = |\det g_{x,\Sigma}|^{1/2} \ dx_1 \ldots dx_d$ on $U$, therefore the function $f$ introduced in the statement of the proposition satisfies $f(x) |\det g_{x,\Sigma}|^{1/2} = h(x) $ on $U$. Since moreover $\varphi(x,x) = 0$ and $a_0(x,x) = 1$, this yields the result.
\end{proof}

\begin{ex}

It follows from the properties of coherent states stated in \cite[Section 5]{ChaBTO} and Equation (\ref{eq:asymp_rawnsley}) that
\[ \mathrm{Id}_{\Hil_k} = T_k(1) = \int_M P_k^x \|\xi_k^u\|^2_{k}  \ \mu_M(x) = \left( 1 + \bigO{k^{-1}} \right) \left(\frac{k}{2\pi}\right)^n \int_M P_k^x  \mu_M(x). \]
where $\pi(u) = x$. Consequently, if we consider the density $\sigma = (\mathrm{Vol}(M))^{-1} \mu_M$ on $M$, then
\[ \rho_k(M,\sigma) = \left(\frac{2\pi}{k}\right)^n  \frac{1}{\mathrm{Vol}(M)} \left( 1 + \bigO{k^{-1}} \right)  \mathrm{Id}_{\Hil_k}. \]
Thus, we finally obtain that 
\[ \Tr\left( \rho_k(M,\sigma)^2 \right) =  \left(\frac{2\pi}{k}\right)^{2n} \frac{ \dim \Hil_k}{\mathrm{Vol}(M)^2} \left( 1 + \bigO{k^{-1}} \right) =  \left(\frac{2\pi}{k}\right)^{n} \frac{1}{\mathrm{Vol}(M)}\left( 1 + \bigO{k^{-1}} \right), \]
Thanks to Equation (\ref{eq:asymp_dim})
\[ \Tr\left( \rho_k(M,\sigma)^2 \right) =  \left(\frac{2\pi}{k}\right)^{n} \frac{1}{\mathrm{Vol}(M)}\left( 1 + \bigO{k^{-1}} \right), \]
which is consistent with the result of the above proposition because the function $f$ associated with $\sigma$ is  $\mathrm{Vol}(M)^{-1}$ (since the Liouville and Riemannian volume forms coincide).

\end{ex}

\subsection{Microsupport and other properties}

Let us now state a few properties of this state $\rho_k(\Sigma,\sigma)$. Given a state $\eta$ and a quantum observable $T$, the expectation of $T$ with respect to $\eta$ is defined as $\mathbb{E}(\eta,T) = \Tr(T \eta)$. In the case where $\eta$ is the state associated with $(\Sigma,\sigma)$, we can obtain a complete asymptotic expansion of this expectation.

\begin{lm}
Let $T_k$ be a self-adjoint Berezin-Toeplitz operator acting on $\Hil_k$, and let $\sum_{\ell \geq 0} \hbar^{\ell} t_{\ell}$ be the covariant symbol of $T_k$ (see \cite[Definition 3]{ChaBTO}). Then $\mathbb{E}(\rho_k(\Sigma,\sigma),T_k)$ has the following asymptotic expansion:
\[ \mathbb{E}(\rho_k(\Sigma,\sigma),T_k) = \sum_{\ell \geq 0} k^{-\ell} \int_{\Sigma} t_{\ell}(x) \ \sigma(x) + \bigO{k^{-\infty}}. \]
In particular, if $T_k = \Pi_k f_0$ for some function $f_0 \in \classe{\infty}{(M,\R)}$, then
\[ \mathbb{E}(\rho_k(\Sigma,\sigma),T_k) = \int_{\Sigma} f_{0}(x) \ \sigma(x) + \bigO{k^{-1}}.\]
\end{lm}

\begin{proof}
Let $(\varphi_j)_{1 \leq j \leq d_k}$, $d_k = \dim(\Hil_k)$, be any orthonormal basis of $\Hil_k$. For $x$ in $\Sigma$, let $u \in L_x$ be a unit vector. Then
\[ \mathbb{E}(\rho_k(\Sigma,\sigma),T_k) = \int_{\Sigma} \scal{ \sum_{j=1}^{d_k} \scal{T_k \xi_k^{u,\mathrm{norm}}}{\varphi_j} _k \varphi_j}{\xi_k^{u,\mathrm{norm}}}_k \sigma(x) =  \int_{\Sigma} \scal{T_k \xi_k^{u,\mathrm{norm}}}{\xi_k^{u,\mathrm{norm}}}_k \sigma(x). \]
The statement follows from the equalities $\scal{T_k \xi_k^u}{\xi_k^u}_k = T_k(x,x)$ and $\| \xi_k^u \|^2_k = \Pi_k(x,x)$, see \cite[Section 5]{ChaBTO}, and from the definition of the covariant symbol, see \cite[Definition 3]{ChaBTO}.
\end{proof}

This result shows in which sense the state associated with $(\Sigma,\sigma)$ concentrates on $\Sigma$ in the semiclassical limit. Indeed, one can introduce, as in \cite[Section 4]{ChaPol}, the microsupport of any state in the following way; the semiclassical measure $\nu_k$ of a state $\eta_k$ is defined as
\[ \int_M f d\nu_k = \Tr(T_k(f) \eta_k) = \mathbb{E}(\eta_k,T_k(f)). \]
Then the microsupport $\mathrm{MS}(\eta_k)$ of $\eta_k$ is the complementary set of the set of points of $M$ having an open neighbourhood $U$ such that $\nu_k(U) = \mathcal{O}(k^{-\infty})$. 

\begin{cor}
The microsupport of $\rho_k(\Sigma,\sigma)$ coincides with $\Sigma$. 
\end{cor}

\begin{proof}
Let $\nu_k$ be the semiclassical measure of $\rho_k(\Sigma,\sigma)$. Let $m \in M \setminus \Sigma$; since $\Sigma$ is closed, there exists an open neighbourhood $V$ of $m$ in $M$ not intersecting $\Sigma$. Let $U$ be an open neighbourhood of $m$ such that $\overline{U} \subset V$, and let $\chi$ be a nonnegative smooth function equal to one on $U$ and compactly supported in $V$. Then 
\[ \nu_k(U) \leq \int_M \chi \ d\nu_k = \mathbb{E}(\rho_k(\Sigma,\sigma),T_k(\chi)). \]
The last term in this equation is given by the previous lemma; it is $\mathcal{O}(k^{-\infty})$ because all the functions in the covariant symbol of $T_k(\chi)$ vanish on $\Sigma$ since the latter does not intersect the support of $\chi$. Therefore $\nu_k(U) = \mathcal{O}(k^{-\infty})$ and $m \notin \mathrm{MS}(\rho_k(\Sigma,\sigma))$.

Conversely, let $m \in \Sigma$ and let $U$ be any open neighbourhood of $m$ in $M$. Choose another open neighbourhood $V$ of $m$ such that $\overline{V} \subset U$, and let $\chi$ be a smooth function, compactly supported in $U$, equal to one on $V$, and such that $0 \leq \chi \leq 1$. Then 
\[ \nu_k(U) \geq \int_M \chi \ d\nu_k = \mathbb{E}(\rho_k(\Sigma,\sigma),T_k(\chi)). \]
But by the previous lemma, we have that
\[ \mathbb{E}(\rho_k(\Sigma,\sigma),T_k(\chi)) = \int_{\Sigma} \chi(x) \sigma(x) + \mathcal{O}(k^{-1}) \geq \int_{\Sigma \cap V} \sigma + \mathcal{O}(k^{-1}). \]
Since the integral of $\sigma$ on $\Sigma \cap V$ is positive, this implies that $m$ belongs to $\mathrm{MS}(\rho_k(\Sigma,\sigma))$.
\end{proof}

Similarly, the variance of $T$ with respect to $\eta$ is $\text{Var}(\eta,T) = \Tr(T^2 \eta) - \Tr(T \eta)^2$.

\begin{lm}
Let $T_k$ be a self-adjoint Berezin-Toeplitz operator acting on $\Hil_k$, with covariant symbol $\sum_{\ell \geq 0} \hbar^{\ell} t_{\ell}$. 
Let $\sum_{\ell \geq 0} \hbar^{\ell} u_{\ell} $ be the covariant symbol of $T_k^2$. Then $\mathrm{Var}(\rho_k(\Sigma,\sigma),T_k)$ has the following asymptotic expansion:
\[ \mathrm{Var}(\rho_k(\Sigma,\sigma),T_k) = \sum_{\ell \geq 0} k^{-\ell} \left( \int_{\Sigma} u_{\ell}(x) \sigma(x) - \sum_{m=0}^{\ell} \int_{\Sigma} \int_{\Sigma} t_m(x) t_{\ell - m}(y) \ \sigma(x) \sigma(y) \right) + \bigO{k^{-\infty}}. \]
In particular, if $T_k = \Pi_k f_0$ with $f_0 \in \classe{\infty}{(M,\R)}$, then 
\[ \mathrm{Var}(\rho_k(\Sigma,\sigma),T_k) = \int_{\Sigma} f_{0}(x)^2 \sigma(x) - \left(\int_{\Sigma} f_{0}(x) \sigma(x)\right)^2 + \bigO{k^{-1}}. \]
\end{lm}

\begin{proof}
Apply the previous lemma to both $T_k^2$ and $T_k$.
\end{proof}

Now, assume that we are in the special case where $T_k = \Pi_k f_0$ and ${f_0}_{|\Sigma} = E \in \R$; then the previous results yield $\mathbb{E}(\rho_k(\Sigma,\sigma),T_k) = E + \bigO{k^{-1}}$ and $\mathrm{Var}(\rho_k(\Sigma,\sigma),T_k) = \bigO{k^{-1}} $.

\subsection{Fidelity for states associated with non intersecting submanifolds}

Let $\Sigma_1, \Sigma_2 \subset M$ be two closed, connected submanifolds of $M$, endowed with densities $\sigma_1, \sigma_2$ such that 
$\int_{\Gamma_i} \sigma_i = 1$, $i=1,2$. Using the notation introduced in Equation (\ref{eq:Lag_state}), we define the states $\rho_{k,i} = \rho_k(\Sigma_i,\sigma_i)$, $i=1,2$. Our goal is to estimate the fidelity $F(\rho_{k,1},\rho_{k,2})$ in the limit $k \to \infty$. Of course, if $(\Sigma_1,\sigma_1) = (\Sigma_2,\sigma_2)$, then $\rho_{k,1} = \rho_{k,2}$ and $F(\rho_{k,1},\rho_{k,2}) = 1$. The following result deals with the case where $\Sigma_1$ and $\Sigma_2$ are disjoint.

\begin{prop}
\label{prop:empty_int}
Assume that $\Sigma_1 \cap \Sigma_2 = \emptyset$. Then $F(\rho_{k,1},\rho_{k,2}) = \bigO{k^{-\infty}}$.
\end{prop}

\begin{proof}
By using the Cauchy-Schwarz inequality for the inner product $(A,B) \mapsto \Tr(B^*A)$ on the space of operators on $\Hil_k$, and the fact that the trace is invariant under cyclic permutations, we get that $F(\rho_{k,1},\rho_{k,2}) \leq \dim(\Hil_k) \Tr(\rho_{k,1} \rho_{k,2})$. Since the dimension of $\Hil_k$ is of order $k^n$, it is therefore sufficient to show that $\Tr(\rho_{k,1} \rho_{k,2}) = \bigO{k^{-\infty}}$. The same computations as in the proof of Proposition \ref{prop:purity} yield
\begin{equation} \Tr(\rho_{k,1} \rho_{k,2}) = \int_{\Sigma_1} \int_{\Sigma_2}  \frac{|\Pi_k(x,y)|^2}{|\Pi_k(x,x)| |\Pi_k(y,y)|} \sigma_1(x) \sigma_2(y). \label{eq:trace}\end{equation}
Since $\Sigma_1 \times \Sigma_2$ does not meet the diagonal of $M \times M$, $\Pi_k$ is uniformly $\bigO{k^{-\infty}}$ on $\Sigma_1 \times \Sigma_2$. Moreover, it follows from Equation (\ref{eq:asympt_projector}) that $|\Pi_k(x,x)| \sim \left( \frac{k}{2\pi} \right)^n$ uniformly on $M$, so the above formula yields $\Tr(\rho_{k,1} \rho_{k,2}) = \bigO{k^{-\infty}}$.
\end{proof}

Consequently, we will now be interested in an intermediate case, namely in the situation where $\Sigma_1$ and $\Sigma_2$ are distinct but have non empty intersection at a finite number of points. Of course in this case fidelity is still expected to tend to zero as $k$ goes to infinity, but one might be able to estimate the rate of convergence and the relation between fidelity and the underlying geometry. As already explained, fidelity is in general too complicated to compute and we will rather be interested in the sub and super fidelities. We will explain how to estimate these quantities when $\Sigma_1$ and $\Sigma_2$ are Lagrangian submanifolds, which moreover intersect transversally at a finite number of points.

\section{Sub and super fidelity for two Lagrangian states}
\label{sect:sub_super}

In this section, we assume that $\Gamma_1$ and $\Gamma_2$ are two closed, connected Lagrangian submanifolds of $M$, endowed with densities $\sigma_1, \sigma_2$ such that $\int_{\Gamma_i} \sigma_i = 1$, $i=1,2$, and intersecting transversally at a finite number of points $m_1, \ldots, m_s$. As before, we set $\rho_{k,i} = \rho_k(\Gamma_i, \sigma_i)$. 

\begin{dfn} \label{dfn:thetas} For $\nu \in \llbracket 1,s \rrbracket$, we consider the principal angles
\[ 0 < \theta_1(m_{\nu}) \leq \ldots \leq \theta_{n}(m_{\nu}) \leq \frac{\pi}{2} \]
between $T_{m_{\nu}}\Gamma_1$ and $T_{m_{\nu}}\Gamma_2$, computed with respect to $g_{m_{\nu}}$ (recall that $g$ is the K\"ahler metric on $M$).\end{dfn}

For $i=1,2$, we introduce as in the statement of Proposition \ref{prop:purity} the Riemannian volume $\mu_{g,\Gamma_i}$ coming from the Riemannian metric induced by $g$ on $\Gamma_i$, and the function $f_i$ such that $\sigma_i = f_i \mu_{g,\Gamma_i}$. For $\nu \in \llbracket 1, s \rrbracket$, we define
\begin{equation}  (\sigma_1,\sigma_2)_{m_{\nu}} := f_1(m_{\nu}) f_2(m_{\nu}) > 0. \label{eq:constant_sigma}\end{equation}

\begin{thm}
\label{thm:subfid}
The sub-fidelity of $\rho_{k,1}$ and $\rho_{k,2}$ satisfies:
\[ E(\rho_{k,1},\rho_{k,2}) =  \left(\frac{2\pi}{k}\right)^{n}  C((\Gamma_1,\sigma_1),(\Gamma_2,\sigma_2)) + \bigO{k^{-(n+1)}}, \]
where $C((\Gamma_1,\sigma_1),(\Gamma_2,\sigma_2)) = C_1 + \sqrt{2(C_2 + C_3)}$ with
\[ C_1 = \sum_{\nu=1}^s \frac{(\sigma_1,\sigma_2)_{m_{\nu}}}{\prod_{\ell = 1}^n \sin(\theta_{\ell}(m_{\nu}))}, \quad C_2  = \sum_{\substack{\nu=1}}^s \sum_{\substack{\mu=1 \\ \mu \neq \nu}}^s \frac{(\sigma_1,\sigma_2)_{m_{\nu}} (\sigma_1,\sigma_2)_{m_{\mu}}}{\prod_{\ell = 1}^n \sin(\theta_{\ell}(m_{\nu})) \sin(\theta_{\ell}(m_{\mu}))} \] 
and finally
\[ C_3 = \sum_{\nu=1}^s \frac{(\sigma_1,\sigma_2)_{m_{\nu}}^2}{\prod_{\ell = 1}^n \sin(\theta_{\ell}(m_{\nu}))} \left( \prod_{\ell = 1}^n\frac{1}{\sin(\theta_{\ell}(m_{\nu}))} - \prod_{\ell = 1}^n\frac{1}{\sqrt{1 + \sin^2(\theta_{\ell}(m_{\nu}))}} \right). \]
\end{thm}

The rest of this section is devoted to the proof of this result; we start by estimating the trace $\Tr(\rho_{k,1}\rho_{k,2})$, which gives $C_1(\Gamma_1,\Gamma_2)$, then we estimate $\Tr((\rho_{k,1}\rho_{k,2})^2)$ to obtain the remaining terms. 

\begin{rmk}
As can be seen from the proofs (and using the complete description of the Bergman kernel), the sub-fidelity actually has a complete asymptotic expansion in powers of $k$ smaller than $-n$; we are only interested here in the first term of this expansion.
\end{rmk}

\subsection{The term $\Tr(\rho_{k,1} \rho_{k,2})$}

We are now ready to estimate the trace of $\rho_{k,1} \rho_{k,2}$.

\begin{thm}
\label{thm:trace}
We have the following estimate:
\[ \Tr(\rho_{k,1}\rho_{k,2}) = \left(\frac{2\pi}{k}\right)^n \left( \sum_{\nu=1}^s \frac{(\sigma_1,\sigma_2)_{m_{\nu}}}{\prod_{\ell = 1}^n \sin(\theta_{\ell}(m_{\nu}))} \right) + \bigO{k^{-(n+1)}}, \]
see Definition \ref{dfn:thetas} and Equation (\ref{eq:constant_sigma}) for notation.
\end{thm}

\begin{proof}
By Equation (\ref{eq:trace}), this trace is given by the formula
\[ \Tr(\rho_{k,1} \rho_{k,2}) = \int_{\Gamma_1} \int_{\Gamma_2}  \frac{|\Pi_k(x,y)|^2}{|\Pi_k(x,x)| |\Pi_k(y,y)|} \sigma_1(x) \sigma_2(y). \]
The same argument that we used in the proof of Proposition \ref{prop:empty_int} shows that the integral over $x,y \in M \setminus \bigcup_{j=1}^p \Omega_{\nu}$, where $\Omega_{\nu}$ is a neighbourhood of the intersection point $m_{\nu}$, is a $\bigO{k^{-\infty}}$. Therefore, we only need to understand what the contribution of the integral
\[ I_{k,\nu} =  \int_{\Gamma_1 \cap \Omega_{\nu}} \int_{\Gamma_2 \cap \Omega_{\nu}}  \frac{|\Pi_k(x,y)|^2}{|\Pi_k(x,x)| |\Pi_k(y,y)|} \sigma_1(x) \sigma_2(y). \]
is, for every $\nu \in \llbracket 1,p \rrbracket$, and to sum up these contributions. Equation (\ref{eq:asympt_projector}) implies that
\[ I_{k,\nu} =  \left(\int_{\Gamma_1 \cap \Omega_{\nu}} \int_{\Gamma_2 \cap \Omega_{\nu}}  |S(x,y)|^{2k} a_0(x,y)^2 \sigma_1(x) \sigma_2(y)\right)\left(1 + \bigO{k^{-1}}\right). \]
By working with a smaller $\Omega_{\nu}$ if necessary, we may assume that $S$ does not vanish on $\Omega_{\nu} \times \Omega_{\nu}$, and define $\varphi = -2\log|S|$ on the latter. Then $I_{k,\nu} = J_{k,\nu}(1+\bigO{k^{-1}})$ with
\[ J_{k,\nu} =  \int_{\Gamma_1 \cap \Omega_{\nu}} \int_{\Gamma_2 \cap \Omega_{\nu}}  \exp(-k\varphi(x,y)) a_0(x,y)^2 \sigma_1(x) \sigma_2(y). \]
We will evaluate this integral by means of the stationary phase method. By taking a smaller $\Omega_{\nu}$ if necessary, we consider a local diffeomorphism $\eta: (\Omega_{\nu},m_{\nu}) \to (\Theta_{\nu},0) \subset \R^{2n}$ such that $\eta(\Gamma_1 \cap \Omega_{\nu}) = \{(u,v) \in \Theta_{\nu}| \ v = 0 \}$ and $ \eta(\Gamma_2 \cap \Omega_{\nu}) = \{(u, v) \in \Theta_{\nu}| \ u = 0 \}$. Let $\kappa_1, \kappa_2$ be such that $\kappa_1(u) = \eta^{-1}(u,0)$ and $\kappa_2(v) = \eta^{-1}(0,v)$. We have that
\[ J_{k,\nu} = \int_{\mathrm{pr}_1(\Theta_{\nu})} \int_{\mathrm{pr}_2(\Theta_{\nu})} \exp(-k\psi(u,v)) b_0(u,v)^2 \ \kappa_1^*\sigma_1(u) \  \kappa_2^*\sigma_2(v),  \]
where $\mathrm{pr}_i$, $i=1,2$ are the projections on the first and second factor of $\R^n \times \R^n$, where the phase reads $\psi(u,v) = \varphi(\kappa_1(u),\kappa_2(v))$ and the amplitude is given by the formula $b_0(u,v) = a_0(\kappa_1(u),\kappa_2(v))$. If $h_1,h_2$ are such that $\kappa_1^*\sigma_1 = h_1(u) du$ and $\kappa_2^*\sigma_2 = h_2(v) dv$ locally, then
\[ J_{k,\nu} = \int_{\mathrm{pr}_1(\Theta_{\nu})} \int_{\mathrm{pr}_2(\Theta_{\nu})} \exp(-k\psi(u,v)) b_0(u,v)^2 h_1(u) h_2(v) \ du \ dv.  \]
The phase $\psi$ is non-negative. Its differential is given by
\[ d\psi(u,v) \cdot (U,V) = d\varphi(\kappa_1(u),\kappa_2(v)) \cdot (d\kappa_1(u) \cdot U, d\kappa_2(v) \cdot V);   \]
therefore, because of \cite[Proposition 1]{ChaBTO}, the point $(u,v)$ is a critical point for $\psi$ if and only if $\kappa_1(u) = \kappa_2(v)$, thus if and only if $u=0=v$. Furthermore, $\psi(0,0) = 0$, and the second order differential of $\psi$ at the critical point $(0,0)$ reads
\[ d^2\psi(0,0) \cdot ((U,V),(X,Y)) = d^2\varphi(m_{\nu},m_{\nu}) \cdot \left( (d\kappa_1(0) \cdot U, d\kappa_2(0) \cdot V),(d\kappa_1(0) \cdot X, d\kappa_2(0) \cdot Y) \right). \]
We will prove that this bilinear form is positive definite. Let $(e_{\ell})_{1 \leq \ell \leq n}$ (respectively $(f_{\ell})_{1 \leq \ell \leq n}$) be an orthonormal basis (with respect to the restriction of $g_{m_{\nu}}$) of the subspace $T_{m_{\nu}} \Gamma_1 \subset T_{m_{\nu}} M$ (respectively $T_{m_{\nu}} \Gamma_2$). We define the vectors $U_{\ell} = (d\kappa_1(0))^{-1} \cdot e_{\ell}$ and $V_{\ell} = (d\kappa_2(0))^{-1} \cdot f_{\ell}$ of $\R^{n}$, for $1 \leq \ell \leq n$. By composing $\eta$ with a linear diffeomorphism if necessary, we may assume that $((U_{\ell},0)_{1 \leq \ell \leq n},(0,V_{\ell})_{1 \leq \ell \leq n})$ is the standard basis of $\R^{2n}$; let us compute the matrix $A$ of $d^2\psi(0,0)$ in this basis. We have that
\[ d^2\psi(0,0) \cdot ((U_{\ell},0),(U_p,0)) = d^2\varphi(m_{\nu},m_{\nu}) \cdot \left( (e_{\ell}, 0), (e_{p}, 0) \right). \]
By \cite[Proposition 1]{ChaBTO}, $d^2\varphi(m_{\nu},m_{\nu})$ has kernel $T_{(m_{\nu},m_{\nu})} \Delta$, where $\Delta$ is the diagonal of $M^2$, and its restriction to $(T_{(m_{\nu},m_{\nu})} \Delta)^{\perp}$ is equal to $2 \tilde{g}_{(m_{\nu},m_{\nu})}$, where we recall that $\tilde{g}$ is the K\"ahler metric on $M \times M$ induced by the symplectic form $\omega \oplus -\omega$ and complex structure $j \oplus -j$. But
\[ (e_{\ell},0) = \frac{1}{2} (e_{\ell},e_{\ell}) + \frac{1}{2} (e_{\ell},-e_{\ell}), \]
is the decomposition of $(e_{\ell},0)$ in the direct sum $T_{(m_{\nu},m_{\nu})} \Delta \oplus (T_{(m_{\nu},m_{\nu})} \Delta)^{\perp}$. Therefore
\[ d^2\psi(0,0) \cdot ((U_{\ell},0),(U_p,0)) = \tilde{g}_{(m_{\nu},m_{\nu})}((e_{\ell},-e_{\ell}),(e_p,0)) = g_{m_{\nu}}(e_{\ell},e_p) = \delta_{\ell,p}. \]
Similarly, $d^2\psi(0,0) \cdot ((0,V_{\ell}),(0,V_p)) = \delta_{\ell,p}$. Finally,
\[ d^2\psi(0,0) \cdot ((U_{\ell},0),(0,V_p)) = \tilde{g}_{(m_{\nu},m_{\nu})}((e_{\ell},-e_{\ell}),(0,f_{p})) = -g_{m_{\nu}}(e_{\ell},f_p), \]
so  $A$ is the block matrix
\[ A = \begin{pmatrix} \mathrm{Id} & -G \\ -G^{\top} & \mathrm{Id} \end{pmatrix}, \]
where $G$ is the $n \times n$ matrix with entries $G_{\ell,p} = g_{m_{\nu}}(e_{\ell},f_p)$. Thus its determinant satisfies $\det(A) = \det(\mathrm{Id} - G^{\top}G)$; hence, Lemma \ref{lm:angle} yields
\[\det(A) = \prod_{\ell = 1}^n \sin^2(\theta_{\ell}(m_{\nu})) > 0 \]
and the stationary phase lemma gives the estimate
\[ J_{k,\nu} = \left( \frac{2\pi}{k} \right)^{n} \frac{h_1(0)h_2(0)a_0(m_{\nu},m_{\nu})^2}{\prod_{\ell = 1}^n \sin(\theta_{\ell}(m_{\nu}))} + \bigO{k^{-(n+1)}}. \]  
We have that $a_0(m_{\nu},m_{\nu}) = 1$, and we claim that $h_1(0) h_2(0) = (\sigma_1,\sigma_2)_{m_{\nu}}$. Indeed, thanks to our choices, we know that $(\kappa_1^*\mu_{g,\Gamma_1})(0) = du$ and $\kappa_1^*\sigma_1 = h_1(u) du$; therefore $h_1(0) = f_1(m_{\nu})$, and similarly $h_2(0) = f_2(m_{\nu})$. We then obtain the result by summing up the contributions of all the intersection points $m_{\nu}$, $1 \leq \nu \leq s$. 
\end{proof}

\begin{rmk}
In this proof we have not used the fact that our submanifolds are Lagrangian, hence the result still holds without this assumption. We also believe that we could even drop the assumption that they are $n$-dimensional and consider instead two submanifolds of respective dimensions $d$ and $2n-d$ intersecting transversally at a finite number of points. Handling this case would require some care but in this setting, $d$ principal angles are still well-defined, and everything should work as if the $n-d$ others are taken to be equal to $\frac{\pi}{2}$. Nevertheless, as we will see below, the Lagrangian assumption is crucial in order to estimate the next term, so we chose to stick to the Lagrangian case for this first result.
\end{rmk}

\subsection{The term $\Tr((\rho_{k,1} \rho_{k,2})^2)$}

We can now estimate the trace of $(\rho_{k,1} \rho_{k,2})^2$, which is equal to
\[ \Tr\left( (\rho_{k,1} \rho_{k,2})^2 \right) = \int_{\Gamma_1} \int_{\Gamma_2} \int_{\Gamma_1} \int_{\Gamma_2} \Tr(P_k^{x_1} P_k^{x_2} P_k^{y_1} P_k^{y_2}) \ \sigma_1(x_1) \sigma_2(x_2) \sigma_1(y_1) \sigma_2(y_2)  .\]
A straightforward computation, similar to the one in the proof of Proposition \ref{prop:purity}, yields 
\[ \Tr(P_k^{x_1} P_k^{x_2} P_k^{y_1} P_k^{y_2}) = \frac{ \scal{\xi_k^{v_2}}{\xi_k^{v_1}} \scal{\xi_k^{v_1}}{\xi_k^{u_2}} \scal{\xi_k^{u_2}}{\xi_k^{u_1}} \scal{\xi_k^{u_1}}{\xi_k^{v_2}} }{ \|\xi_k^{u_1}\|^2 \|\xi_k^{u_2}\|^2 \|\xi_k^{v_1}\|^2 \|\xi_k^{v_2}\|^2} \]
where $u_i$ (respectively $v_i$), $i=1,2$ is any unit vector in $L_{x_i}$ (respectively $L_{y_i}$). This trace is therefore a $\bigO{k^{-\infty}}$ uniformly on $M^2 \setminus \bigcup_{\nu = 1}^s (\Omega_{\nu} \times \Omega_{\nu})$ where $\Omega_{\nu}$ is a neighbourhood of $m_{\nu}$ in $M$. Consequently, the only non negligible contributions to $\Tr\left( (\rho_{k,1} \rho_{k,2})^2 \right)$ come from the integrals 
\[ I_{k,\nu} = \int_{\Gamma_1 \cap \Omega_{\nu}} \int_{\Gamma_2 \cap \Omega_{\nu}} \int_{\Gamma_1 \cap \Omega_{\nu}} \int_{\Gamma_2 \cap \Omega_{\nu}} \Tr(P_k^{x_1} P_k^{x_2} P_k^{y_1} P_k^{y_2}) \ \sigma_1(x_1) \sigma_2(x_2) \sigma_1(y_1) \sigma_2(y_2),\]
for $\nu \in \llbracket 1, s \rrbracket$. In order to estimate the scalar products appearing in this integral, let $S$ be as in Equation (\ref{eq:asympt_projector}), let $t$ be a local section of $L$ over $\Omega_{\nu}$ with unit norm, let $\psi: \Omega_{\nu} \times \Omega_{\nu} \to \C$ be such that 
\[ S(x,y) = \exp(i\psi(x,y)) \ t(x) \otimes \overline{t(y)} \]
over $\Omega_{\nu} \times \Omega_{\nu}$, and set $u_i = t(x_i), v_i = t(y_i)$, $i=1,2$. We derive from {\cite[Section 5]{ChaBTO}} that
\[ \scal{\xi_k^{v_2}}{\xi_k^{v_1}} = \overline{t(y_1)}^k \cdot \Pi_k(y_1,y_2) \cdot t(y_2)^k =  \left( \frac{k}{2\pi} \right)^n \exp(ik\psi(y_1,y_2)) \left( a_0(y_1,y_2) + \bigO{k^{-1}} \right) \] 
uniformly on any compact subset of $\Omega_{\nu} \times \Omega_{\nu}$, and we obtain similar expressions for the other scalar products. Hence, $I_{k,\nu} = J_{k,\nu}\left( 1 + \bigO{k^{-1}} \right)$ with
\[ J_{k,\nu} = \int_{\Gamma_1 \cap \Omega_{\nu}} \int_{\Gamma_2 \cap \Omega_{\nu}} \int_{\Gamma_1 \cap \Omega_{\nu}} \int_{\Gamma_2 \cap \Omega_{\nu}} \exp(ik\Psi(x_1,x_2,y_1,y_2)) b_0(x_1,x_2,y_1,y_2)  \ \sigma_1(x_1) \sigma_2(x_2) \sigma_1(y_1) \sigma_2(y_2) \]
where the phase $\Psi$ is given by
\[ \Psi(x_1,x_2,y_1,y_2) = \psi(y_1,y_2) + \psi(x_2,y_1) + \psi(x_1,x_2) + \psi(y_2,x_1), \]
and $b_0(x_1,x_2,y_1,y_2) = a_0(y_1,y_2) a_0(x_2,y_1) a_0(x_1,x_2) a_0(y_2,x_1)$. Now, we introduce as in the proof of Theorem \ref{thm:trace} a local diffeomorphism $\eta: (\Omega_{\nu},m_{\nu}) \to (\Theta_{\nu},0) \subset \R^{2n}$ such that 
\[\eta(\Gamma_1 \cap \Omega_{\nu}) = \{(u,v) \in \Theta_{\nu}| \ v = 0 \}, \qquad \eta(\Gamma_2 \cap \Omega_{\nu}) = \{(u, v) \in \Theta_{\nu}| \ u = 0 \}, \]
and the functions $\kappa_i: \mathrm{pr}_i(\Theta_{\nu}) \to \Omega_{\nu}$ defined by the formulas $\kappa_1(u) = \eta^{-1}(u,0)$ and $\kappa_2(v) = \eta^{-1}(0,v)$ (here we recall that $\mathrm{pr}_i$, $i=1,2$ are the projections on the first and second factor of $\R^n \times \R^n$). We also introduce again the functions $h_1,h_2$ such that $\kappa_1^*\sigma_1 = h_1(u) du$ and $\kappa_2^*\sigma_2 = h_2(v) dv$ locally. Then 
\[ J_{k,\nu} = \int_{\mathrm{pr}_1(\Theta_{\nu})} \int_{\mathrm{pr}_2(\Theta_{\nu})} \int_{\mathrm{pr}_1(\Theta_{\nu})} \int_{\mathrm{pr}_2(\Theta_{\nu})} \exp(ik\Phi(u,v,w,z)) c_0(u,v,w,z) \ du \ dv \ dw \ dz\]
where the amplitude $c_0$ is given by
\[ c_0(u,v,w,z) = b_0(\kappa_1(u),\kappa_2(v),\kappa_1(w),\kappa_2(z)) h_1(u)h_2(v)h_1(w)h_2(z) \]
and the phase $\Phi$ reads $\Phi(u,v,w,z) = \Psi(\kappa_1(u),\kappa_2(v),\kappa_1(w),\kappa_2(z))$. We will estimate $J_{k,\nu}$ thanks to another application of the stationary phase method. The imaginary part of $\Phi$ is non-negative and vanishes only at the point $0 = (0,0,0,0)$. 

\paragraph{Computation of $d\Phi$ and critical points of $\Phi$.}

Let $\widetilde{\nabla}$ be the connection induced by $\nabla$ on the line bundle $L \boxtimes \overline{L}$.

\begin{lm}
Let $\alpha_S$ be the differential form defined by the equality $\widetilde{\nabla} S = -i \alpha_S \otimes S$ in a neighbourhood of the diagonal $\Delta$ of $M^2$ where $S$ does not vanish; then
\[ d\Phi(u,v,w,z) \cdot (U,V,W,Z) = f(w,z,W,Z) + f(u,v,U,V) + g(v,w,V,W) + g(z,u,Z,U) \]
where $f$ and $g$ are defined as $f(a,b,A,B) = -{\alpha_S}_{(\kappa_1(a),\kappa_2(b))}(d\kappa_1(a) \cdot A, d\kappa_2(b) \cdot B)$ and $g(a,b,A,B) = -{\alpha_S}_{(\kappa_2(a),\kappa_1(b))}(d\kappa_2(a) \cdot A, d\kappa_1(b) \cdot B)$.
\end{lm}

\begin{proof}
We start from the local expression $S(x,y) = \exp(i\psi(x,y)) t(x) \otimes \overline{t(y)}$, which yields
\[ \widetilde{\nabla}S = i d\psi \otimes  S + \exp(i \psi) \widetilde{\nabla}(t(x) \otimes \overline{t(y)}).  \]
In order to compute the second term, we introduce the local differential form $\beta$ such that $\nabla t = \beta \otimes t$. Then $ \widetilde{\nabla}S = \left(i d\psi + p_1^*\beta + p_2^*\bar{\beta} \right)\otimes  S$, where $p_1, p_2$ are the projections on the first and second factor of $M \times M$. This means that $-i\alpha_S = i d\psi + p_1^*\beta + p_2^*\bar{\beta}$. We claim that there exists a real-valued form $\gamma$ such that $\beta = i \gamma$; indeed, 
\[ 0 = dh(t,t) = h(\nabla t,t) + h(t, \nabla t) = \beta + \bar{\beta} \]
since $\nabla$ and $h$ are compatible. Consequenly, $d\psi = - \alpha_S - p_1^*\gamma + p_2^*\gamma$. Now, the quantity $d\Phi(u,v,w,z) \cdot (U,V,W,Z)$ is the sum of the following four terms:
\[ d\psi((\kappa_1(w),\kappa_2(z))) \cdot (d\kappa_1(w) \cdot W, d\kappa_2(z) \cdot Z), \quad d\psi((\kappa_2(v),\kappa_1(w))) \cdot (d\kappa_2(v) \cdot V, d\kappa_1(w) \cdot W),\]
\[ d\psi((\kappa_1(u),\kappa_2(v))) \cdot (d\kappa_1(u) \cdot U, d\kappa_2(v) \cdot V), \quad d\psi((\kappa_2(z),\kappa_1(u))) \cdot (d\kappa_2(z) \cdot Z, d\kappa_1(u) \cdot U). \]
The quantity $-\gamma_{\kappa_1(w)}(d\kappa_1(w) \cdot W)$ coming from the first term cancels the quantity $\gamma_{\kappa_1(w)}(d\kappa_1(w) \cdot W)$ coming from the second one, and so on.
\end{proof}

Since $\alpha_S$ vanishes on the diagonal (see \cite[Proposition 1]{ChaBTO} and \cite[Lemma 4.3]{ChaHalf}), an immediate corollary of this result is that the differential of $\Phi$ vanishes at $(0,0,0,0)$.

\paragraph{Computation of the determinant of the Hessian of $\Phi$ at the critical point.}

Let $\overline{M}$ be $M$ endowed with the symplectic structure $-\omega$ and the complex structure $-j$; $M \times \overline{M}$ is equipped with the symplectic form $\tilde{\omega} = \omega \oplus - \omega = p_1^* \omega - p_2^* \omega$ with $p_1, p_2$ the natural projections. Similarly, $\tilde{\text{\j}}$ denotes the complex structure $j \oplus -j$ on $M \times \overline{M}$.

Having in mind \cite[Lemma 4.3]{ChaHalf} (or \cite[Section 2.6]{Cha_symp} in a more general setting), we introduce the section $B_S$ of $\left(T^*(M \times \overline{M}) \otimes T^*(M \times \overline{M}) \right) \otimes \C \to \Delta$ such that for any vector fields $X, Y$ of $M \times \overline{M}$, $\derLie_X (\alpha_S(Y)) = B_S(X,Y)$ along $\Delta$; if we set $C = d\kappa_1(0)$ and $D = d\kappa_2(0)$, then for $\mathcal{U} = (U,V,W,Z)$ and $\mathcal{V} = (\hat{U},\hat{V},\hat{W},\hat{Z})$,
\[ d^2\Phi(0) \cdot \left(\mathcal{U}, \mathcal{V} \right) = f(W,Z,\hat{W},\hat{Z}) + f(U,V,\hat{U},\hat{V}) + g(V,W,\hat{V},\hat{W}) + g(Z,U,\hat{Z},\hat{U}) \] 
where $f$ and $g$ read $f(X_1,X_2,X_3,X_4) = - {B_S}_{(m_{\nu},m_{\nu})}((C \cdot X_1, D \cdot X_2),(C \cdot X_3, D \cdot X_4))$ and $g(X_1,X_2,X_3,X_4) = - {B_S}_{(m_{\nu},m_{\nu})}((D \cdot X_1, C \cdot X_2),(D \cdot X_3, C \cdot X_4))$. As before, we consider an orthonormal basis $(e_{\ell})_{1 \leq \ell \leq n}$ (respectively $(f_{\ell})_{1 \leq \ell \leq n}$) of the subspace $T_{m_{\nu}} \Gamma_1 \subset T_{m_{\nu}} M$ (respectively $T_{m_{\nu}} \Gamma_2$), and we assume that the vectors $U_{\ell} = C^{-1} \cdot e_{\ell}, V_{\ell} = D^{-1} \cdot f_{\ell}$ of $\R^{n}$, for $1 \leq \ell \leq n$ are such that $((U_{\ell},0)_{1 \leq \ell \leq n},(0,V_{\ell})_{1 \leq \ell \leq n})$ is the standard basis of $\R^{2n}$. Let $G, \Xi$ be the $n \times n$ matrices with entries $G_{p,q} = g_{m_{\nu}}(e_p,f_q)$ and $\Xi_{p,q} = \omega_{m_{\nu}}(e_p,f_q)$.

\begin{lm}
In the basis $(U_{\ell},0,0,0)_{1 \leq \ell \leq n},(0,0,U_{\ell},0)_{1 \leq \ell \leq n}$, $(0,V_{\ell},0,0)_{1 \leq \ell \leq n},(0,0,0,V_{\ell})_{1 \leq \ell \leq n}$ of $\R^{4n}$, the matrix of $d^2\Phi(0)$ is the block matrix 
\[ H = \begin{pmatrix} i I_{2n} & A \\ A^{\top} & i I_{2n} \end{pmatrix}, \qquad A = \frac{1}{2} \begin{pmatrix} -\Xi - iG & \Xi- iG \\ \Xi - iG & -\Xi - iG \end{pmatrix}.\]
\end{lm}

\begin{proof}
It is clear from the above expression of $d^2\Phi(0)$ that  
\[d^2\Phi(0) \cdot ((U_p,0,0,0),(0,0,U_q,0)) = 0 = d^2\Phi(0) \cdot ((0,V_p,0,0),(0,0,0,V_q)). \]
In order to compute the other terms, we introduce the projection $q$ from $T_x(M \times \overline{M}) \otimes \C$ onto $T_x^{0,1}(M \times \overline{M})$ with kernel $T_x \Delta \otimes \C$, so that $B_S(X,Y) = \tilde{\omega}(q(X),Y)$ (see for instance \cite[Lemma 4.3]{ChaHalf} or \cite[Proposition 2.15]{Cha_symp}). We need to compute $q(e_{\ell},0)$ and $q(0,e_{\ell})$ (and similarly for $f_{\ell}$). So we look for $X,Y \in T_{m_{\nu}}M$ and $Z \in T_{m_{\nu}}M \otimes \C$ such that 
\[ (e_{\ell},0) = (X,Y) + \tilde{\text{\j}}(X,Y) + (Z,Z) = (X+ijX+Z,Y-ijY+Z), \]
in which case $q(e_{\ell},0) = (X+ijX,Y-ijY)$. A straightforward computation shows that $2X = -2Y = e_{\ell}$ and $Z = X-ijX$, hence $q(e_{\ell},0) = \frac{1}{2} \left( e_{\ell} + ij e_{\ell}, - e_{\ell} + i j e_{\ell} \right)$. We obtain in a similar fashion that $q(0,e_{\ell}) = \frac{1}{2} \left( -e_{\ell} - ij e_{\ell}, e_{\ell} - i j e_{\ell} \right)$. Now, we have that
\[ d^2\Phi(0) \cdot ((U_p,0,0,0),(U_q,0,0,0)) = - {\tilde{\omega}}_{(m_{\nu},m_{\nu})}(q(e_p,0),(e_q,0)) - {\tilde{\omega}}_{(m_{\nu},m_{\nu})}(q(0,e_p),(0,e_q)). \]
The first term satisfies
\[ {\tilde{\omega}}_{(m_{\nu},m_{\nu})}(q(e_p,0),(e_q,0)) = \frac{1}{2} \omega_{m_{\nu}}(e_p + ij e_p,e_q) = \frac{1}{2}  \omega_{m_{\nu}}(e_p,e_q) + \frac{i}{2} \omega_{m_{\nu}}(j e_p,e_q); \]
since $T_{m_{\nu}} \Gamma_1$ is Lagrangian, $\omega_{m_{\nu}}(e_p,e_q) = 0$, and finally 
\[ {\tilde{\omega}}_{(m_{\nu},m_{\nu})}(q(e_p,0),(e_q,0)) = -\frac{i}{2} g_{m_{\nu}}(e_p,e_q) = -\frac{i}{2} \delta_{p,q}. \]
A similar computation shows that ${\tilde{\omega}}_{(m_{\nu},m_{\nu})}(q(0,e_p),(0,e_q)) = -\frac{i}{2} \delta_{p,q}$. Therefore,
\[ d^2\Phi(0) \cdot ((U_p,0,0,0),(U_q,0,0,0)) = i \delta_{p,q}. \]
We find the same result for $d^2\Phi(0) \cdot ((0,V_p,0,0),(0,V_q,0,0))$, $d^2\Phi(0) \cdot ((0,0,U_p,0),(0,0,U_q,0))$ and $d^2\Phi(0) \cdot ((0,0,0,V_p),(0,0,0,V_q))$. Combining this with the previous result, we obtain that the two diagonal blocks of $H$ are equal to $i I_{2n}$. Now, we have that
\[  d^2\Phi(0) \cdot ((U_p,0,0,0),(0,V_q,0,0)) =  -\tilde{\omega}_{(m_{\nu},m_{\nu})}(q(e_p,0),(0,f_q)); \]
but we also have that
\[ \tilde{\omega}_{(m_{\nu},m_{\nu})}(q(e_p,0),(0,f_q)) = -\frac{1}{2} \omega_{m_{\nu}}(-e_p + ij e_p,f_q) = \frac{1}{2} \left( \omega_{m_{\nu}}(e_p,f_q) - i \omega_{m_{\nu}}(je_p,f_q) \right). \]
So we finally obtain that 
\[  d^2\Phi(0) \cdot ((U_p,0,0,0),(0,V_q,0,0)) =  -\frac{1}{2} \left( \omega_{m_{\nu}}(e_p,f_q) + i g_{m_{\nu}}(e_p,f_q) \right). \]
We also immediately deduce from this that
\[  d^2\Phi(0) \cdot ((0,0,U_p,0),(0,0,0,V_q)) =  -\frac{1}{2} \left( \omega_{m_{\nu}}(e_p,f_q) + i g_{m_{\nu}}(e_p,f_q) \right). \]
Finally, we derive
\[ d^2\Phi(0) \cdot ((U_p,0,0,0),(0,0,0,V_q)) = \frac{1}{2} \left( \omega_{m_{\nu}}(e_p,f_q) - i g_{m_{\nu}}(e_p,f_q) \right) \]
from a similar computation. The same holds for $d^2\Phi(0) \cdot ((0,0,U_p,0),(0,V_q,0,0))$.
\end{proof}

This result yields $\det(-iH) = \det(I_{2n} + A^{\top}A)$. But one readily checks that 
\[  I_{2n} + A^{\top}A=  I_{2n} + \frac{1}{2} \begin{pmatrix} \Xi^{\top}\Xi - G^{\top}G & -\Xi^{\top}\Xi - G^{\top}G  \\ -\Xi^{\top}\Xi - G^{\top}G & \Xi^{\top}\Xi - G^{\top}G \end{pmatrix} =  \begin{pmatrix} P & Q \\ Q & P \end{pmatrix}, \]
where the matrices $P$ and $Q$ are given by
\[ P =  I_n + \frac{1}{2}( \Xi^{\top}\Xi - G^{\top}G), \qquad Q =  -\frac{1}{2}(\Xi^{\top}\Xi + G^{\top}G). \]
Therefore,we finally obtain that 
\[ \det(-iH) = \det(P + Q) \det(P-Q) = \det\left(I_n - G^{\top}G\right) \det\left(I_n + \Xi^{\top} \Xi\right). \]
It follows from Lemma \ref{lm:angle} that $\det(I_n - G^{\top}G) = \prod_{\ell = 1}^n \sin^2(\theta_{\ell})$, and from Lemma \ref{lm:det_symp} that $\det(I_n + \Xi^{\top} \Xi) = \prod_{\ell = 1}^n \left(1 + \sin^2(\theta_{\ell})\right)$, where $0 < \theta_1 \leq \ldots \leq \theta_n \leq \pi/2$ are the principal angles between $T_{m_{\nu}} \Gamma_1$ and $T_{m_{\nu}} \Gamma_2$. Consequently,
\[ \det(-iH) = \prod_{\ell = 1}^n \sin^2(\theta_{\ell}) \left(1 + \sin^2(\theta_{\ell})\right) > 0. \]

An application of the stationary phase lemma, as in the proof of Theorem \ref{thm:trace}, yields the following result.

\begin{thm}
\label{thm:trace_square}
We have the following estimate:
\[ \Tr\left((\rho_{k,1}\rho_{k,2})^2\right) = \left(\frac{2\pi}{k}\right)^{2n} \left( \sum_{\nu=1}^s \frac{(\sigma_1,\sigma_2)_{m_{\nu}}^2}{\prod_{\ell = 1}^n \sin(\theta_{\ell}(m_{\nu})) \sqrt{1 + \sin^2(\theta_{\ell}(m_{\nu}))}}   \right) + \bigO{k^{-(2n+1)}}, \]
see Definition \ref{dfn:thetas} and Equation (\ref{eq:constant_sigma}) for notation.
\end{thm}

\subsection{Proof of Theorem \ref{thm:subfid}}

The statement of Theorem \ref{thm:subfid} is a direct consequence of Theorems \ref{thm:trace} and \ref{thm:trace_square}. Recall that $ E(\rho_{k,1},\rho_{k,2}) = \Tr(\rho_{k,1} \rho_{k,2}) + \sqrt{2} \sqrt{\Tr(\rho_{k,1}\rho_{k,2})^2 - \Tr((\rho_{k,1}\rho_{k,2})^2)}$. By Theorem \ref{thm:trace},
\[ \Tr(\rho_{k,1} \rho_{k,2}) = \left(\frac{2\pi}{k}\right)^n \left( \sum_{\nu=1}^s \frac{(\sigma_1,\sigma_2)_{m_{\nu}}}{\prod_{\ell = 1}^n \sin(\theta_{\ell}(m_{\nu}))} \right) + \bigO{k^{-(n+1)}}, \]
which implies that 
\[\Tr(\rho_{k,1} \rho_{k,2})^2 = \left(\frac{2\pi}{k}\right)^{2n} \left( \sum_{\nu=1}^s \frac{(\sigma_1,\sigma_2)_{m_{\nu}}}{\prod_{\ell = 1}^n \sin(\theta_{\ell}(m_{\nu}))} \right)^2 + \bigO{k^{-(2n+1)}}. \]
By expanding the square of the sum as 
\[ \left( \sum_{\nu=1}^s \frac{(\sigma_1,\sigma_2)_{m_{\nu}}}{\prod_{\ell = 1}^n \sin(\theta_{\ell}(m_{\nu}))} \right)^2 =  \sum_{\nu=1}^s \frac{(\sigma_1,\sigma_2)_{m_{\nu}}^2}{\prod_{\ell = 1}^n \sin^2(\theta_{\ell}(m_{\nu}))}  + \sum_{\substack{\nu=1}}^s \sum_{\substack{\mu=1 \\ \mu \neq \nu}}^s \frac{(\sigma_1,\sigma_2)_{m_{\nu}} (\sigma_1,\sigma_2)_{m_{\mu}}}{\prod_{\ell = 1}^n \sin(\theta_{\ell}(m_{\nu})) \sin(\theta_{\ell}(m_{\mu}))}, \]
and by using the result of Theorem \ref{thm:trace_square} and the fact that $\sqrt{u_k + \bigO{k^{-1}}} = \sqrt{u_k} + \bigO{k^{-1}}$ whenever $u_k \geq 0$, we obtain the desired expression.

\subsection{Super-fidelity}

Using the previous results, it is now quite easy to estimate the super-fidelity of the states $\rho_{k,1}$ and $\rho_{k,2}$ attached to $(\Gamma_1,\sigma_1)$ and $(\Gamma_2,\sigma_2)$. We introduce as in the statement of Proposition \ref{prop:purity} the functions $f_j$, $j=1,2$ such that $\sigma_j = f_j \mu_{g,\Sigma_j}$ where $\mu_{g,\Sigma_j}$ is the Riemannian volume on $\Sigma_j$ corresponding to the Riemannian metric induced by $g$ on $\Sigma_j$. 

\begin{thm}
\label{thm:superfid}
The super-fidelity of $\rho_{k,1}$ and $\rho_{k,2}$ satisfies:
\[ G(\rho_{k,1},\rho_{k,2}) = 1 -  \frac{1}{2} \left(\frac{2\pi}{k}\right)^{\frac{n}{2}}  \left( \int_{\Gamma_1} f_1 \sigma_1 + \int_{\Gamma_2} f_2 \sigma_2 \right) + \bigO{k^{-\min\left(n,\frac{n}{2} + 1\right)}}. \]
\end{thm}

\begin{proof}
Recall that $ G(\rho_{k,1},\rho_{k,2}) = \Tr(\rho_{k,1} \rho_{k,2}) + \sqrt{\left(1 - \Tr\left(\rho_{k,1}^2\right) \right) \left( 1 - \Tr\left(\rho_{k,2}^2\right) \right)}$. The first term has been estimated in Theorem \ref{thm:trace}; it is a $\mathcal{O}(k^{-n})$. Moreover, thanks to Proposition \ref{prop:purity}, we know that 
\[ \Tr\left(\rho_{k,j}^2\right) =   \left( \frac{2 \pi}{k} \right)^{\frac{n}{2}} \left( \int_{\Gamma_j} f_j \sigma_j + \bigO{k^{-1}} \right),\]
for $j=1,2$, therefore 
\[ \left(1 - \Tr\left(\rho_{k,1}^2\right) \right) \left( 1 - \Tr\left(\rho_{k,2}^2\right) \right) = 1 - \left( \frac{2 \pi}{k} \right)^{\frac{n}{2}} \left( \int_{\Gamma_1} f_1 \sigma_1 + \int_{\Gamma_2} f_2 \sigma_2 \right) +  \bigO{k^{-\min\left(n,\frac{n}{2} + 1\right)}}. \]
We deduce from this and from $\sqrt{1-x} = 1 - \tfrac{x}{2} + \bigO{x^2}$ that
\[ \sqrt{\left(1 - \Tr\left(\rho_{k,1}^2\right) \right) \left( 1 - \Tr\left(\rho_{k,2}^2\right) \right)} = 1 - \frac{1}{2} \left( \frac{2 \pi}{k} \right)^{\frac{n}{2}} \left( \int_{\Gamma_1} f_1 \sigma_1 + \int_{\Gamma_2} f_2 \sigma_2 \right) +  \bigO{k^{-\min\left(n,\frac{n}{2} + 1\right)}}. \]
\end{proof}

\section{A family of examples on the two-sphere with improved upper bound for fidelity}
\label{sect:examples}

\subsection{Quantization of the sphere}

We consider the sphere $\mathbb{S}^2$ with symplectic form $-\frac{1}{2} \omega_{\mathbb{S}^2} = \frac{1}{2} \sin \varphi \ d\theta \wedge d\varphi$. Its quantization is now quite standard material, hence we only describe it quickly; we work with $\C\P^1$ endowed with the Fubini-Study symplectic form $\omega_{\text{FS}} = \frac{i dz \wedge d\bar{z}}{(1+|z|^2)^2}$, and use the fact that the stereographic projection (from the north pole to the equator) $\pi_N: \S^2 \to \C\P^1$ is a symplectomorphism. On $\C\P^1$, we consider the hyperplane bundle $L = \mathcal{O}(1)$, i.e. the dual of the tautological line bundle $ \mathcal{O}(-1) = \left\{ ([u],v) \in \C\P^1 \times \C^2| \ v \in \C u \right\}$. We endow the latter with its natural holomorphic structure and with the Hermitian form induced by the standard one on the trivial bundle $\C\P^1 \times \C^2$. Then $L$ is equipped with the dual Hermitian form, and its Chern connection $\nabla$ has curvature $-i\omega_{\text{FS}}$, thus $L \to \C\P^1$ is a prequantum line bundle. The following result is well-known (see for instance \cite[Theorem 15.5]{Dem}).

\begin{prop}
\label{prop:iso_homog}
There is a canonical isomorphism between $\Hil_k = H^0(\C\P^1,L^k)$ and the space $\C_k[Z_1,Z_2]$ of homogeneous polynomials of degree $k$ in two complex variables.  
\end{prop}

This isomorphism is constructed by sending a section $s$ of $L^k \to \C\P^1$ to the function $u \in \C^2 \setminus \{0\} \mapsto \langle s(u) | u^{\otimes k} \rangle$, where $\langle \cdot | \cdot \rangle$ stands for the duality pairing between fibers of $\mathcal{O}(k)$ and $\mathcal{O}(-k)$. This isomorphism yields the scalar product
\[ \scal{P}{Q}_k = \int_{\C} \frac{P(1,z) \overline{Q(1,z)}}{(1+|z|^2)^{k+2}} \ |dz \wedge d\bar{z}| \]
on $\C_k[Z_1,Z_2]$, and one readily checks that the monomials
\[ e_{\ell} =  \sqrt{\frac{(k+1)\binom{k}{\ell}}{2\pi}} \  Z_1^{k-\ell} Z_2^{\ell}, \quad 0 \leq \ell \leq k \]
form an orthonormal basis of $\C_k[Z_1,Z_2]$.

Let $U_0 = \{ [z_0:z_1] \in \C\P^1| \ z_0 \neq 0 \}$ be the first standard coordinate chart, endowed with the complex coordinate $z= z_1 / z_0$. Over $U_0$, we define the local non-vanishing section $s_0$ of $\mathcal{O}(-1)$ by $s_0(z) = ([1:z],(1,z))$, and we introduce the dual section $t_0$, i.e. the unique section of $L \to U_0$ such that $t_0(s_0) = 1$. Then the above isomorphism sends $P \in \C_k[Z_1,Z_2]$ to $P(1,z) t_0(z)$, and one readily checks that
\begin{equation} \Pi_k(z,w) = \frac{k+1}{2\pi} (1+z \bar{w})^k \ t_0^k(z) \otimes \overline{t_0}^k(w). \label{eq:proj_sphere}\end{equation}
The local section $u = (1 + |z|^2)^{1/2} \ t_0$ has unit norm, and the coherent vector $\xi_k^{u(z)}$ satisfies
\[ \xi_k^{u(z)}(w) = \frac{k+1}{2\pi} \frac{(1+\bar{z}w)^k}{(1+|z|^2)^{\frac{k}{2}}} \ t_0^k(w), \qquad \left\| \xi_k^{u(z)} \right\|^2_k = \frac{k+1}{2\pi}. \]
Hence, a straightforward computation yields that for $0 \leq \ell \leq k$,
\[ P_k^z e_{\ell} = \frac{z^{\ell} \sqrt{\binom{k}{\ell}}}{(1+|z|^2)^k} (1+\bar{z}w)^k \ t_0^k(w) = \frac{z^{\ell} \sqrt{\binom{k}{\ell}}}{(1+|z|^2)^k} \sum_{m=0}^k \bar{z}^m \sqrt{\binom{k}{m}}  e_m. \]

\subsection{Two orthogonal great circles on the sphere $\mathbb{S}^2$}

We briefly explain the case of orthogonal great circles on $\mathbb{S}^2$. Let $\Gamma_1 = \{x_3=0\}$ and $\Gamma_2 = \{x_1=0\}$, with respective densities $ \sigma_1 = \frac{d\theta}{2\pi}, \sigma_2 = \frac{d\varphi}{2\pi}$. Then $\Gamma_1$ is sent by $\pi_N$ to the unit circle $\{\exp(it)| \ 0 \leq t \leq 2\pi\}$ in $\C$ and $\Gamma_2$ to the line $i\R = \{i y | \ y \in \R \} \subset \C$; moreover,
\[ (\pi_N)_* \sigma_1 = \frac{dt}{2\pi}, \qquad (\pi_N)_*\sigma_2 = \frac{dy}{\pi(1+y^2)}. \] 
Let $\rho_{k,1} = \rho_k(\Gamma_1,\sigma_1)$; by definition,
\[ \scal{\rho_{k,1} e_{\ell}}{e_m} = \int_0^{2\pi} \scal{P_k^{\exp(it)}e_{\ell}}{e_m} \frac{dt}{2\pi} = \frac{1}{2^k} \sqrt{\binom{k}{\ell}\binom{k}{m}} \int_0^{2\pi} \exp(i(\ell-m)t) \frac{dt}{2\pi}; \]
hence we obtain that the matrix of $\rho_{k,1}$ in the orthonormal basis $(e_{\ell})_{0 \leq \ell \leq k}$ reads
\[ \rho_{k,1} = \frac{1}{2^k} \mathrm{diag}\left(\binom{k}{0}, \ldots, \binom{k}{\ell}, \ldots, \binom{k}{k}\right), \]
which means that $\rho_{k,1}$ is prepared according to a binomial probability distribution with respect to this basis. The matrix elements of $\rho_{k,2} = \rho_k(\Gamma_2,\sigma_2)$ are given by the formula
\[ \scal{\rho_{k,2} e_{\ell}}{e_m} = \frac{i^{\ell-m}}{\pi} \sqrt{\binom{k}{\ell}\binom{k}{m}} \int_{-\infty}^{+ \infty} \frac{y^{\ell + m}}{(1+y^2)^{k+1}} \ dy. \]
This integral vanishes when $\ell + m$ is odd, and if $\ell + m = 2p$ is even, it is equal to
\[ I_{k,p} = \int_{-\infty}^{+ \infty} \frac{y^{2p}}{(1+y^2)^{k+1}} \ dy = 2 \int_{0}^{+ \infty} \frac{y^{2p}}{(1+y^2)^{k+1}} \ dy. \]
We can compute this quantity by means of the Beta function, see e.g. \cite[Section 6.2]{AbraSte}.

\begin{lm}
For every $p \in \llbracket 0,k \rrbracket$, $I_{k,p} = \frac{\pi}{4^k} \frac{\binom{2k}{k} \binom{k}{p}}{\binom{2k}{2p}}$.
\end{lm}

Consequently, we obtain that 
\[ \scal{\rho_{k,2} e_{m+2q}}{e_m} = \frac{(-1)^q \binom{2k}{k}}{4^k}  \frac{\binom{k}{m+q} \sqrt{\binom{k}{m+2q}\binom{k}{m}}}{\binom{2k}{2(m+q)}} \]
for $0 \leq m \leq k$ and $\lceil{-m/2} \rceil \leq q \leq \lfloor (k-m)/2 \rfloor$. In particular, 
\[  \scal{\rho_{k,2} e_m}{e_m} = \frac{\binom{2k}{k}}{4^k}  \frac{\binom{k}{m}^2 }{\binom{2k}{2m}} = \frac{1}{4^k} \binom{2m}{m} \binom{2(k-m)}{k-m}. \]
The fact that $\Tr(\rho_{k,2}) = 1$ is then equivalent to the identity 
\[ \sum_{m=0}^k \binom{2m}{m} \binom{2(k-m)}{k-m} = 4^k, \]
which can be derived from the expansion $(1-4x)^{-1/2} = \sum_{r=0}^{+\infty} \binom{2r}{r} x^r$ for every $x$ satisfying $-1/4 < x < 1/4$. Moreover, we obtain that
\begin{equation} \Tr(\rho_{k,1} \rho_{k,2}) = \frac{1}{8^k} \sum_{m=0}^k \binom{k}{m}\binom{2m}{m} \binom{2(k-m)}{k-m}. \label{eq:trace_ortho_theo} \end{equation}
$\Gamma_1$ and $\Gamma_2$ intersect transversally at $m_1 = (0,-1,0)$ and $m_2 = (0,1,0)$. Obviously $\theta_1(m_1) = \theta_1(m_2) = \frac{\pi}{2}$ and one can check that $ (\sigma_1,\sigma_2)_{m_1} = (\sigma_1,\sigma_2)_{m_2} = \frac{1}{2\pi^2}$. Therefore, Theorem \ref{thm:trace} gives $\Tr(\rho_{k,1} \rho_{k,2}) = \frac{2}{k\pi} + \bigO{k^{-2}}$. We check this numerically by plotting $k \Tr(\rho_{k,1} \rho_{k,2})$ as a function of $k$, see Figure \ref{fig:ktrace_ortho} (there most probably exist direct techniques to estimate the sum in Equation (\ref{eq:trace_ortho_theo}), but we are not familiar with them). Furthermore, Theorem \ref{thm:subfid} yields
\begin{equation} E(\rho_{k,1},\rho_{k,2}) = \frac{2}{k\pi} \left( 1 + \sqrt{\frac{2 \sqrt{2}-1}{\sqrt{2}}} \right) + \mathcal{O}(k^{-2}). \label{eq:subfid_ortho_theo} \end{equation}
Figure \ref{fig:sub_fidelity_ortho} displays $E\left(\rho_{k,1}, \rho_{k,2}\right)$ and $k E\left(\rho_{k,1}, \rho_{k,2}\right)$ as functions of $k$.

\subsection{Non necessarily orthogonal great circles}

Let $(\Gamma_1, \sigma_1)$ be as in the previous example. Let $0 < \alpha \leq \pi/2$ and let $\Gamma_2^{\alpha}$ be the great circle given by the equation $x_3 = x_1 \tan \alpha$ (or $x_1 = 0$ if $\alpha = \frac{\pi}{2}$), so that $\Gamma_2^{0} = \Gamma_1$ and $\Gamma_2^{\pi/2} = \Gamma_2$ (see Figure \ref{fig:sphere}). Let $\sigma_2^{\alpha}$ be the density induced on $\Gamma_2$ by $\sigma_1$ \emph{via} the rotation $R_{\alpha}$ of angle $\alpha$ about the $x_2$ axis, which sends $\Gamma_1$ to $\Gamma_2^{\alpha}$. Trying to compute explicitly the matrix elements of $\rho_{k,2}^{\alpha}$ as in the previous part leads to complicated integrals for which we do not know closed forms; therefore numerical evaluation would require to approximate these integrals and would be costly and possibly not very accurate. Instead, we prefer to use the following method, which is more efficient.

\begin{figure}[h]
\begin{center}
\includegraphics[scale=0.6]{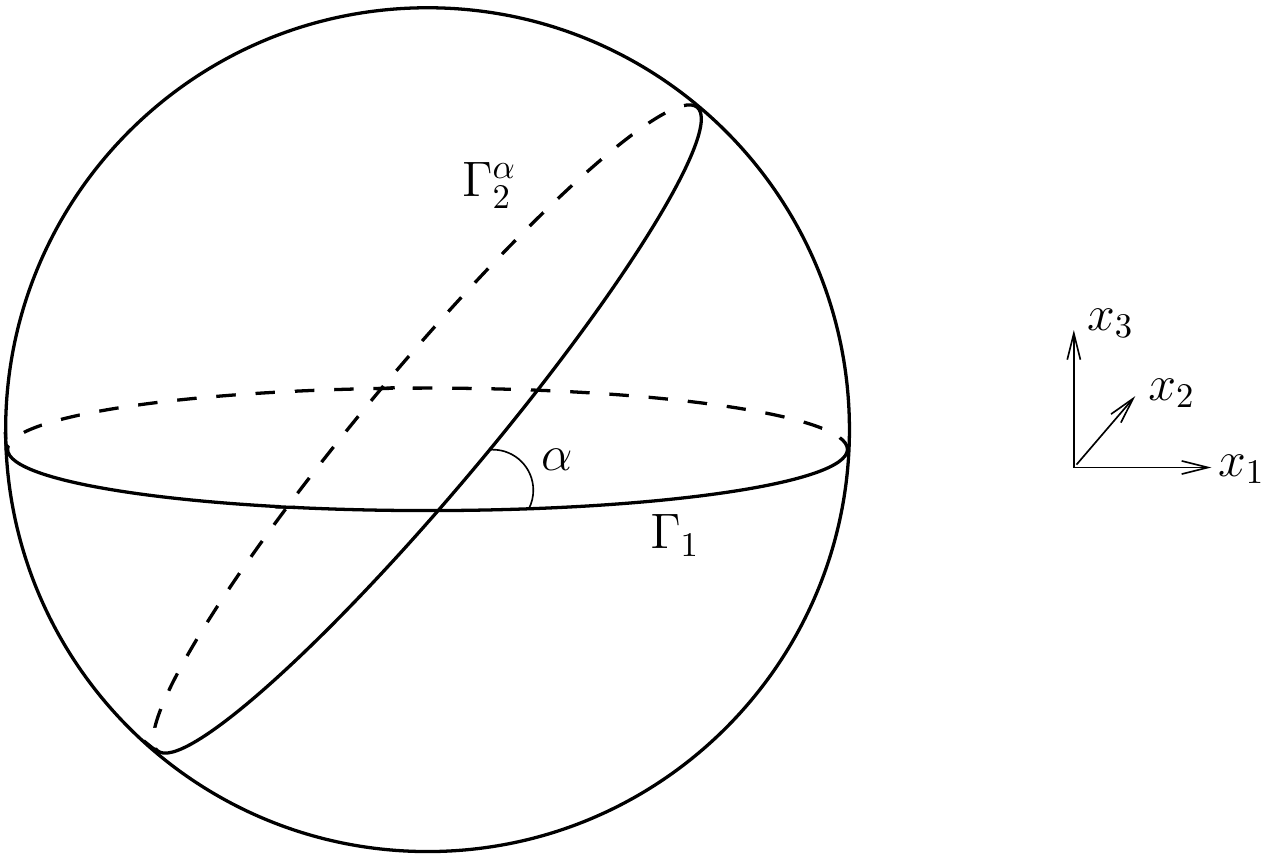}
\end{center}
\caption{The submanifolds $\Gamma_1$ and $\Gamma_2^{\alpha}$.}
\label{fig:sphere}
\end{figure}

Let $\zeta_k: SU(2) \to GL(\C_k[Z_1,Z_2])$ be the natural representation of $SU(2)$ in $\C_k[Z_1,Z_2]$:
\begin{equation} \forall (g,P) \in SU(2) \times \C_k[Z_1,Z_2], \qquad \zeta_k(g)(P) = P \circ g^{-1} \label{eq:rep_su} \end{equation}
where $SU(2)$ acts on $\C^2$ in the standard way. Observe that this representation is unitary with respect to the scalar product on $\C_k[Z_1,Z_2]$ defined above. Note also that we are in the presence of other actions of $SU(2)$: the natural action on $\C^2$, which induces an action on $\C\P^1$ and on its tautological bundle, which itself induces by duality an action on the prequantum line bundle $L \to \C\P^1$, which in turn induces an action on $L^k \to \C\P^1$. Whenever the context allows to distinguish between these actions, we denote by $gu$ the action of $g \in SU(2)$ on $u$ belonging to any of these sets. Furthermore, $SU(2)$ acts on sections of $L^k \to \C\P^1$, by the formula
\[ \forall (g,s) \in SU(2) \times \classe{\infty}{(M,L^k)}, \ \forall m \in \C\P^1, \quad (g s)(m) = g s(g^{-1} m); \]
this yields an action on holomorphic sections. The latter is compatible with $\zeta_k$ through the isomorphism introduced in Proposition \ref{prop:iso_homog}; therefore we will slightly abuse notation by using $(g,\phi) \in SU(2) \times \Hil_k \mapsto \zeta_k(g) \phi $ for this action. We now consider the matrix
\[ \tau_2 = \frac{1}{2} \begin{pmatrix} 0 & -1\\ 1 & 0 \end{pmatrix} \in \mathfrak{su}(2) \simeq \mathfrak{so}(3) \]
which is the infinitesimal generator of rotations about the $x_2$ axis.

\begin{lm}
\label{lm:equiv}
Let $g_{\alpha} = \exp(i \alpha \tau_2) \in SU(2)$ and $U_k(\alpha) = \zeta_k(g_{\alpha})$; then $ \rho_{k,2}^{\alpha} = U_k(\alpha) \rho_{k,1} U_k(\alpha)^*$.
\end{lm}

We believe that this lemma is standard, but nonetheless give a proof in Appendix A. The operator $U_k(\alpha)$ can be computed as follows; let $\zeta_k'$ be the representation of $\mathfrak{su}(2)$ in $\C_k[Z_1,Z_2]$ which is the derived representation of the one given by Equation (\ref{eq:rep_su}):
\[ \forall (\xi,P) \in \mathfrak{su}(2) \times \C_k[Z_1,Z_2], \qquad \zeta_k'(\xi)(P) = \left.\frac{d}{dt}\right|_{t=0} \zeta_k(\exp(t \xi))(P). \]
Then $U_k(\alpha)$ can be computed as $ U_k(\alpha) = \exp(i\alpha \zeta_k'(\tau_2))$. A straightforward computation shows that, for $0 \leq \ell \leq k$, 
\[ \zeta_k'(\tau_2)(e_{\ell}) = \frac{1}{2} \sqrt{(\ell+1)(k-\ell)} \ e_{\ell+1} - \frac{1}{2} \sqrt{\ell(k-\ell+1)} \ e_{\ell-1}. \]
Consequently, we can compute numerically the matrix of $U_k(\alpha)$, and thus the matrix of $\rho_{k,2}^{\alpha}$, in the basis $(e_{\ell})_{0 \leq \ell \leq k}$; therefore we can evaluate the sub-fidelity of $\rho_{k,1}$ and $\rho_{k,2}^{\alpha}$.

Since $\theta_1(m_1) = \theta_1(m_2) = \alpha$ and $ (\sigma_1,\sigma_2^{\alpha})_{m_1} = (\sigma_1,\sigma_2^{\alpha})_{m_2} = \frac{1}{2\pi^2}$, Theorem \ref{thm:trace} yields 
\[ \Tr(\rho_{k,1} \rho_{k,2}^{\alpha}) = \frac{2}{k\pi \sin \alpha} + \bigO{k^{-2}}. \]
We check this numerically for the case $\alpha = \frac{\pi}{4}$, see Figure \ref{fig:ktrace_pisur4}. Moreover, Theorem \ref{thm:subfid}  gives 
\begin{equation} E\left(\rho_{k,1}, \rho_{k,2}^{\alpha}\right) = \frac{2}{k\pi \sin \alpha} \left( 1 + \sqrt{2 - \frac{\sin \alpha}{\sqrt{1 + \sin^2 \alpha}}} \right) + \bigO{k^{-2}}. \label{eq:subfid_pi4} \end{equation}
We check this for the case $\alpha = \frac{\pi}{4}$ in Figure \ref{fig:sub_fidelity_pisur4}, and in Figure \ref{fig:subfid_angle_varying} we compare the value of the sub-fidelity for a fixed large $k$ to its theoretical equivalent as a function of $\alpha$; note that since $k$ is fixed, we cannot take $\alpha$ arbitrarily close to zero.

\subsection{Obtaining a better estimate for fidelity in this example}

It turns out that one can obtain a much better bound for the fidelity of the states $\rho_{k,1}$ and $\rho_{k,2}^{\alpha}$ defined above, by comparing it to the fidelity of certain Berezin-Toeplitz operators. Unfortunately, this strategy relies on a certain number of symmetries and good properties of this particular example, hence it does not work as it is in the general case. Nevertheless, it is quite remarkable that such a good estimate holds, and perhaps some parts of the proof could give insight on how to handle the general case; this is why we will give a detailed explanation of the method, which includes non trivial steps and requires care.

\subsubsection{Comparing both states to Berezin-Toeplitz operators}

We begin by comparing $\rho_{k,1}$ to a certain Berezin-Toeplitz operator. In order to do so, we may give the following heuristic argument: this state is prepared according to a binomial distribution with respect to the orthonormal basis introduced above, with higher weight at basis elements corresponding to points that are close to the equator, where close means at distance of order $k^{-1/2}$. Indeed, it is standard that the binomial coefficients $\binom{k}{\ell}$ that are of the same order as the central binomial coefficient $\binom{k}{\lfloor k/2 \rfloor}$ are such that $|\lfloor k/2 \rfloor-\ell|$ is of order $\sqrt{k}$, and the corresponding basis elements are supported in a neighbourhood of size $k^{-1/2}$ of the equator. Consequently, when $k \to +\infty$, we expect the appearance of the density function of a normal distribution centered at $x_3 = 0$. Therefore, $\rho_{k,1}$ might be related, for $k$ large, to the Berezin-Toeplitz operator $T_k(\lambda \exp(-c k x_3^2))$ for some $c > 0$ and $\lambda \in \R$ (see Equation (\ref{eq:def_BTO}) for the definition of this operator). In fact, for technical reasons that will appear later, we prefer to replace $k$ by $k+1$ in this expression.

In order to be more precise, we argue as follows. The largest matrix element of $\rho_{k,1}$ is
\begin{equation} \frac{1}{2^k} \binom{k}{\lfloor k/2 \rfloor} \underset{k \to +\infty}{\sim} \sqrt\frac{2}{\pi k}. \label{eq:middle_coeff_rho1} \end{equation}
Moreover, the matrix elements of the Berezin-Toeplitz operator associated with a function depending only on $x_3$ can be computed as follows. 

\begin{lm}
\label{lm:BTO_radial}
Let $g \in \classe{\infty}{(\R)}$ and let $f: \S^2 \to \R$ be defined as $f(x_1,x_2,x_3) = g(x_3)$. Then $\scal{T_k(f) e_{\ell}}{e_m}_k = 0$ if $\ell \neq m$ and 
\[ \scal{T_k(f) e_{\ell}}{e_{\ell}}_k = \frac{(k+1) \binom{k}{\ell}}{2^{k+1}} \int_{-1}^1 (1 + x)^{\ell} (1 - x)^{k - \ell} g(x) \ dx. \]
\end{lm}

The proof is more or less a folklore computation; it is available in Appendix A. For $f_k(x_1,x_2,x_3) = \lambda \exp(-c(k+1) x_3^2)$, this gives
\begin{equation} \scal{T_k(f_k) e_{\ell}}{e_{\ell}}_k = \frac{\lambda (k+1) \binom{k}{\ell}}{2^{k+1}} \int_{-1}^1 (1 + x)^{\ell} (1 - x)^{k - \ell} \exp(-c(k+1)x^2) \ dx. \label{eq:coeffs_gaussian} \end{equation}
From this formula, we obtain that the trace
\[ \Tr(T_k(f_k)) = \frac{\lambda (k+1)}{2} \int_{-1}^1 \exp(-c(k+1)x^2) \ dx = \frac{\lambda}{2} \sqrt{\frac{(k+1)\pi}{c}} \mathrm{erf}\left(\sqrt{c(k+1)}\right), \]
where $\mathrm{erf}$ is the error function, is of order $\sqrt{k+1}$. Hence what we really want is to compare $\rho_{k,1}$ to $\tfrac{1}{\sqrt{k+1}} T_k(f_k)$, and we would like that $c$ and $\lambda$ satisfy the relation
\begin{equation} \lambda = 2 \sqrt{\frac{c}{\pi}}, \label{eq:lambda_c_1} \end{equation}
so that the latter has trace close to one. Assume for simplicity that $k$ is even; then
\[ \scal{T_k(f_k) e_{\frac{k}{2}}}{e_{\frac{k}{2}}}_k = \frac{\lambda (k+1) \binom{k}{\frac{k}{2}}}{2^{k+1}} \int_{-1}^1 (1 - x^2)^{\frac{k}{2}} \exp(-c(k+1)x^2) \ dx. \]
We can evaluate the integral by means of Laplace's method; indeed, it is of the form $\int_{-1}^1 \exp(-k\phi(x)) a(x) \ dx$ where $a(x) = \exp(-c x^2)$ and $\phi(x) = c x^2 - \frac{1}{2} \ln(1-x^2)$ for $-1 < x < 1$. We obtain that
\[ \frac{1}{\sqrt{k+1}} \scal{T_k(f_k) e_{\frac{k}{2}}}{e_{\frac{k}{2}}}_k \sim_{k \to +\infty} \frac{\lambda}{ \sqrt{(2 c +1)k}}. \]
Comparing this with Equation (\ref{eq:middle_coeff_rho1}), we see that we want $\lambda$ and $c$ to satisfy the relation 
\begin{equation} \lambda = \sqrt{\frac{2(2c+1)}{\pi}}.  \label{eq:lambda_c_2} \end{equation}
One cannot choose $c$ and $\lambda$ such that both Equations (\ref{eq:lambda_c_1}) and (\ref{eq:lambda_c_2}) are satisfied. In what follows, we will take any $c$ and choose $\lambda$ so that the latter is satisfied. In this case,
\[ \frac{1}{\sqrt{k+1}} \Tr(T_k(f_k)) \sim_{k \to +\infty} \sqrt{1 + \frac{1}{2c}}, \]
and the way to make this quantity become close to one is to let the constant $c$ go to $+\infty$.

This analysis should lead to a good approximation for the coefficients $\scal{\rho_{k,1} e_{\ell}}{e_\ell}_k$ where $|\ell - \tfrac{k}{2}|$ is of order $\sqrt{k}$, but there is no reason to expect this approximation to still be good for the other coefficients. Nevertheless, the following nice property holds.

\begin{lm}
\label{lm:rho_approx}
For every $c \geq 2$ and every $k \geq 1$, we have that 
\begin{equation} \rho_{k,1} \leq \frac{1}{\sqrt{k+1}} T_k(f_k^c) \label{eq:bound_rho1} \end{equation}
where $f_k^c: \S^2 \to \R^+$ is given by the formula
\begin{equation} f_k^c(x_1, x_2, x_3) = \sqrt{\frac{2(2c+1)}{\pi}} \exp\left( -c (k +1 ) x_3^2 \right). \label{eq:def_fk} \end{equation}
\end{lm}

\begin{proof}
Since both operators are diagonal in the basis $(e_{\ell})_{0 \leq \ell \leq k}$, we only need to compare their respective coefficients. Since $\scal{\rho_{k,1} e_{\ell}}{e_{\ell}}_k = 2^{-k} \binom{k}{\ell}$ and in view of Equation (\ref{eq:coeffs_gaussian}), this requires to check that the inequality 
\[  \sqrt{\frac{(k+1)(2c+1)}{2\pi}} \int_{-1}^1 (1 + x)^{\ell} (1 - x)^{k-\ell} \exp\left( - c(k+1) x^2 \right) \ dx \geq 1 \]
holds. Let us assume for the sake of simplicity that $k$ is even, the odd case being similar. One readily checks that the above integral is minimal for $\ell = \tfrac{k}{2}$. Hence we need to study
\[ I_k(c) = \sqrt{2c+1} \int_{-1}^1 (1 - x^2)^{\frac{k}{2}} \exp\left( - c(k+1) x^2 \right) \ dx. \]
One can check that $I_k$ is decreasing in $c \in [2, +\infty)$, and setting $y = \sqrt{c} \ x$ yields
\[ I_k(c) = \sqrt{2+\frac{1}{c}} \int_{-c}^c \left(1 - \frac{y^2}{c}\right)^{\frac{k}{2}} \exp\left( - (k+1) y^2 \right) \ dy  \underset{c \to +\infty}{\longrightarrow} \sqrt{2} \int_{\R} \exp\left( - (k+1) y^2 \right) \ dy = \sqrt{\frac{2\pi}{k+1}}. \]
Thus for every $c \geq 2$, $I_k(c) \geq \sqrt{\frac{2\pi}{k+1}}$, which implies the above inequality.

\end{proof}

The next step is to observe that there is an exact version of Egorov's theorem for rotations on $\S^2$. This is well-known, but we give a simple proof using our notation, in Appendix A, for the sake of completeness.

\begin{prop}
\label{prop:exact_egorov}
Let $f \in \classe{\infty}{(M)}$, $k \geq 1$ and $\beta \in [0,2\pi]$. Let $U_k(\beta) = \zeta_k(g_{\beta})$ as above; then $U_k(\beta) T_k(f) U_k(\beta)^* = T_k(f \circ R_{-\beta})$, where we recall that $R_{\gamma}$ is the rotation of angle $\gamma$ about the $x_2$ axis.
\end{prop}

Since conjugation by a unitary operator preserves the order, this implies that 
\begin{equation} \rho_{k,2}^{\alpha} \leq \frac{1}{\sqrt{k+1}} T_k(f_k^c \circ R_{-\alpha}) \label{eq:bound_rho2}  \end{equation}
for every $c > 0$, with $f_k^c$ as above. This allows us to obtain the following upper bound.

\begin{prop}
\label{prop:bound_fid_rho_BTO}
The fidelity of $\rho_{k,1}$ and $\rho_{k,2}^{\alpha}$ satisfies
\[ F(\rho_{k,1}, \rho_{k,2}^{\alpha}) \leq \frac{1}{k+1} F\left(T_k(f_k^c), T_k(f_k^c \circ R_{-\alpha})\right), \] 
for every $c \geq 2$ and $k \geq 1$, where $f_k^c$ is the function defined in Equation (\ref{eq:def_fk}).
\end{prop}

\begin{proof}
This immediately follows from Equations (\ref{eq:bound_rho1}) and (\ref{eq:bound_rho2}) and from the monotonicity of the fidelity, see for instance \cite{Mol}: if $A, B, C$ are positive semidefinite Hermitian operators with $A \leq B$, then $F(A,C) \leq F(B,C)$.
\end{proof}

\subsubsection{Estimating the new fidelity function}

As a consequence of the previous result, if we manage to show that the fidelity of $T_k(f_k^c)$ and $T_k(f_k^c \circ R_{-\alpha})$ is of order $\bigO{1}$, we will know that the fidelity of $\rho_{k,1}$ and $\rho_{k,2}^{\alpha}$ is a $\bigO{k^{-1}}$. In Figures \ref{fig:comp_fid_BTO_pi2} and \ref{fig:comp_fid_BTO_pi4}, we compare $F(T_k(f_k^c), T_k(f_k^c \circ R_{-\alpha}))$ with the rescaled fidelity $k F(\rho_{k,1}, \rho_{k,2}^{\alpha})$ for $\alpha = \frac{\pi}{2}$ and $\alpha = \frac{\pi}{4}$, and different values of $c$. We observe on these numerical simulations that for large $c$, the above inequality seems to give an excellent approximation for $F(\rho_{k,1}, \rho_{k,2}^{\alpha})$. We will now try to use this fact to obtain a good upper bound on this fidelity. 

\paragraph{Change of scale.} In order to estimate $F(T_k(f_k^c), T_k(f_k^c \circ R_{-\alpha}))$, a natural idea is to try to approximate the operator $ \sqrt{T_k\left(f_k^c\right)}  \sqrt{T_k\left(g_k^c\right)}$ involved in the definition of this fidelity by another Berezin-Toeplitz operator. For instance, it is tempting to conjecture that the square root of $T_k(f_k^c)$ coincides with $T_k(\sqrt{f_k^c})$ up to some small remainder, but one cannot apply the usual symbolic calculus for Berezin-Toeplitz operators here, because $f_k^c$ does not belong to any reasonable symbol class. Indeed, it is of the form $f(k^{1/2} \cdot)$ for some $f$ independent of $k$, and $1/2$ is precisely the critical exponent; the product rule with sharp remainder for Berezin-Toeplitz operators \cite[Equation (P3)]{ChaPolsharp} reads, for functions of the form $f_k = f(k^{\varepsilon} \cdot)$ and $g_k = g(k^{\varepsilon} \cdot)$ with $f$ and $g$ of unit uniform norm,
\[ \| T_k(f_k) T_k(g_k) - T_k(f_k g_k) \| \leq \gamma k^{-1+2\varepsilon} \]
for some constant $\gamma > 0$. Hence the remainder is indeed small if and only if $\varepsilon < 1/2$.

In order to overcome this difficulty, the idea is to replace this power $1/2$ by $1/2 - \delta$ for some $\delta > 0$. More precisely, let 
\[ f: \S^2 \to \R^+, \quad (x_1, x_2, x_3) \mapsto \sqrt{\frac{2(2c+1)}{\pi}} \exp\left( -x_3^2 \right), \]
so that $f_k^c = f\left(\sqrt{c(k+1)} \ \cdot \right)$, and given $0 < \delta < 1/2$, let $f_k^{c,\delta} = f(\sqrt{c}(k+1)^{\frac{1}{2} - \delta} \cdot)$. In order to simplify notation, we also introduce the function $g = f \circ R_{-\alpha}$, so that $g_k^c := f_k^c \circ R_{-\alpha} = g\left(\sqrt{c(k+1)} \ \cdot \right)$, and define $g_k^{c,\delta}$ in the same way. Then $f_k^c \leq f_k^{c,\delta}$, hence we obtain with the same arguments as above that 
\begin{equation} F\left(T_k(f_k^c), T_k(f_k^c \circ R_{-\alpha})\right) \leq F\left(T_k(f_k^{c,\delta}), T_k(g_k^{c,\delta}) \right) =  \left\| \sqrt{T_k\left(f_k^{c,\delta}\right)}  \sqrt{T_k\left(g_k^{c,\delta}\right)}  \right\|_{\Tr}^2. \label{eq:bound_fid_delta} \end{equation}
What we have gained is that we can use the product rule for the operators on the right-hand side of this inequality to replace $\sqrt{T_k(f_k^{c,\delta})} \sqrt{T_k(g_k^{c,\delta})}$ by $T_k\left(\sqrt{f_k^{c,\delta} g_k^{c,\delta}}\right)$, and the trace norm of the latter is easy to compute, as a simple application of the stationary phase lemma, with details available in Appendix B.

\begin{prop}
\label{prop:trace_norm_BTO}
For every $\delta \in (0,\frac{1}{2})$,
\[ \left\| T_k\left(\sqrt{f_k^{c,\delta} g_k^{c,\delta}}\right) \right\|_{\Tr} =  \frac{2 k^{2\delta}}{\sqrt{c \pi} \sin \alpha} + \bigO{k^{4\delta-1} c^{-\frac{3}{2}}}. \]
\end{prop}

Note that when $c$ is of order $k^{4 \delta}$, this trace norm is a $\bigO{1}$; however, we will see below that we cannot consider such a $c$.

\paragraph{Control of the remainders.} The tricky part is to understand the structure of the remainders appearing when replacing $\sqrt{T_k(f_k^{c,\delta})} \sqrt{T_k(g_k^{c,\delta})}$ by $T_k\left(\sqrt{f_k^{c,\delta} g_k^{c,\delta}}\right)$. By the product rule \cite[Equation (P3)]{ChaPolsharp}, there exists $\gamma > 0$ such that $T_k\left( \sqrt{f_k^{c,\delta}}  \right)^2 = T_k\left(f_k^{c,\delta}\right) + A_k$ where $\|  A_k \| \leq \gamma c^{\frac{3}{2}} k^{-2\delta}$. Since the square root is operator monotone, this yields \cite{And} $T_k\left( \sqrt{f_k^{c,\delta}}  \right)  =  \sqrt{T_k\left(f_k^{c,\delta}\right)} + R_k$ where $\|  R_k \| \leq \gamma_1 c^{\frac{3}{4}} k^{-\delta}$. Since $\sqrt{T_k\left(f_k^{c,\delta}\right)}$ and $T_k\left( \sqrt{f_k^{c,\delta}}  \right)$
have norm smaller than some constant times $c^{\frac{1}{4}}$, $\|  R_k \| \leq \gamma_2 \min(c^{\frac{1}{4}}, c^{\frac{3}{4}} k^{-\delta})$. By applying the exact version of Egorov's theorem stated earlier, we deduce from this that $ \sqrt{T_k\left(g_k^{c,\delta}\right)} = T_k\left( \sqrt{g_k^{c,\delta}}  \right)  + S_k $
with $\|  S_k \| \leq  \gamma_2 \min(c^{\frac{1}{4}}, c^{\frac{3}{4}} k^{-\delta})$. Now, the triangle inequality for the trace norm reads
\begin{equation} \begin{split} \left\| \sqrt{T_k\left(f_k^{c,\delta}\right)}  \sqrt{T_k\left(g_k^{c,\delta}\right)}  \right\|_{\Tr} \leq  \left\|T_k\left( \sqrt{f_k^{c,\delta}}  \right)  T_k\left( \sqrt{g_k^{c,\delta}}  \right)  \right\|_{\Tr} + \left\|T_k\left( \sqrt{f_k^{c,\delta}}  \right)  S_k  \right\|_{\Tr} \\ + \left\|R_k  T_k\left( \sqrt{g_k^{c,\delta}}  \right)  \right\|_{\Tr} + \left\|R_k S_k  \right\|_{\Tr}. \end{split} \label{eq:estimate_tracenorm}\end{equation}
We start by estimating the last three terms on the right-hand side of this equation. This is in fact delicate, since we want to discriminate between what happens near the intersection points of $\Gamma_1$ and $\Gamma_2^{\alpha}$ and what happens away of these points. In order to do so, we consider a cutoff function $\chi \in \classe{\infty}{(\R,\R^+)}$ smaller than one, equal to one on $[-1/2,1/2]$ and vanishing outside $(-1,1)$, and we define for $r > 1$
\[\chi_k^{r,\delta}: \S^2 \to \R, \quad (x_1, x_2, x_3) \mapsto \chi(r k^{\frac{1}{2} - \delta} x_3) \chi(r k^{\frac{1}{2} - \delta} x_3 \circ R_{-\alpha}), \]
so that $\chi_k^{r,\delta}$ vanishes outside the union of two ``parallelograms'' centered at each of these intersection points and with side length of order $r^{-1} k^{\delta - \frac{1}{2}}$. Writing $1 = \chi_k^{r,\delta} + 1 - \chi_k^{r,\delta}$ and using the triangle inequality, we obtain that
\begin{equation} \left\|T_k\left( \sqrt{f_k^{c,\delta}}  \right)  S_k  \right\|_{\Tr} \leq \left\|T_k\left( \chi_k^{r,\delta} \sqrt{f_k^{c,\delta}}  \right)  S_k  \right\|_{\Tr} + \left\|T_k\left( (1 - \chi_k^{r,\delta}) \sqrt{f_k^{c,\delta}}  \right)  S_k  \right\|_{\Tr}. \label{eq:first_term} \end{equation}
Regarding the first term, H\"older's inequality for Schatten norms yields
\[ \left\|T_k\left( \chi_k^{r,\delta} \sqrt{f_k^{c,\delta}}  \right)  S_k  \right\|_{\Tr} \leq \left\|T_k\left( \chi_k^{r,\delta} \sqrt{f_k^{c,\delta}}  \right)  \right\|_{\Tr} \left\| S_k  \right\| = \Tr\left( T_k\left( \chi_k^{r,\delta} \sqrt{f_k^{c,\delta}}  \right)  \right) \left\| S_k  \right\|,  \]
where the last equality comes from the fact that $T_k\left( \chi_k^{r,\delta} \sqrt{f_k^{c,\delta}}  \right) \geq 0$ since $\chi_k^{r,\delta} \sqrt{f_k^{c,\delta}} $ takes its values in $\R^+$. The trace of this operator satisfies
\[ \Tr\left( T_k\left( \chi_k^{r,\delta} \sqrt{f_k^{c,\delta}}  \right)  \right) = \frac{k+1}{2\pi} \int_{\S^2} \chi_k^{r,\delta} \sqrt{f_k^{c,\delta}} \ d\mu \leq \frac{k+1}{2\pi}  \left(\frac{2(2c+1)}{\pi}\right)^{\frac{1}{4}} \int_{\S^2} \chi_k^{r,\delta} \ d\mu, \]
hence it is a $\bigO{k^{2\delta} r^{-2} c^{\frac{1}{4}}}$, since the area of each of the aforementioned parallelograms is of order $r^{-2} k^{2\delta - 1}$. Consequently,
\[ \left\|T_k\left( \chi_k^{r,\delta} \sqrt{f_k^{c,\delta}}  \right)  S_k  \right\|_{\Tr} = \bigO{k^{2\delta} r^{-2}\min(c^{\frac{1}{2}}, ck^{-\delta})}.  \]
In order to estimate the second term on the right-hand side of Equation (\ref{eq:first_term}), we use once again H\"older's inequality to derive
\[  \left\|T_k\left( (1 - \chi_k^{r,\delta}) \sqrt{f_k^{c,\delta}}  \right)  S_k  \right\|_{\Tr} \leq  \left\|T_k\left( (1 - \chi_k^{r,\delta}) \sqrt{f_k^{c,\delta}}  \right)  \right\| \Tr(S_k). \]
We have that 
\[ \left\|T_k\left( (1 - \chi_k^{r,\delta}) \sqrt{f_k^{c,\delta}}  \right)  \right\| \leq \left\| (1 - \chi_k^{r,\delta}) \sqrt{f_k^{c,\delta}} \right\|_{\infty} = \bigO{c^{\frac{1}{4}} \exp(-r^{-2} c)}. \]
Since moreover $\Tr(S_k) \leq \dim(\Hil_k) \| S_k \|$, we obtain that 
\[  \left\|T_k\left( (1 - \chi_k^{r,\delta}) \sqrt{f_k^{c,\delta}}  \right)  S_k  \right\|_{\Tr} = \bigO{k \exp(-r^{-2} c) \min(c^{\frac{1}{2}}, c k^{-\delta})}, \]
and finally, we deduce from Equation (\ref{eq:first_term}) that $\left\|T_k\left( \sqrt{f_k^{c,\delta}}  \right)  S_k  \right\|_{\Tr} = \bigO{\varepsilon(c,r,k)}$ where
\begin{equation} \varepsilon(c,r,k) = \max\left(k^{2\delta} r^{-2}, k \exp(-r^{-2} c)\right) \min\left(c^{\frac{1}{2}}, c k^{-\delta}\right). \label{eq:epsilon}\end{equation}

The trace norm of $R_k  T_k\left( \sqrt{g_k^{c,\delta}}  \right)$ can be estimated in a similar way. It remains to control the trace norm of $R_k S_k$; we do not expect this term to be small. However, we can say the following: from Lemma \ref{lm:BTO_radial}, we know that both $T_k\left( \sqrt{f_k^{c,\delta}} \right)$ and $ \sqrt{T_k\left(f_k^{c,\delta}\right)}$ are diagonal in the basis $(e_{\ell})_{0 \leq \ell \leq k}$, hence $R_k$ also is. Since moreover $R_k \leq T_k\left( \sqrt{f_k^{c,\delta}} \right)$, we conclude that $R_k^2 \leq T_k\left( \sqrt{f_k^{c,\delta}} \right)^2$. Thus, it follows from Proposition \ref{prop:exact_egorov} that $S_k^2 \leq T_k\left( \sqrt{g_k^{c,\delta}} \right)^2$ as well. Therefore, the monotonicity of the fidelity function yields
\[ \left\|R_k S_k  \right\|_{\Tr} = F(R_k^2,S_k^2) \leq F\left(T_k\left( \sqrt{f_k^{c,\delta}} \right)^2, T_k\left( \sqrt{g_k^{c,\delta}}\right)^2 \right) = \left\|T_k\left( \sqrt{f_k^{c,\delta}} \right) T_k\left( \sqrt{g_k^{c,\delta}} \right) \right\|_{\Tr}. \]
Using all of the above estimates in Equation (\ref{eq:estimate_tracenorm}), we finally obtain that 
\begin{equation} \left\| \sqrt{T_k\left(f_k^{c,\delta}\right)}  \sqrt{T_k\left(g_k^{c,\delta}\right)}  \right\|_{\Tr} \leq  2 \left\|T_k\left( \sqrt{f_k^{c,\delta}}  \right)  T_k\left( \sqrt{g_k^{c,\delta}} \right) \right\|_{\Tr}  + \bigO{\varepsilon(c,r,k)}, \label{eq:second_estimate_tracenorm}\end{equation}
see Equation (\ref{eq:epsilon}). It remains to control the remainders which appear when we replace $\left\|T_k\left( \sqrt{f_k^{c,\delta}}  \right)  T_k\left( \sqrt{g_k^{c,\delta}}  \right)  \right\|_{\Tr}$ by $\left\|T_k\left( \sqrt{f_k^{c,\delta} g_k^{c,\delta}}  \right)  \right\|_{\Tr} $. We claim that we can argue as before to control them, thanks to the cutoff function $\chi_k^{r,\delta}$; indeed, the uniform norm of the function $(1 - \chi_k^{r,\delta})\sqrt{f_k^{c,\delta} g_k^{c,\delta}}$ is also bounded by some constant times $\exp(- \lambda r^{-2} c)$ where $\lambda > 0$ does not depend on $c, r, k$. Hence we get the estimate 
\[  \left\| \sqrt{T_k\left(f_k^{c,\delta}\right)}  \sqrt{T_k\left(g_k^{c,\delta}\right)}  \right\|_{\Tr} \leq  2 \left\|T_k\left( \sqrt{f_k^{c,\delta} g_k^{c,\delta}}  \right)  \right\|_{\Tr} + \bigO{\varepsilon(c,r,k)}, \]
which yields the following result.

\begin{thm}
\label{thm:fid_sphere}
The fidelity of $\rho_{k,1}$ and $\rho_{k,2}^{\alpha}$ satisfies, for every $\delta \in (0,\frac{1}{2}]$,
\[ F(\rho_{k,1},\rho_{k,2}^{\alpha}) \leq \frac{16 k^{3\delta - 1}}{\pi \sin^2 \alpha} + \bigO{k^{\frac{25 \delta}{12} - 1}}.  \]
\end{thm}

\begin{proof}
The above inequality and Proposition \ref{prop:trace_norm_BTO} yield
\[  \left\| \sqrt{T_k\left(f_k^{c,\delta}\right)}  \sqrt{T_k\left(g_k^{c,\delta}\right)}  \right\|_{\Tr} \leq \frac{4 k^{2\delta}}{\sqrt{c \pi} \sin \alpha} + \nu(c,r,k)  \]
where $\nu(c,r,k) = \bigO{k^{4\delta-1} c^{-\frac{3}{2}}} + \bigO{\varepsilon(c,r,k)}$. We would like to take $c = k^{4 \delta}$ so that the first term is a $\bigO{1}$; but then $\nu(c,r,k)$ would be a $\bigO{\max(k^{4\delta} r^{-2}, k^{1 + 2\delta} \exp(- \lambda k^{4 \delta} r^{-2} ))}$, which can not be made into a $o(1)$ no matter which $r$ we choose. So instead we choose $c = k^{\delta}$, so that the first term is a $\bigO{k^{\frac{3\delta}{2}}}$ and 
\[ \nu(c,r,k) = \bigO{k^{\frac{5\delta}{2}-1}} + \bigO{\max(k^{2\delta} r^{-2}, k \exp(- \lambda  r^{-2} k^{\delta})) }. \] 
We want the term in the exponential to be of order $k^{\varepsilon}$ for some $\varepsilon > 0$, and at the same time that $k^{2\delta} r^{-2} = o(k^{\frac{3\delta}{2}})$. In order to do so, we can choose for instance $r = k^{\frac{\delta}{3}}$; then 
\[ \nu(c,r,k) = \bigO{k^{\frac{5\delta}{2}-1}}  + \bigO{\max(k^{\frac{4\delta}{3}}, k \exp(- \lambda  k^{\frac{\delta}{3}})) } = \bigO{k^{\frac{4\delta}{3}}}. \]
Thus, for these choices, we obtain that 
\[  \left\| \sqrt{T_k\left(f_k^{c,\delta}\right)}  \sqrt{T_k\left(g_k^{c,\delta}\right)}  \right\|_{\Tr} \leq \frac{4 k^{\frac{3\delta}{2}}}{\sqrt{\pi} \sin \alpha} + \bigO{k^{\frac{4\delta}{3}}}. \]
Indeed, $\frac{5 \delta}{2} - 1 < \frac{4 \delta}{3}$ since $0 < \delta \leq 1/2$. Consequently, we deduce from Equation (\ref{eq:bound_fid_delta}) that
\[ F\left(T_k(f_k^c), T_k(f_k^c \circ R_{-\alpha})\right) \leq \frac{16 k^{3 \delta}}{\pi \sin^2 \alpha} + \bigO{k^{\frac{25\delta}{12}}} \]
for such $c$, and we use Proposition \ref{prop:bound_fid_rho_BTO} to conclude.
\end{proof}

We conjecture that the constant appearing in this result is not so bad, i.e. that, in fact, this fidelity has an equivalent of the form $F(\rho_{k,1},\rho_{k,2}^{\alpha}) \sim \frac{C}{k \sin^2 \alpha}$ for some constant $C > 0$ when $k$ goes to infinity. We investigate this conjecture in Figure \ref{fig:fid_angle_varying}, where we display the (rescaled) fidelity of $\rho_{k,1}$ and $\rho_{k,2}^{\alpha}$ for some fixed large $k$, as a function of the angle $\alpha$. From this figure, we guess that our conjecture may be true up to allowing that $C = C(\alpha)$ is a function of $\alpha$ taking its values in a small interval.

We display the fidelity of $\rho_{k,1}$ and $\rho_{k,2}^{\alpha}$ together with their sub-fidelity, as functions of $k$, in Figures \ref{fig:fid_vs_subfid_pi2} (where $\alpha = \frac{\pi}{2}$) and \ref{fig:fid_vs_subfid_pi3} (where $\alpha = \frac{\pi}{3}$).

\section{Numerics and a conjecture}
\label{sect:numerics}

\subsection{Numerical computations}

We gather here the outcome of numerical simulations for our examples on $\S^2$.

\begin{figure}[H]
\begin{center}
\includegraphics[scale=0.35]{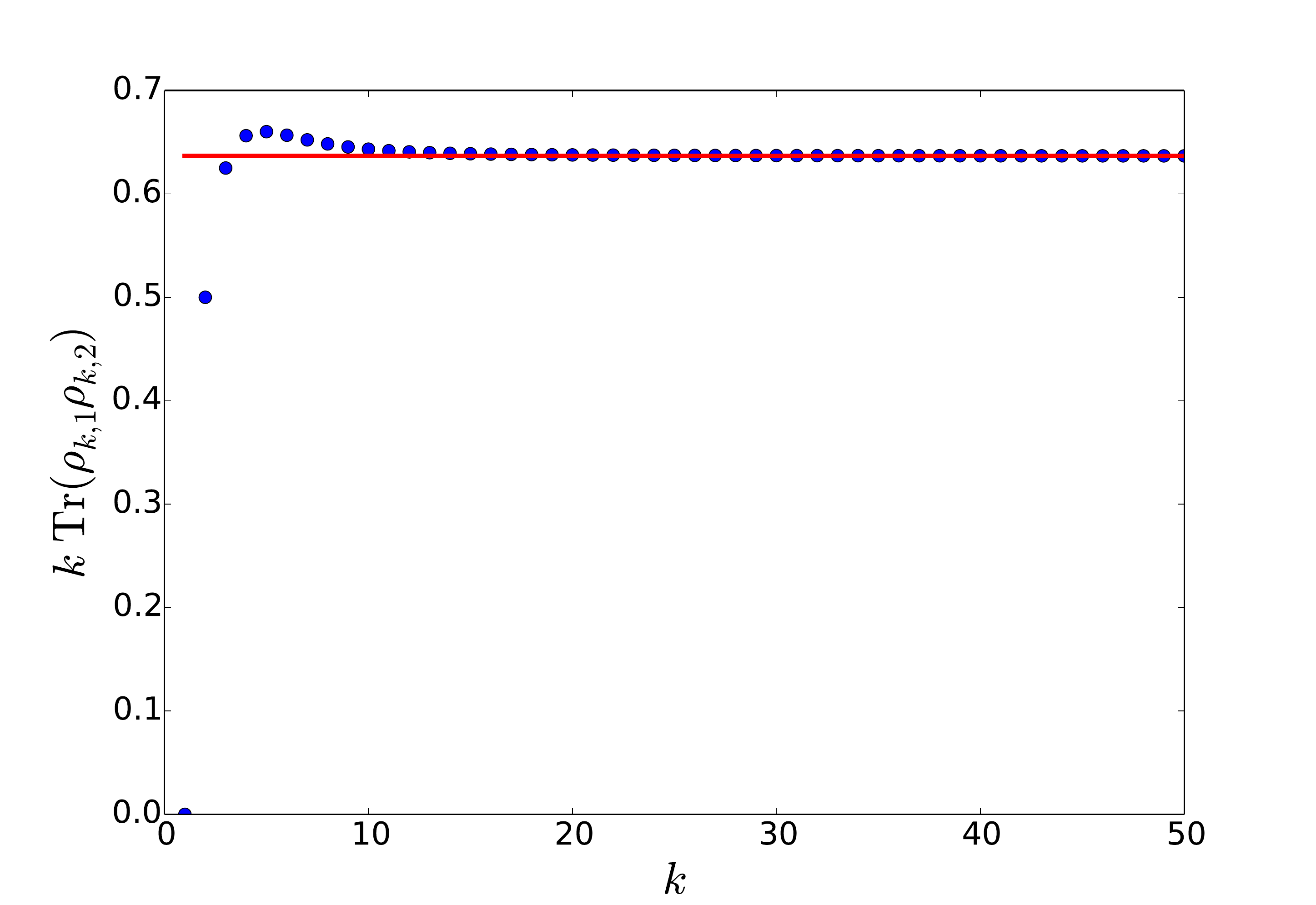}
\end{center}
\caption{The blue circles represent $k \Tr(\rho_{k,1} \rho_{k,2})$ as a function of $k$, for $1 \leq k \leq 50$, computed numerically from Equation (\ref{eq:trace_ortho_theo}). The red line is the theoretical limit $\frac{2}{\pi}$.}
\label{fig:ktrace_ortho}
\end{figure}

\begin{figure}[H]
\hspace{-5mm}
\subfigure[The blue crosses represent $E(\rho_{k,1},\rho_{k,2})$, while the red circles stand for the first term  on the right-hand side of Equation (\ref{eq:subfid_ortho_theo}).]{\includegraphics[scale=0.29]{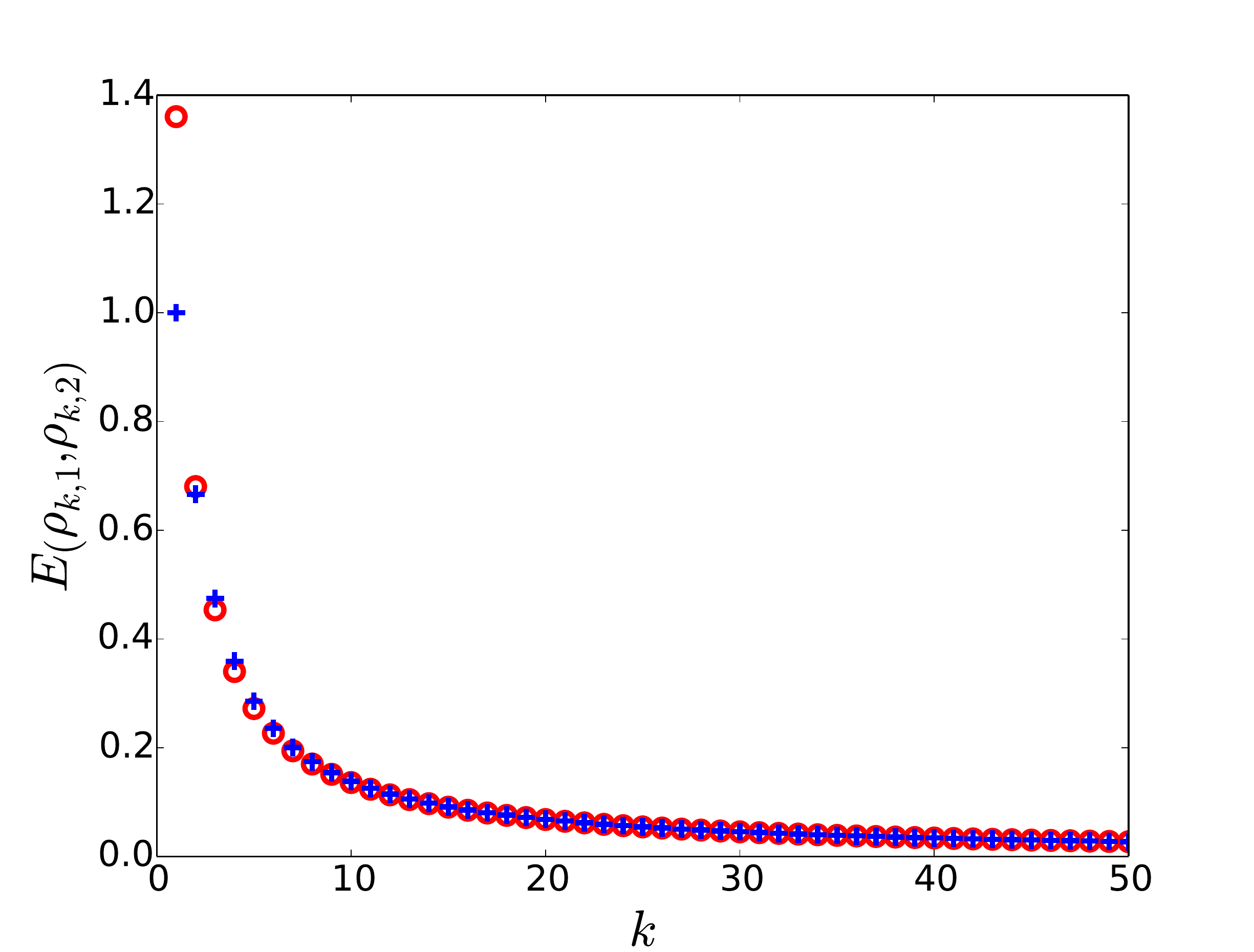} } 
\hspace{5mm}
\subfigure[The blue crosses represent $k E(\rho_{k,1},\rho_{k,2})$, and the red line corresponds to the constant $\frac{2}{\pi} \left( 1 + \sqrt{\frac{2 \sqrt{2}-1}{\sqrt{2}}} \right)$.]{\includegraphics[scale=0.29]{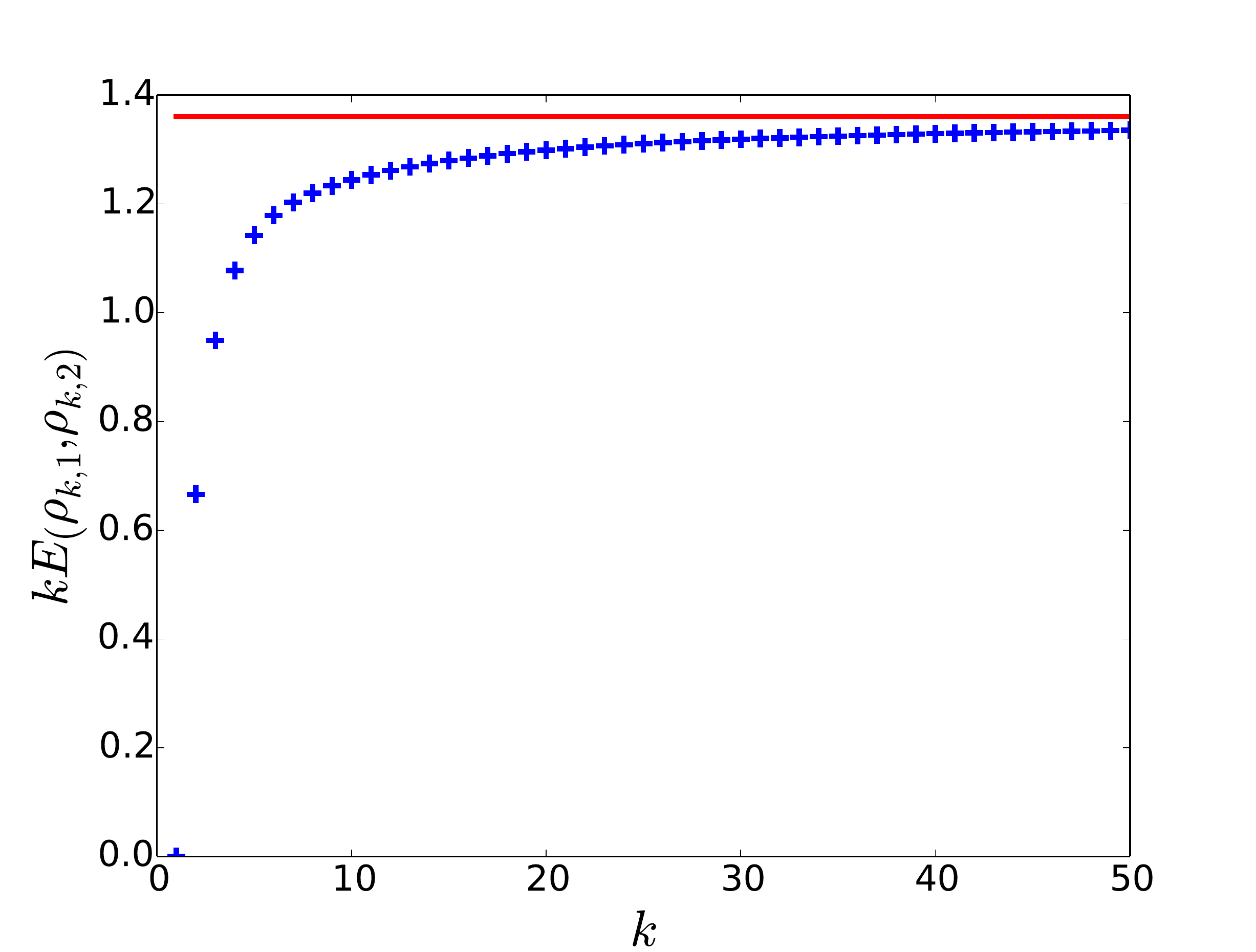} }
\caption{Sub-fidelity $E(\rho_{k,1},\rho_{k,2})$ and $k E(\rho_{k,1},\rho_{k,2})$, as functions of $k$, for $1 \leq k \leq 50$.}
\label{fig:sub_fidelity_ortho}
\end{figure}

\begin{figure}[H]
\begin{center}
\includegraphics[scale=0.35]{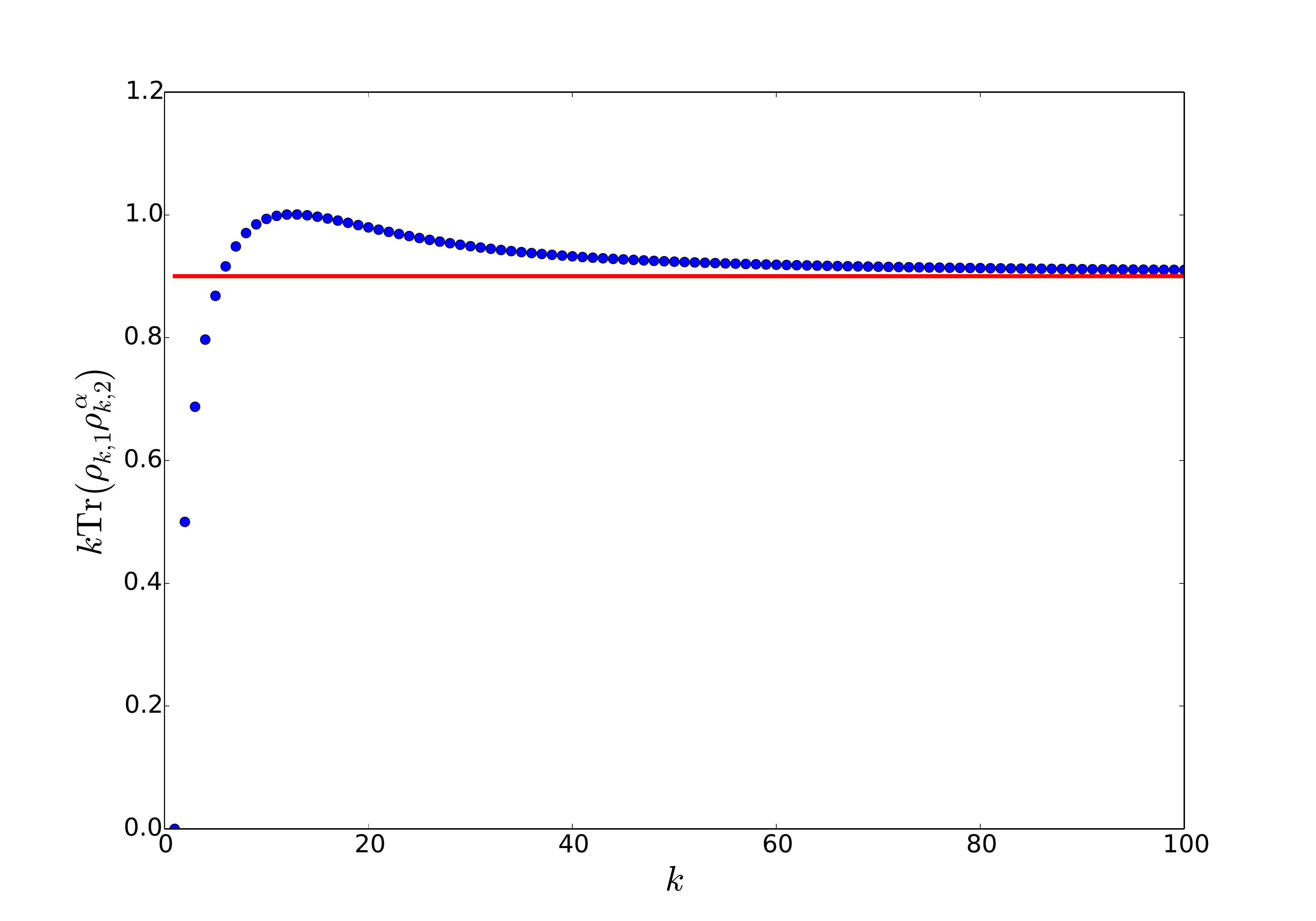}
\end{center}
\caption{The blue circles represent $k \Tr(\rho_{k,1} \rho_{k,2}^{\alpha})$ as a function of $k$ for $\alpha = \frac{\pi}{4}$, $1 \leq k \leq 100$. The red line corresponds to the theoretical limit $\frac{2}{\pi \sin \alpha} = \frac{2 \sqrt{2}}{\pi}$.}
\label{fig:ktrace_pisur4}
\end{figure}

\begin{figure}[H]
\hspace{-5mm}
\subfigure[The blue crosses correspond to $E(\rho_{k,1},\rho_{k,2}^{\alpha})$, and the red circles correspond to the first term on the right-hand side of Equation (\ref{eq:subfid_pi4}).]{\includegraphics[scale=0.3]{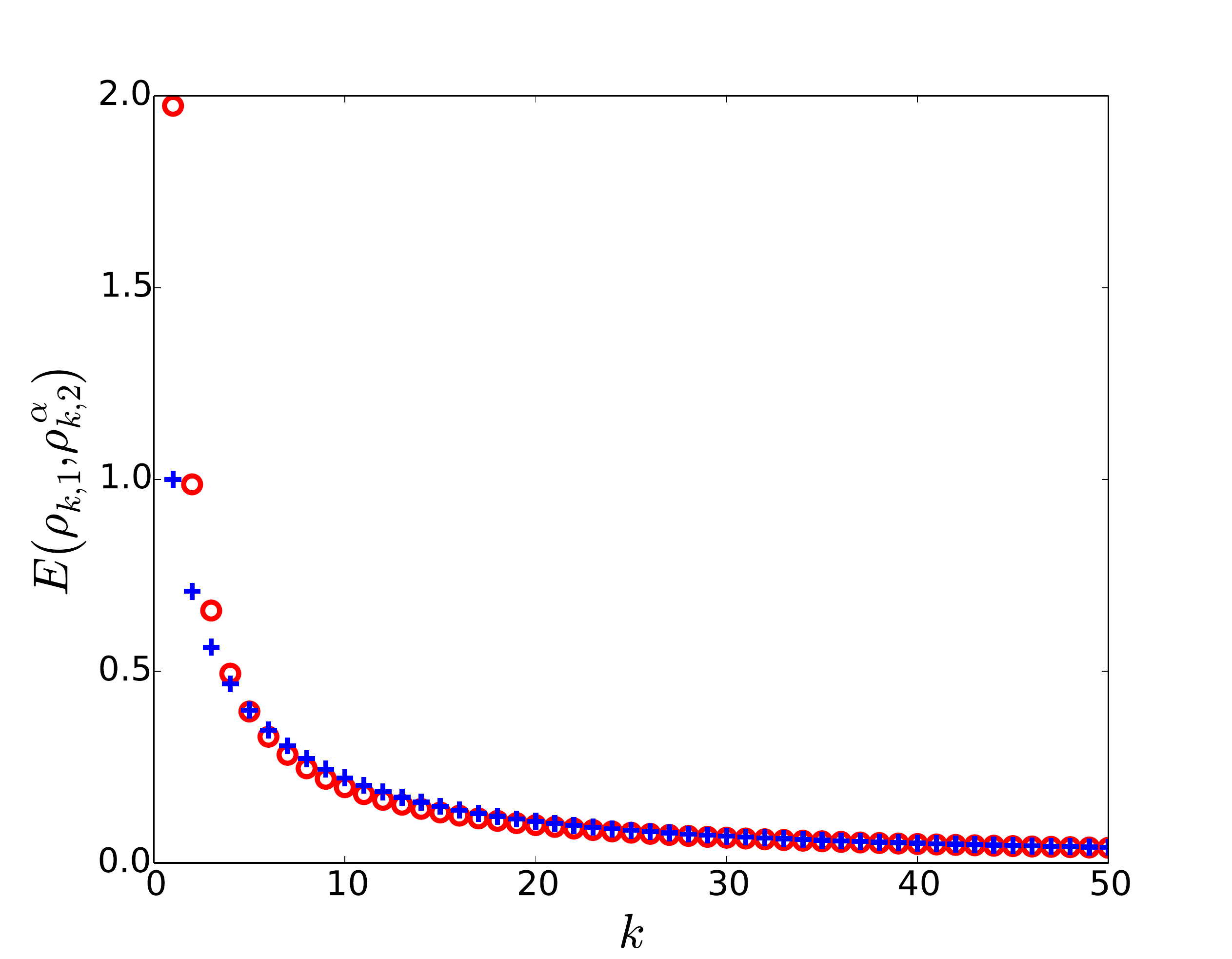} } 
\hspace{5mm}
\subfigure[The blues crosses correspond to the quantity $k E(\rho_{k,1},\rho_{k,2}^{\alpha})$, while the red line represents the constant $\frac{2}{\pi \sin \alpha} \left( 1 + \sqrt{2 - \frac{\sin \alpha}{\sqrt{1 + \sin^2 \alpha}}} \right) = \frac{2\sqrt{2}}{\pi} \left( 1 + \sqrt{2 - \frac{1}{\sqrt{3}}} \right)$.]{\includegraphics[scale=0.3]{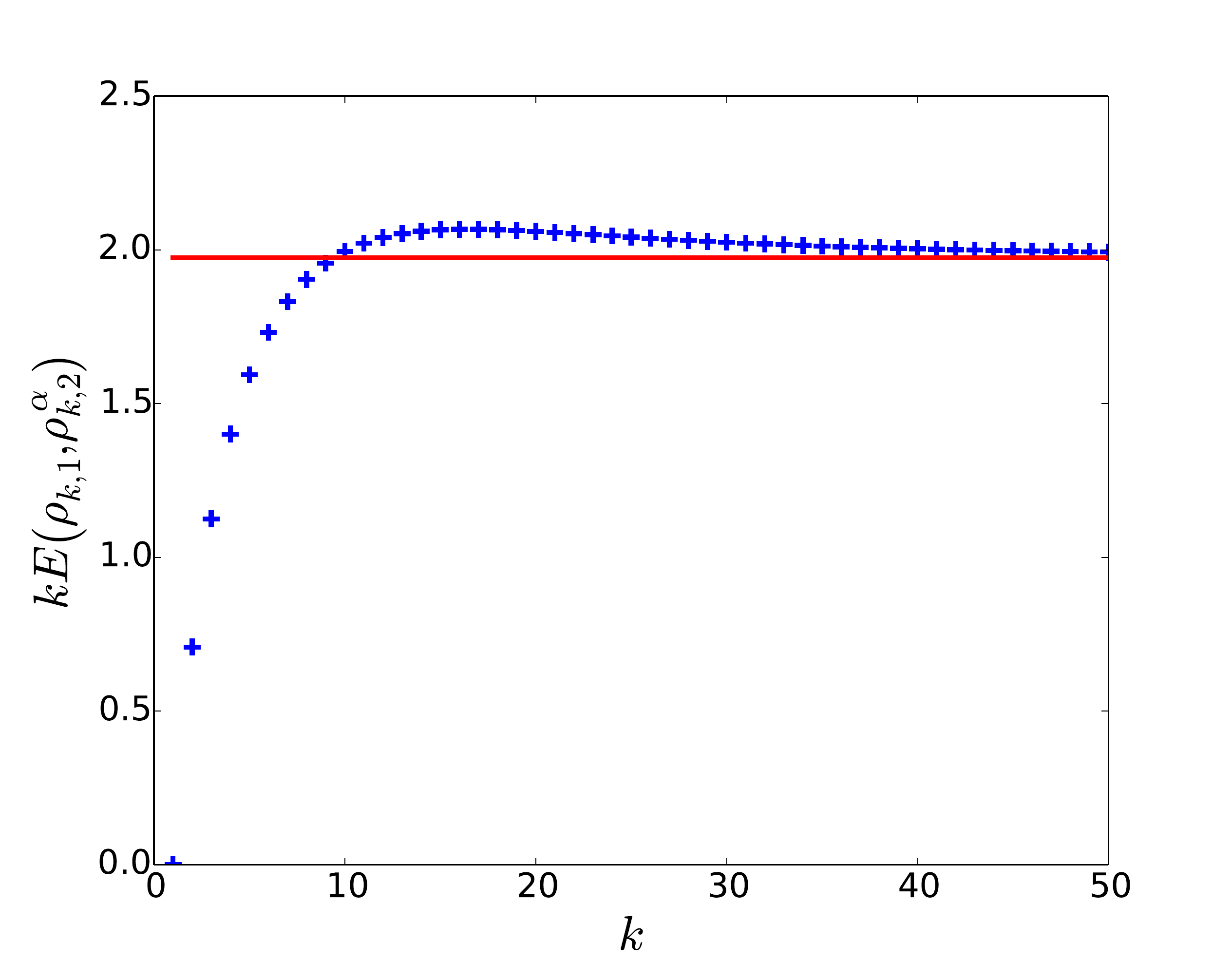} }
\caption{Sub-fidelity $E(\rho_{k,1},\rho_{k,2}^{\alpha})$ and $k E(\rho_{k,1},\rho_{k,2}^{\alpha})$, as functions of $k$, for $1 \leq k \leq 50$ and $\alpha = \frac{\pi}{4}$. }
\label{fig:sub_fidelity_pisur4}
\end{figure}

\begin{figure}[H]
\begin{center}
\includegraphics[scale=0.33]{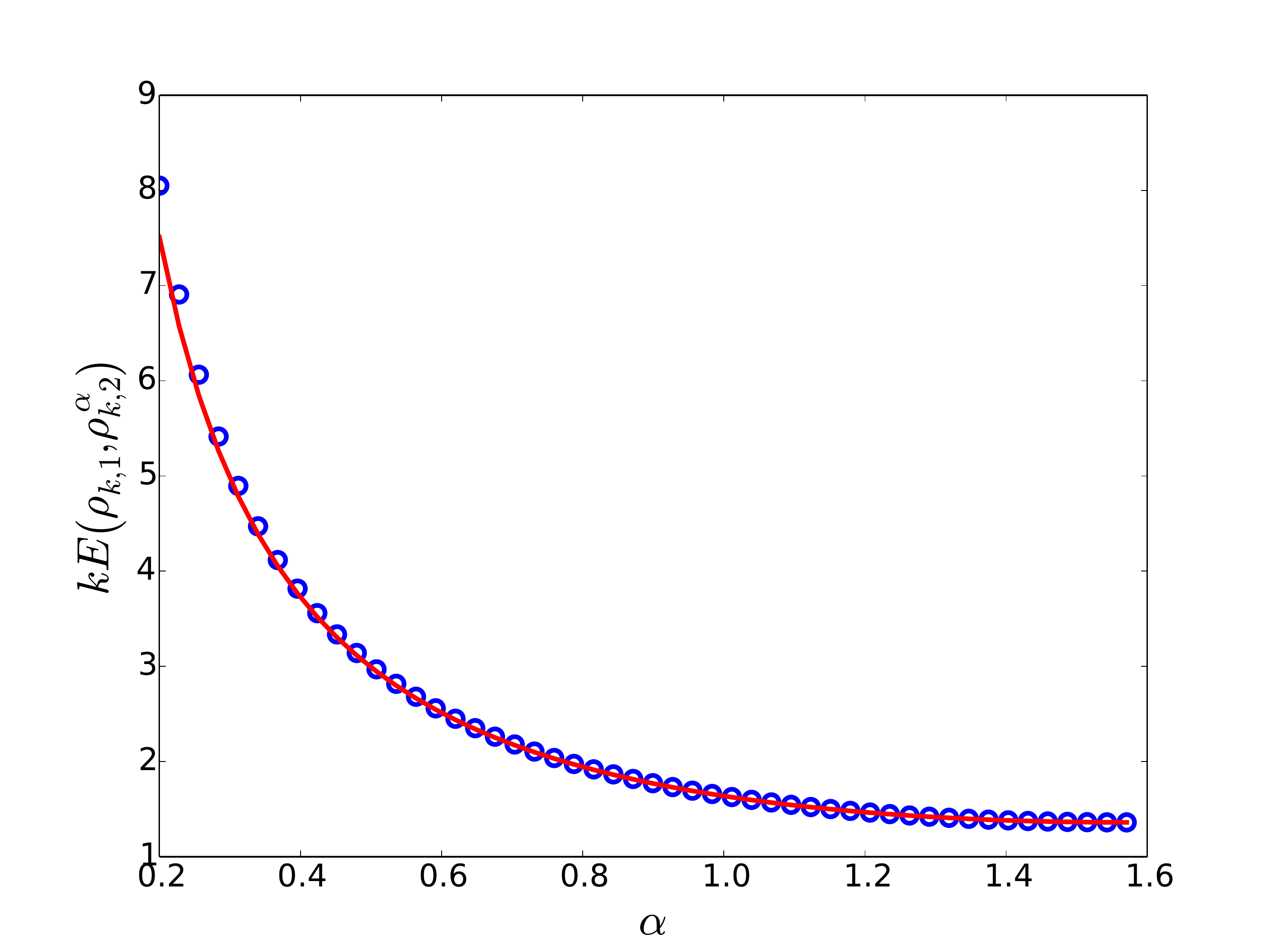}
\end{center}
\caption{The blue circles represent the value of $k E\left(\rho_{k,1}, \rho_{k,2}^{\alpha}\right) $ as a function of $\alpha$ for $k = 500$ and $0.2 \leq \alpha \leq \frac{\pi}{2}$. The red line corresponds to the theoretical equivalent $\alpha \mapsto \frac{2}{\pi \sin \alpha} \left( 1 + \sqrt{2 - \frac{\sin \alpha}{\sqrt{1 + \sin^2 \alpha}}} \right)$ obtained in Equation (\ref{eq:subfid_pi4}).}
\label{fig:subfid_angle_varying}
\end{figure}

\begin{figure}[H]
\includegraphics[scale=0.24]{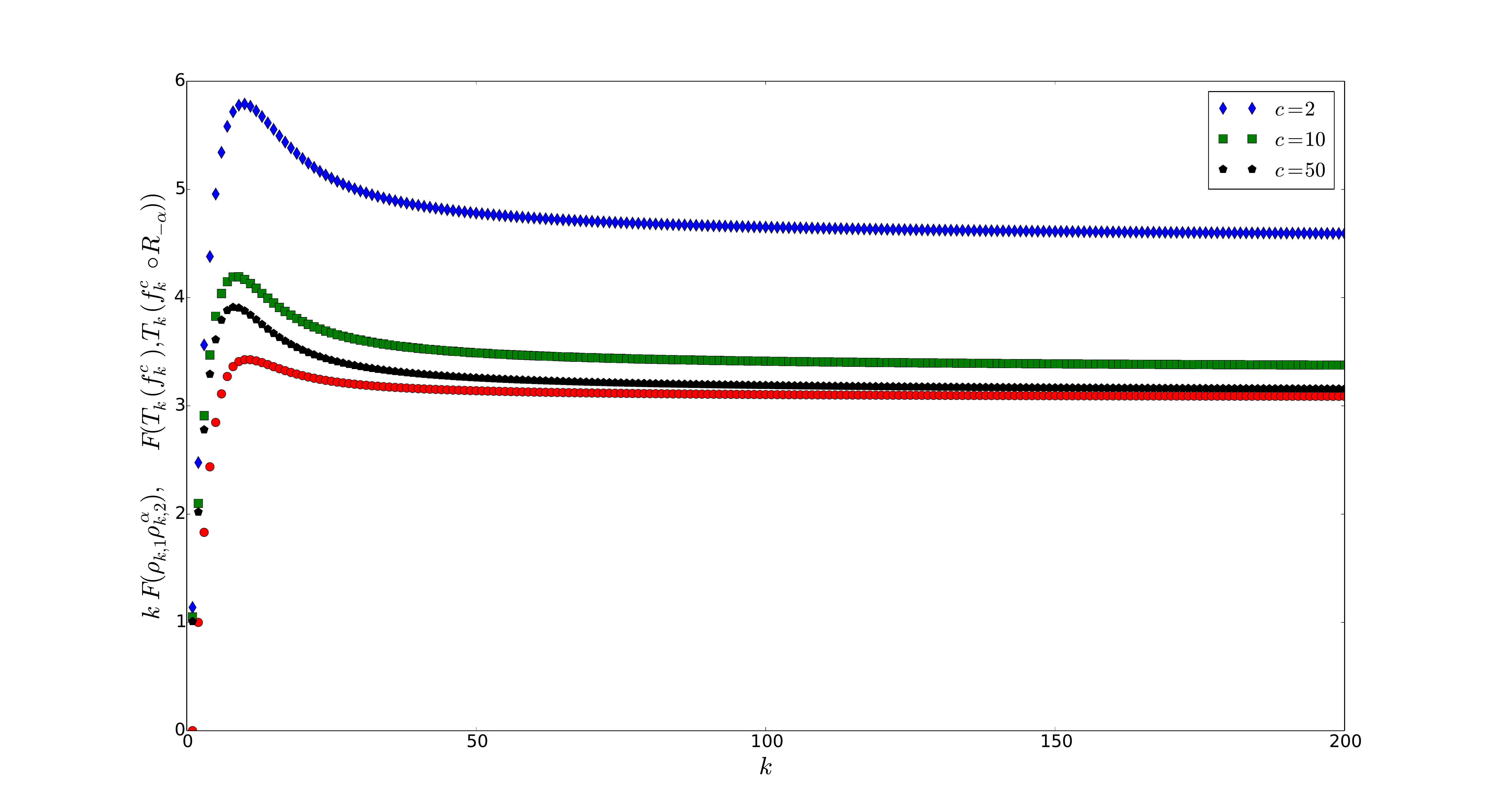}
\caption{Comparison between the rescaled fidelity $k F(\rho_{k,1}, \rho_{k,2}^{\alpha})$ (red circles) and $F(T_k(f_k^c), T_k(f_k^c \circ R_{-\alpha}))$ for $c = 2$ (blue diamonds), $c = 10$ (green squares) and $c = 50$ (black pentagons); here $\alpha = \frac{\pi}{2}$ and $1 \leq k \leq 200$.}
\label{fig:comp_fid_BTO_pi2}
\end{figure}

\begin{figure}[H]
\includegraphics[scale=0.24]{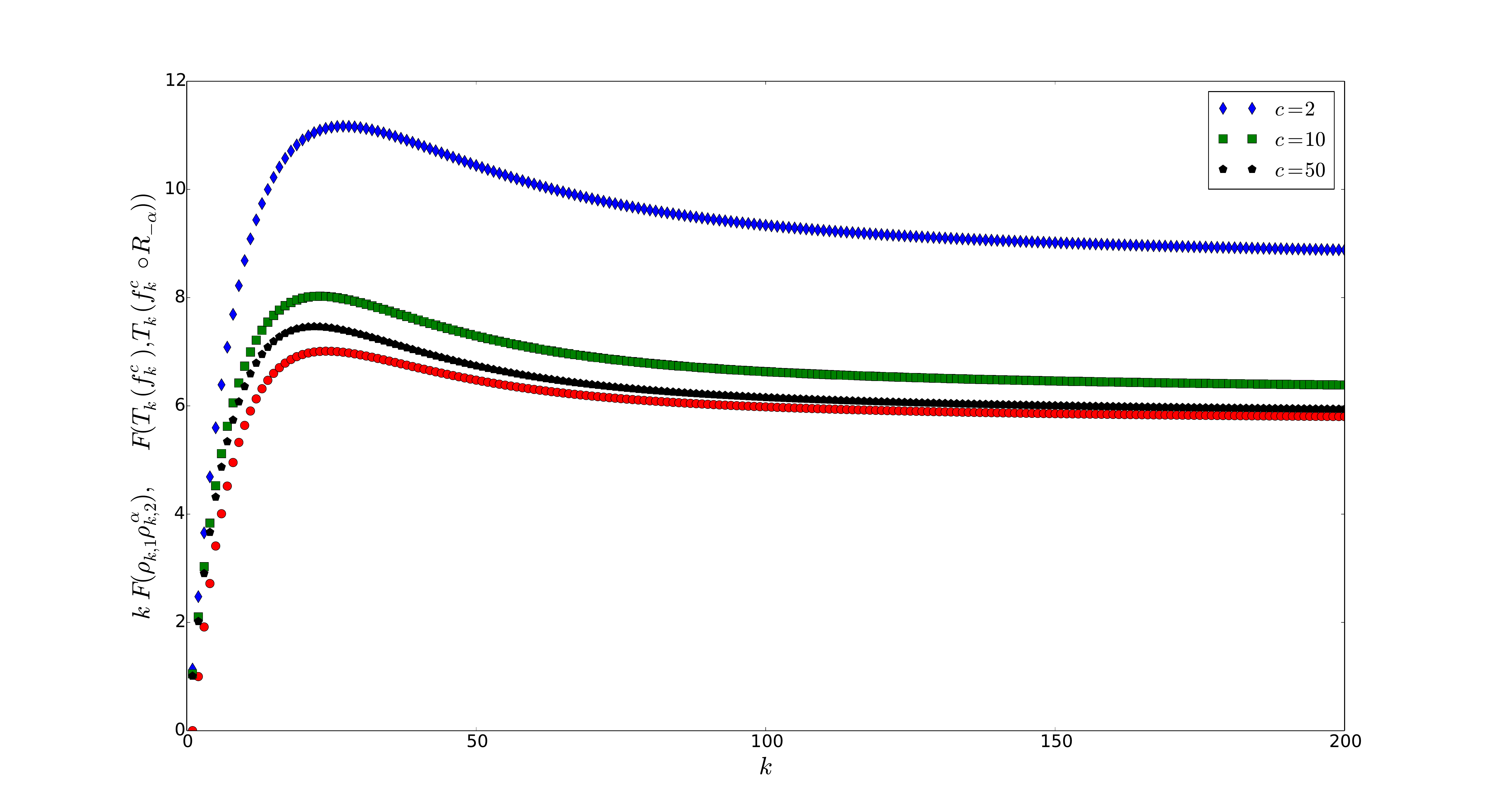}
\caption{Comparison between the rescaled fidelity $k F(\rho_{k,1}, \rho_{k,2}^{\alpha})$ (red circles) and $F(T_k(f_k^c), T_k(f_k^c \circ R_{-\alpha}))$ for $c = 2$ (blue diamonds), $c = 10$ (green squares) and $c = 50$ (black pentagons); here $\alpha = \frac{\pi}{4}$ and $1 \leq k \leq 200$.}
\label{fig:comp_fid_BTO_pi4}
\end{figure}

\begin{figure}[H]
\begin{center}
\includegraphics[scale=0.35]{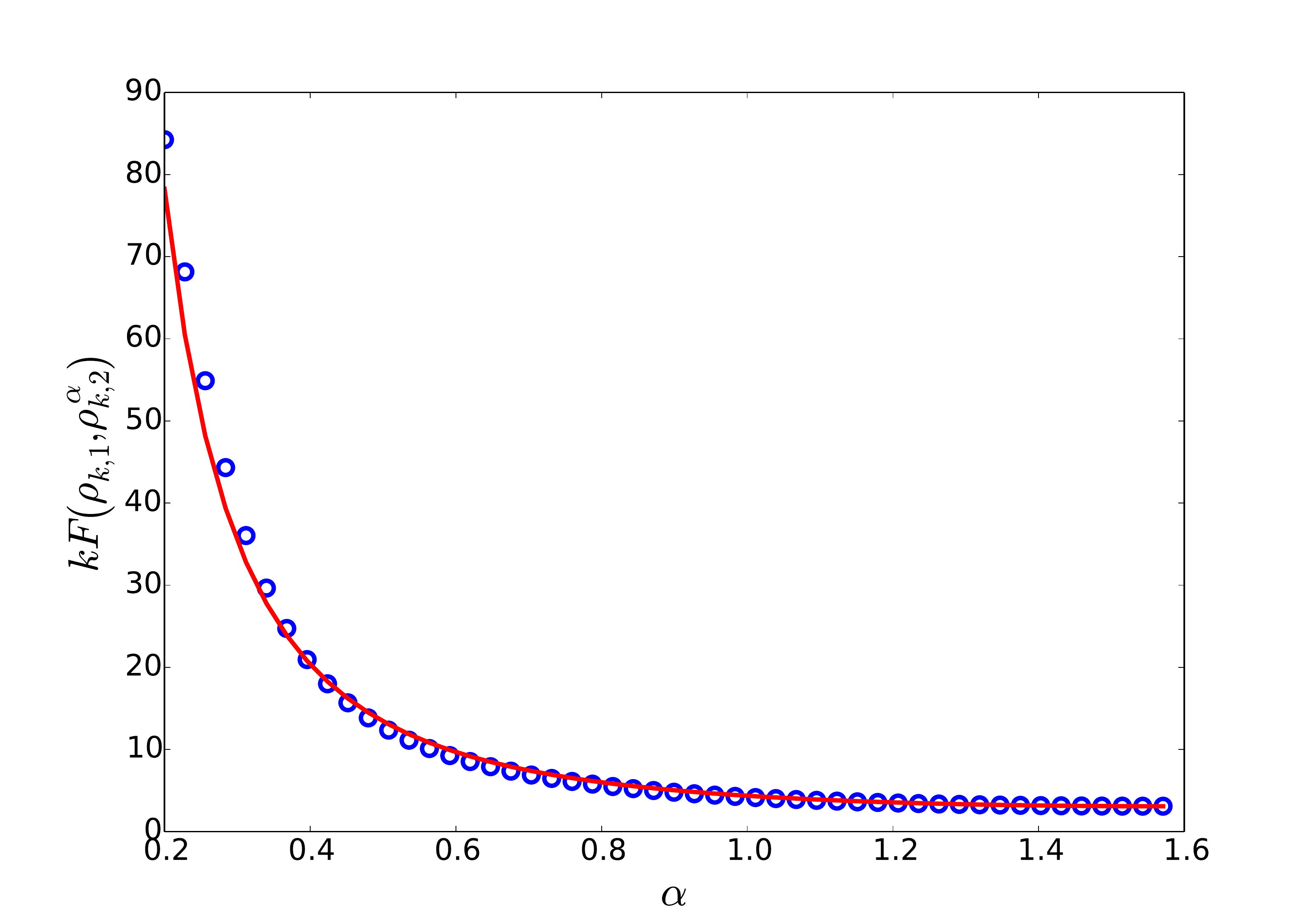}
\end{center}
\caption{The blue circles represent the value of $k F\left(\rho_{k,1}, \rho_{k,2}^{\alpha}\right) $ as a function of $\alpha$ for $k = 200$ and $0.2 \leq \alpha \leq \frac{\pi}{2}$. The red line corresponds to the conjectural equivalent $\alpha \mapsto \frac{C}{\sin^2 \alpha}$, where $C$ has been determined numerically from the case $\alpha = \frac{\pi}{2}$.}
\label{fig:fid_angle_varying}
\end{figure}

\begin{figure}[H]
\hspace{-5mm}
\subfigure[$E(\rho_{k,1},\rho_{k,2}^{\alpha})$ and $F(\rho_{k,1},\rho_{k,2}^{\alpha})$.]{\includegraphics[scale=0.3]{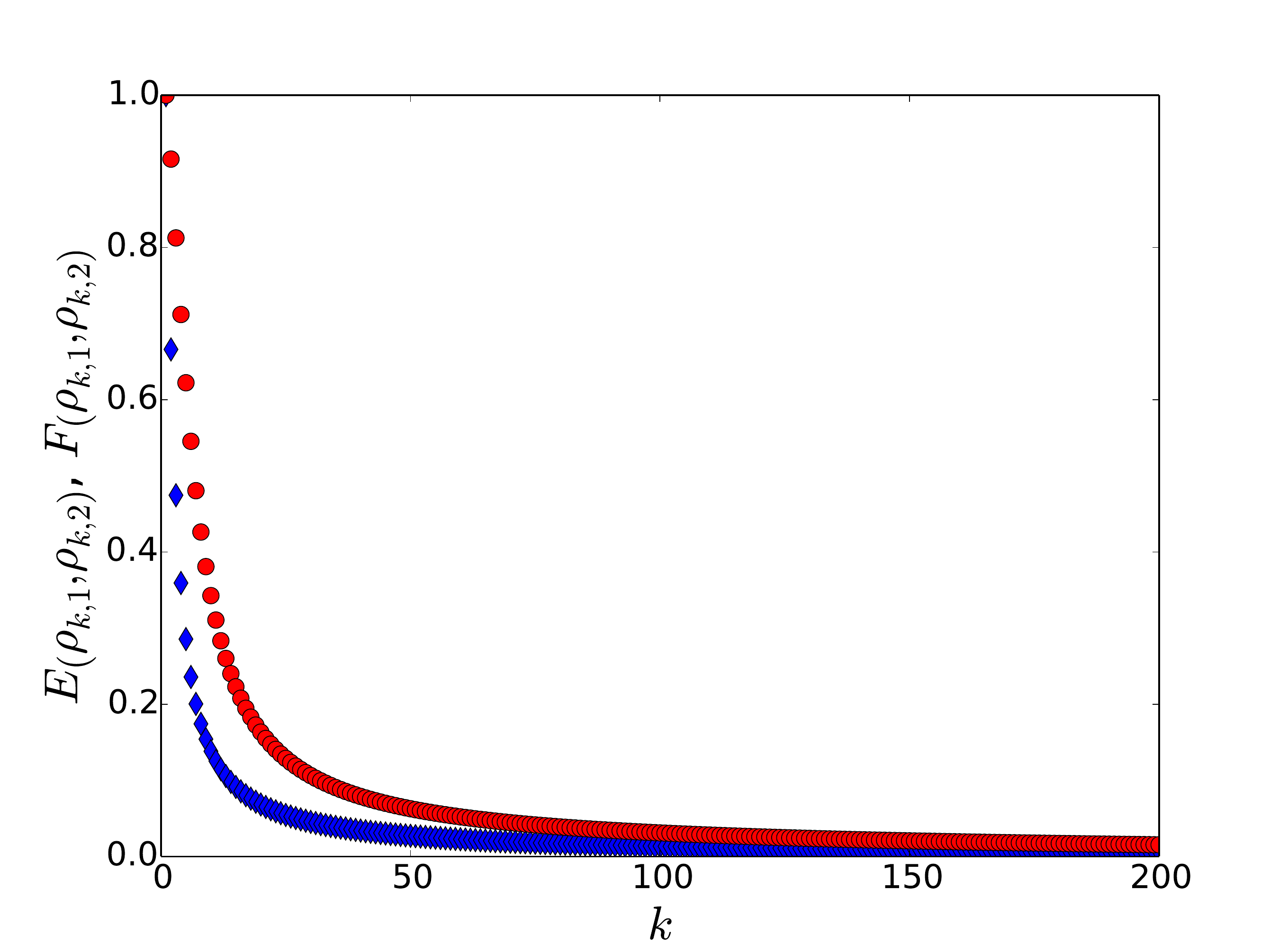} } 
\subfigure[$k E(\rho_{k,1},\rho_{k,2}^{\alpha})$ and $ k F(\rho_{k,1},\rho_{k,2}^{\alpha})$.]{\includegraphics[scale=0.3]{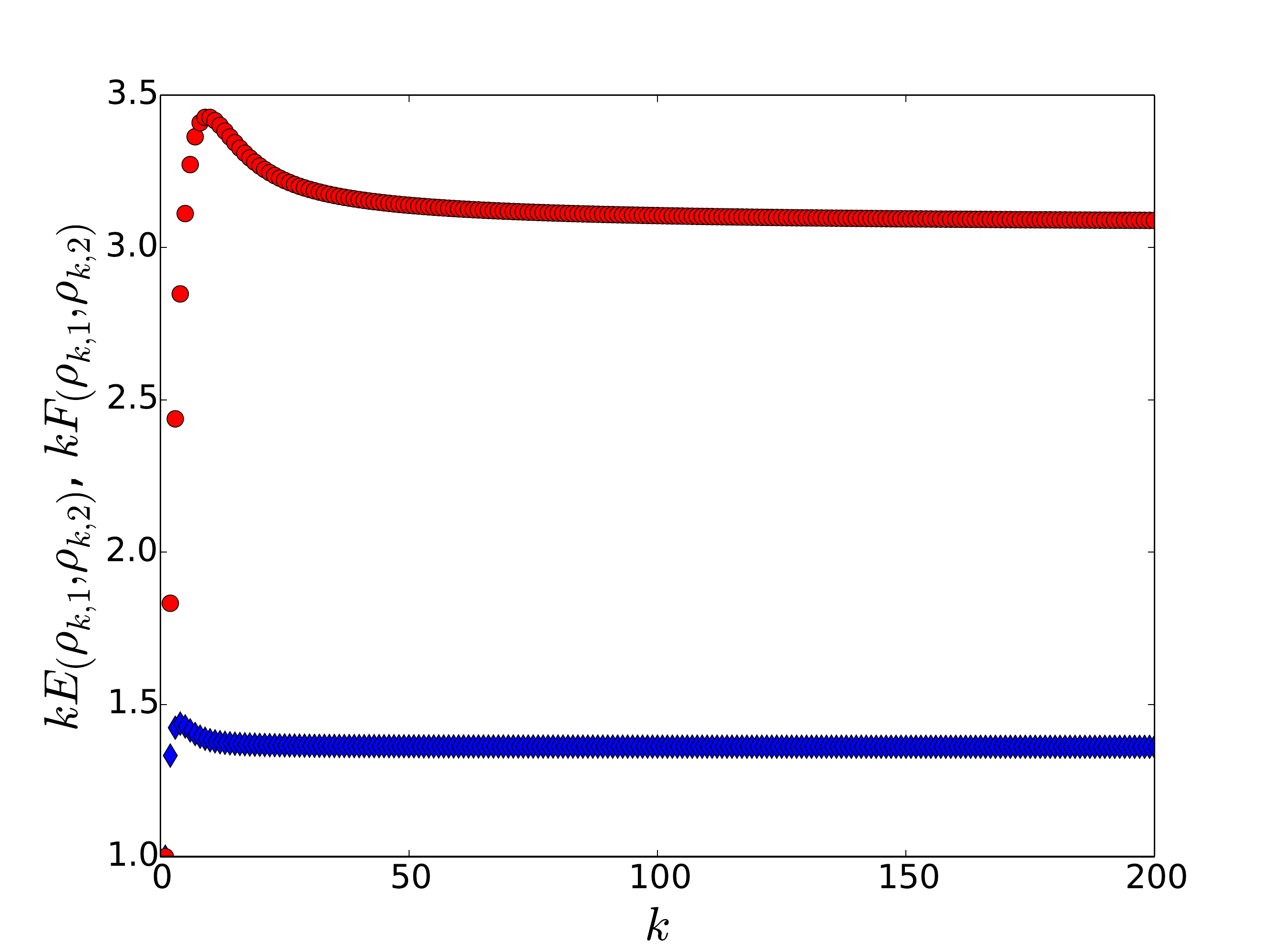} }
\caption{Comparison between the fidelity and sub-fidelity of $\rho_{k,1}$ and $\rho_{k,2}$, and their rescaled versions, as functions of $k$, $1 \leq k \leq 200$. The blue diamonds correspond to sub-fidelity, while the red circles represent fidelity.}
\label{fig:fid_vs_subfid_pi2}
\end{figure}

\begin{figure}[H]
\hspace{-5mm}
\subfigure[$E(\rho_{k,1},\rho_{k,2}^{\alpha})$ and $F(\rho_{k,1},\rho_{k,2}^{\alpha})$.]{\includegraphics[scale=0.3]{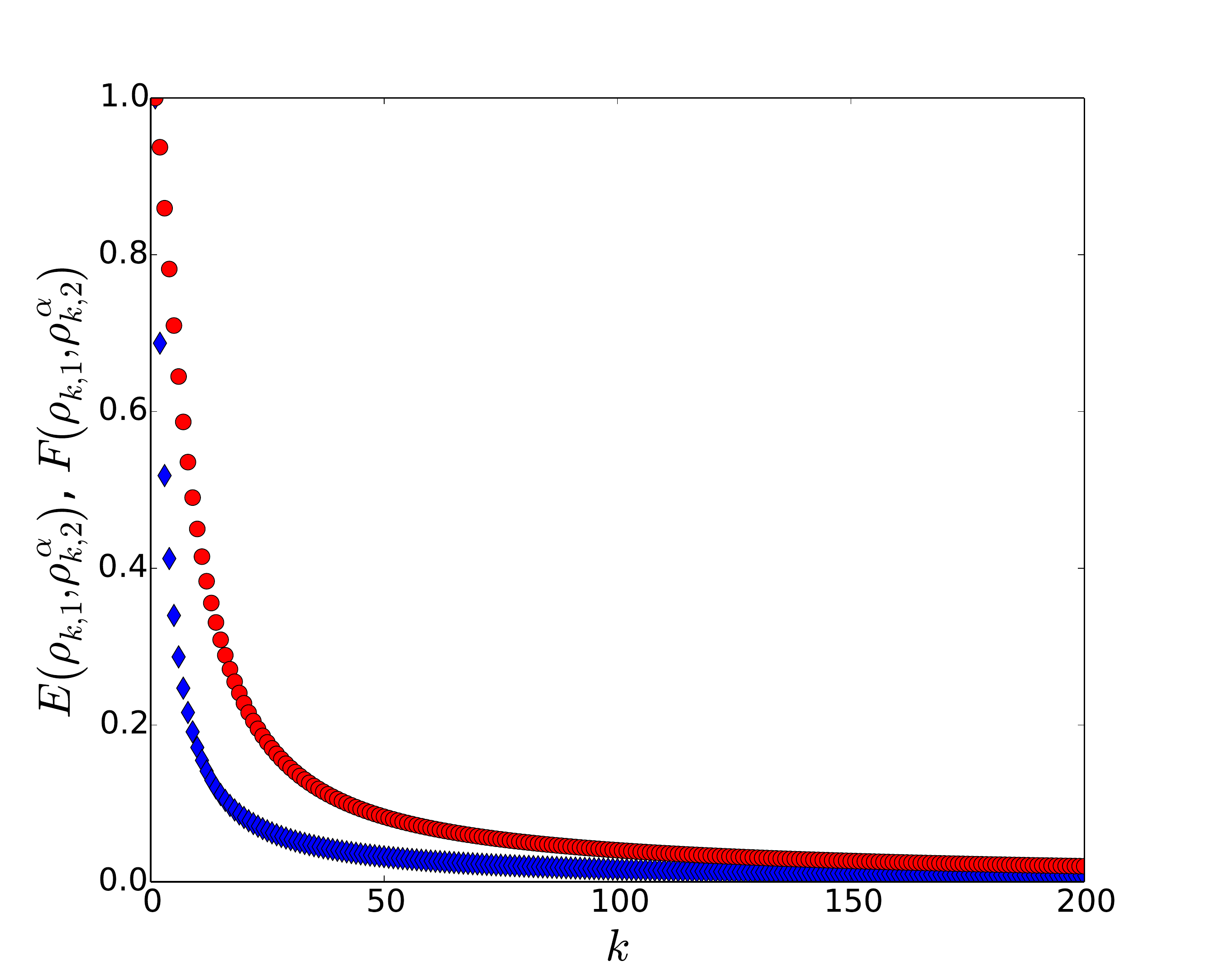} } 
\subfigure[$k E(\rho_{k,1},\rho_{k,2}^{\alpha})$ and $ kF(\rho_{k,1},\rho_{k,2}^{\alpha})$.]{\includegraphics[scale=0.3]{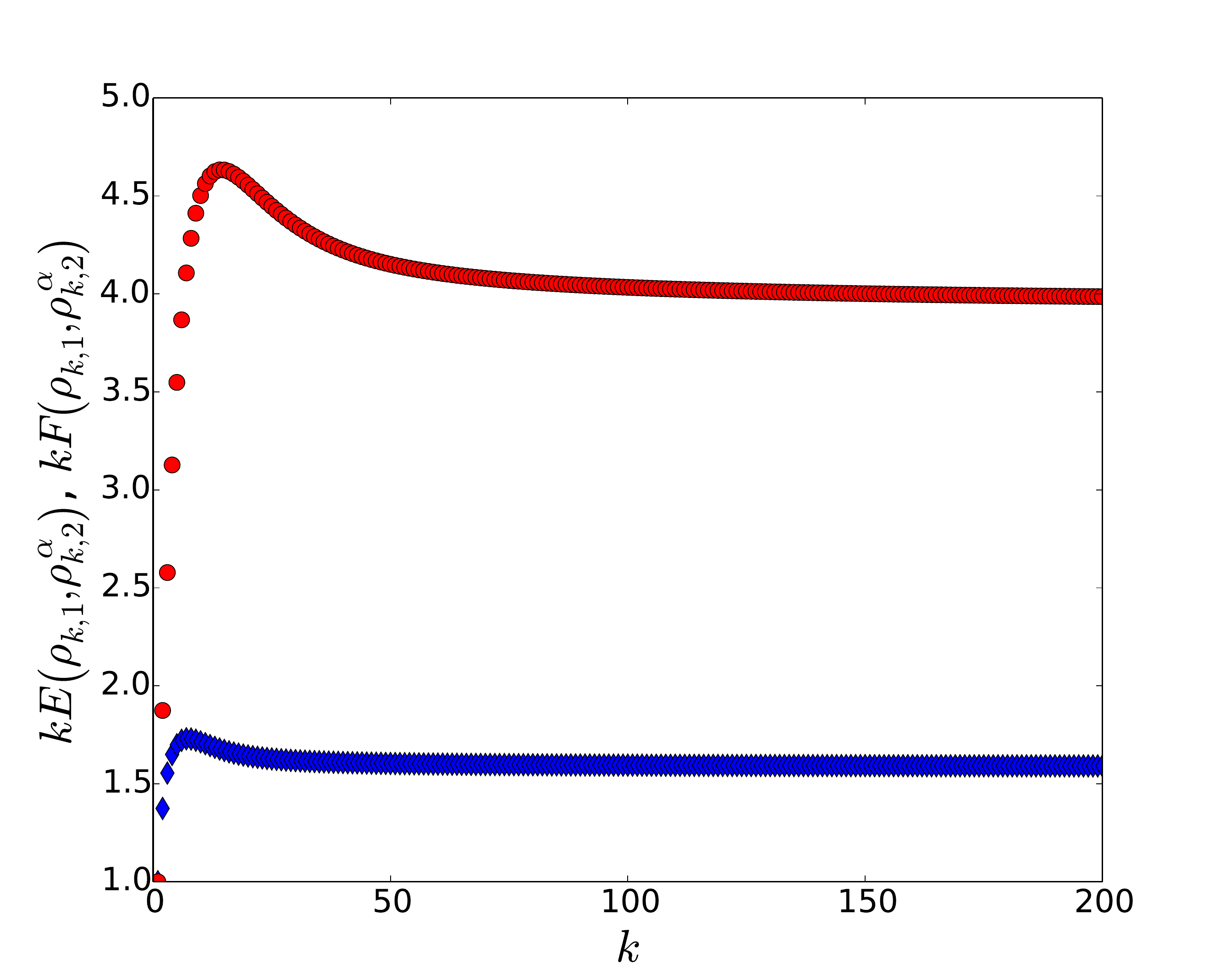} }
\caption{Comparison between fidelity and sub-fidelity of $\rho_{k,1}$ and $\rho_{k,2}^{\alpha}$ for $\alpha =  \frac{\pi}{3}$, as functions of $k$, $1 \leq k \leq 200$. The blue diamonds correspond to sub-fidelity, while the red circles represent fidelity.}
\label{fig:fid_vs_subfid_pi3}
\end{figure}

\subsection{Comparison between fidelity and sub-fidelity}

In view of the previous results, we expect the fidelity to be of the same order as the sub-fidelity, namely $\bigO{k^{-1}}$, but there is no reason that their equivalents are the same. In fact, we already know how the constants compare since $F \geq E$. These considerations lead us to the following conjecture.

\begin{conj}
Let $(\Gamma_1,\sigma_1)$ and $(\Gamma_1,\sigma_1)$ be two closed connected Lagrangian submanifolds with probability densities of a closed quantizable K\"ahler manifold $M$, intersecting transversally at a finite number of points. Let $C((\Gamma_1,\sigma_1),(\Gamma_2,\sigma_2))$ be as in Theorem \ref{thm:subfid}. Then there exists some constant $\widetilde{C}((\Gamma_1,\sigma_1),(\Gamma_2,\sigma_2)) \geq C((\Gamma_1,\sigma_1),(\Gamma_2,\sigma_2))$
such that 
\[ F(\rho_{k,1},\rho_{k,2}) = \left(\frac{2\pi}{k}\right)^{n}  \widetilde{C}((\Gamma_1,\sigma_1),(\Gamma_2,\sigma_2)) + \bigO{k^{-(n+1)}}. \]
\end{conj}

This would mean that the fidelity is of the same order of magnitude as the sub-fidelity. Besides evidence given by this example, this conjecture seems reasonable for the two following reasons. The first one is that the states that we consider are far from pure states, hence their super-fidelity is a very bad upper bound for their fidelity and we expect the latter to be much closer to the sub-fidelity. The second one is that when $\psi_k, \phi_k$ are pure states, i.e. elements in $\Hil_k$ of unit norm, then their fidelity is given by $F(\phi_k,\psi_k) =  |\scal{\phi_k}{\psi_k}|^2$.
But it is known (see \cite{BorPauUri} but also \cite[Theorem 6.1]{ChaPoly} for instance) that the scalar product of two pure states associated with Bohr-Sommerfeld Lagrangians has the following equivalent when $k$ goes to infinity:
\[ \scal{\phi_k}{\psi_k} \sim \left(\frac{2\pi}{k}\right)^{\frac{n}{2}} C(\Gamma_1,\Gamma_2). \]
Therefore our conjecture could be seen as some kind of generalization of this result, in a different context. 

\paragraph{Acknowledgements.}

Part of this work was supported by the European Research Council Advanced Grant 338809. We thank Leonid Polterovich for proposing the topic and for numerous useful discussions. We also thank Laurent Charles and Alejandro Uribe for their interest in this work and some helpful remarks. Finally, we thank St\'ephane Nonnenmacher for suggesting the approach that ultimately led to Theorem \ref{thm:fid_sphere}. We thank an anonymous referee for very useful advice regarding the exposition.

\appendix

\section{Some proofs of folklore results}

\begin{proof}[Proof of Lemma \ref{lm:equiv}]
Observe first that for every unit vector $u \in L$ and for every $g \in SU(2)$, the equality $\xi_k^{g u} = \zeta_k(g) \xi_k^u$ holds (and consequently $\| \xi_k^{g u} \|_k = \| \xi_k^{u} \|_k$). Indeed, by the properties stated in \cite[Section 5]{ChaBTO}, we have that for every $x \in \C\P^1$,
\[ \xi_k^{g u}(x) = \Pi_k(x,\pi(gu)) \cdot (gu)^k = \Pi_k(x,g\pi(u)) \cdot (gu)^k \]
where $\pi$ is the projection $L \to \C\P^1$. Since the kernel $\Pi_k$ is $SU(2)$-equivariant (this can be checked for instance in local coordinates thanks to Equation (\ref{eq:proj_sphere})), we finally obtain that
\[ \xi_k^{g u}(x) = g \left( \Pi_k(g^{-1}x,\pi(u)) \cdot u^k \right) = g \xi_k^u(g^{-1}x) = (\zeta_k(g) \xi_k^u)(x), \]
as announced. Therefore, for $\phi \in \Hil_k$,
\[ U_k(\alpha) \rho_{k,1} \phi = \int_{\Gamma_1} \frac{\scal{\phi}{\xi_k^u} \xi_k^{g_{\alpha} u}}{\| \xi_k^{u} \|^2_k} \sigma_1(y) \]
with $u$ any unit vector in $L_y$. Since $\sigma_2^{\alpha} = (R_{\alpha})_* \sigma_1$, this yields
\[ U_k(\alpha) \rho_{k,1} \phi = \int_{\Gamma_2^{\alpha}} \frac{\langle \phi, \xi_k^{g_{\alpha}^{-1}u} \rangle \xi_k^{u}}{\| \xi_k^{u} \|^2_k} \sigma_2^{\alpha}(x) = \int_{\Gamma_2^{\alpha}} \frac{\scal{\phi}{U_k(\alpha)^* \xi_k^{u}} \xi_k^{u}}{\| \xi_k^{u} \|^2_k} \sigma_2^{\alpha}(x). \]
Writing $\scal{\phi}{U_k(\alpha)^* \xi_k^{u}} = \scal{U_k(\alpha) \phi}{\xi_k^{u}}$, we finally obtain that $U_k(\alpha) \rho_{k,1} \phi = \rho_{k,2}^{\alpha} U_k(\alpha) \phi$.
\end{proof}

\begin{proof}[Proof of Lemma \ref{lm:BTO_radial}]
Let $\ell \in \llbracket 0,k \rrbracket$, and let $P_{\ell} = Z_1^{k-\ell} Z_2^{\ell}$ (so that $e_{\ell}$ is $P_{\ell}$ normalized). Recall that by considering the stereographic projection, we identify $P_{\ell}$ with $z \mapsto z^{\ell}$; hence Equation (\ref{eq:proj_sphere}) yields
\[ (T_k(f)P_{\ell})(z) = \frac{k+1}{2\pi} \int_{\C} \frac{(1 + z \bar{w})^k}{(1 + |w|^2)^{k+2}} g\left( \frac{|w|^2 - 1}{|w|^2 + 1} \right) w^{\ell} \ |dw \wedge d\bar{w}|. \]
By expanding the term $(1 + z \bar{w})^k$, we obtain that
\[ (T_k(f)P_{\ell})(z) = \frac{k+1}{2\pi} \sum_{m=0}^k \binom{k}{m} \left( \int_{\C} \frac{w^{\ell} \bar{w}^m}{(1 + |w|^2)^{k+2}} g\left( \frac{|w|^2 - 1}{|w|^2 + 1} \right) \ |dw \wedge d\bar{w}| \right) z^m. \]
By using polar coordinates $w = r \exp(i\theta)$, this yields 
\[ (T_k(f)P_{\ell})(z) = \frac{k+1}{\pi} \sum_{m=0}^k \binom{k}{m} I_{\ell,m} \left( \int_0^{+ \infty} \frac{r^{\ell + m +1} }{(1 + r^2)^{k+2}} g\left( \frac{r^2 - 1}{r^2 + 1} \right)  \ dr \right) z^m, \]
where $I_{\ell,m} = \int_{0}^{2\pi} \exp(i(\ell-m)\theta) \ d\theta$, which vanishes if $\ell \neq m$ and is equal to $2\pi$ otherwise. Hence the first part of the statement is proved, and  
\[ \scal{T_k(f) e_{\ell}}{e_{\ell}}_k = 2(k+1) \binom{k}{\ell} \int_0^{+ \infty} \frac{r^{2\ell +1} }{(1 + r^2)^{k+2}} g\left( \frac{r^2 - 1}{r^2 + 1} \right)  \ dr.\]
The change of variable $x = \frac{r^2 - 1}{r^2 + 1}$ then yields the desired formula.
\end{proof}

\begin{proof}[Proof of Proposition \ref{prop:exact_egorov}]
We proceed as in the proof of Lemma \ref{lm:equiv}. Let $\varphi \in \Hil_k$ and $x \in \C\P^1$; then
\[ (T_k(f) U_k(\beta)^* \varphi)(x) = \int_{\C\P^1} f(y) \ \Pi_k(x,y) \cdot \left( U_k(\beta)^* \varphi \right)(y) \ d\mu(y) \]
where $\mu$ is the Liouville measure associated with the Fubini-Study structure. Since $\left(U_k(\beta)^* \varphi \right)(y) = g_{-\beta} \varphi(g_{\beta} y)$ and 
\[ \Pi_k(x,y) \cdot  g_{-\beta} \varphi(g_{\beta} y) = g_{-\beta} \left( \Pi_k(g_{\beta}x,g_{\beta} y) \cdot \varphi(g_{\beta}y) \right), \]
the above expression reduces to
\[ (T_k(f) U_k(\beta)^* \varphi)(x) = \int_{\C\P^1} f(y) \ g_{-\beta} \left( \Pi_k(g_{\beta}x,g_{\beta} y) \cdot \varphi(g_{\beta}y)\right) \ d\mu(y). \]
This in turn yields, using the change of variables $z = g_{\beta} y$,
\[ (T_k(f) U_k(\beta)^* \varphi)(x) = \int_{\C\P^1} f(g_{-\beta}  z) \ g_{-\beta} \left( \Pi_k(g_{\beta}x, z) \cdot \varphi(z)\right) \ d\mu(z). \]
Here we have used that $\mu$ is $SU(2)$-invariant. This amounts to
\[ T_k(f) U_k(\beta)^* \varphi =  g_{-\beta} \left( \int_{\C\P^1} (f \circ g_{-\beta})(z) \ \Pi_k(\cdot, z) \cdot \varphi(z) \ d\mu(z) \right) = U_k(\beta)^* T_k(f \circ g_{-\beta}) \varphi . \]
Since $U_k(\beta)$ is unitary, this yields the desired result.
\end{proof}

\section{A stationary phase computation}

\begin{proof}[Proof of Proposition \ref{prop:trace_norm_BTO}]
Since $T_k\left(\sqrt{f_k^{c,\delta} g_k^{c,\delta}}\right) \geq 0$, its trace norm is equal to its trace, which satisfies
\[ \Tr\left( T_k\left(\sqrt{f_k^{c,\delta} g_k^{c,\delta}}\right)  \right) = \frac{k+1}{2\pi} \int_{\S^2} \sqrt{f_k^{c,\delta} g_k^{c,\delta}}  \ d\mu. \]
Since $R_{-\alpha}(x_1,x_2,x_3) = (x_1 \cos \alpha+  x_3  \sin \alpha, x_2, x_1 \sin \alpha + x_3 \cos \alpha)$,
this means that we need to evaluate the integral
\[ I = \int_{\S^2} \exp\left( - \frac{c}{2} (k+1)^{1-2\delta} \left( x_3^2 + (x_1 \sin \alpha + x_3 \cos \alpha)^2 \right) \right) \ d\mu(x_1,x_2,x_3). \]
By using the stereographic projection and polar coordinates, we obtain that
\[ I = 2 \int_0^{+\infty}  \int_0^{2\pi}  a_0(r) \exp(-\tau \phi(r,\theta)) \ dr d\theta, \]
where $\tau = c(k+1)^{1-2\delta}$, $a_0(r) = \frac{r}{(1+r^2)^2}$, and the phase $\phi$ satisfies 
\[  \phi(r,\theta) = \frac{1}{2} \left( \left(\frac{r^2 - 1}{r^2 + 1}\right)^2 +  \left(\frac{r^2 \cos \alpha + 2r \sin \alpha \cos \theta - \cos \alpha}{r^2+1}\right)^2  \right). \]
We will estimate $I$ by means of the stationary phase lemma. The phase $\phi$ is nonnegative and vanishes if and only if $(r,\theta) = (1, \frac{\pi}{2})$ or $(r,\theta) = (1,-\frac{\pi}{2})$ (these are the intersection points of the images of $\Gamma_1$ and $\Gamma_2^{\alpha}$ by $\pi_N$). Its derivative with respect to $r$ read
\[ \dpar{\phi}{r}(r,\theta) = 2\left(\frac{r (r^2 - 1) + (2 r \cos \alpha + (1-r^2) \sin \alpha \cos \theta)((r^2-1)\cos \alpha + 2 r \sin \alpha \cos \theta)}{(r^2 + 1)^3}\right),  \]
while its derivative with respect to $\theta$ is given by the formula 
\[ \dpar{\phi}{\theta}(r,\theta) = \frac{ 2 r \sin \alpha \sin \theta ((1-r^2)\cos \alpha - 2 r \sin \alpha \cos \theta)}{(r^2 + 1)^2}. \]
Both vanish at the two points $(1, \frac{\pi}{2})$ and $(1,-\frac{\pi}{2})$. Moreover, one readily checks that the Hessian matrices $H_1$ and $H_2$ of $\phi$ at $(1, \frac{\pi}{2})$ and $(1,-\frac{\pi}{2})$ respectively read
\[ H_1 =  \begin{pmatrix} 1 + \cos^2 \alpha & - \cos \alpha \sin \alpha \\ - \cos \alpha \sin \alpha & \sin^2 \alpha  \end{pmatrix}, \quad H_2 = \begin{pmatrix} 1 + \cos^2 \alpha &  \cos \alpha \sin \alpha \\  \cos \alpha \sin \alpha & \sin^2 \alpha  \end{pmatrix}. \]
These matrices both have determinant $\sin^2 \alpha$; hence, the stationary phase lemma yields
\[ I = 4 \left( \frac{2 \pi}{\tau} \right) \frac{ a_0(1)}{\sin \alpha} + \bigO{\tau^{-2}} = \frac{2 \pi}{\tau \sin \alpha} + \bigO{\tau^{-2}}. \]
Consequently, we finally obtain that
\[ \Tr\left( T_k\left(\sqrt{f_k^{c,\delta} g_k^{c,\delta}}\right)  \right) = \frac{2 k^{2\delta}}{\sqrt{c \pi} \sin \alpha} + \bigO{k^{4\delta-1} c^{-\frac{3}{2}}}, \]
as announced.
\end{proof}

\bibliographystyle{abbrv}
\bibliography{fidelity}

\end{document}